\documentclass[nonacm, authorversion, sigconf, final]{acmart} 
\pdfoutput=1
\usepackage[T1]{fontenc}
\usepackage{graphicx}
\usepackage{tikz}
\usetikzlibrary{calc}
\usetikzlibrary{positioning}
\usetikzlibrary{arrows,automata}
\usetikzlibrary{patterns,patterns.meta}
\usetikzlibrary{shapes.geometric}
\usepackage{xspace}
\usepackage{xcolor}
\usepackage{amsmath}
\usepackage{thmtools}
\usepackage{thm-restate}
\usepackage{mathrsfs}
\usepackage{ifdraft}
\usepackage{wrapfig}
\usepackage{subcaption}
\usepackage{verbatim}
\usepackage{booktabs}
\usepackage{multirow}
\usepackage{paralist}
\usepackage{lipsum}
\usepackage{amsbsy}
\usepackage[normalem]{ulem}
\usepackage{textcomp}
\usepackage{stmaryrd}
\usepackage{mathpartir}
\usepackage{keyval}
\usepackage{algorithm}
\usepackage{algorithmicx}
\usepackage[noend]{algpseudocode}
\usepackage{makecell}
\usepackage{stackengine}

\RequirePackage{centernot}
\RequirePackage[Symbol]{upgreek}

\RequirePackage{etoolbox}
\RequirePackage{mathtools}

\newcommand{\Kleenestar}{\mathrel{\vphantom{\to}^*}}
\let\originalleft\left
\let\originalright\right
\renewcommand{\left}{\mathopen{}\mathclose\bgroup\originalleft}
\renewcommand{\right}{\aftergroup\egroup\originalright}

\newcommand{\pvec}[1]{\vec{#1}\mkern2mu\vphantom{#1}}

\newcommand{\Nat}{\mathbb{N}}

\newcommand{\protocol}{\ensuremath{\mathcal{G}}}

\newcommand{\Procs}{\ensuremath{\mathcal{P}}}
\newcommand{\ProcsOf}[1]{\ensuremath{\mathcal{P}_{\! #1}}}\definecolor{roleColor}{rgb}{0.1, 0.3, 0.1}\newcommand{\roleCol}[1]{{\color{roleColor}#1}}\newcommand{\roleFmt}[1]{\roleCol{\mathtt{#1}}}

\newcommand{\procA}{{\roleFmt{p}}}
\newcommand{\procB}{{\roleFmt{q}}}
\newcommand{\procC}{{\color{roleColor}\roleFmt{r}}}
\newcommand{\procD}{{\color{roleColor}\roleFmt{s}}}
\newcommand{\procE}{{\color{roleColor}\roleFmt{t}}}

\newcommand{\val}{\ensuremath{m}}
\newcommand{\MsgVals}{\ensuremath{\mathcal{V}}}

\newcommand{\CSMabb}[1]{\ensuremath{\mathcal{#1}}}
\newcommand{\CSM}[1]{\ensuremath{\{\!\!\{#1_\procA\}\!\!\}_{\procA \in \Procs}}}

\newcommand{\makeAsync}[1]{\hat{#1}}

\newcommand{\CLTS}[1]{\ensuremath{\{\!\!\{#1_\procA\}\!\!\}_{\procA \in \Procs}}}

\newcommand{\TM}{\mathit{TM}\xspace} 
\newcommand{\condX}{(X)\xspace}
\newcommand{\TapeAlph}{\Delta}
\newcommand{\TMStates}{Q}

\newcommand{\msgBackForth}[3]{#1\leftrightarrow#2:#3}
\newcommand{\markEndLoop}{\bot}
\newcommand{\markNewRound}{\circ}
\newcommand{\markBeginConf}{{\langle\negthinspace\langle}}
\newcommand{\markEndConf}{{\rangle\negthinspace\rangle}}
\newcommand{\invisibleEndLine}{\phantom{s}}
\newcommand{\msc}{\operatorname{msc}}
\newcommand{\lbl}{\mathsf}
\newcommand{\sinkfinal}{sink-final\xspace}

\newcommand{\emptystring}{\varepsilon}

\newcommand{\set}[1]{\{#1\}}

\newcommand{\lang}{\mathcal{L}}
\newcommand{\langasync}{\mathcal{L}_{\mathit{async}}}
\newcommand{\langsync}{\mathcal{L}_{\mathit{sync}}}

\newcommand{\SyncToAsync}{\ensuremath{\texttt{\upshape split}}}
\newcommand{\channels}{\ensuremath{\mathsf{Chan}}}

\newcommand{\channel}[2]{\ensuremath{#1,#2}}

\newcommand{\AlphSync}{\ensuremath{\Gamma}} \newcommand{\AlphAsync}{\ensuremath{\hat{\Gamma}}}\newcommand{\AlphAsyncSubscript}{\ensuremath{Σ_{\mathit{async}}}}
\newcommand{\Alphabet}{\AlphAsync}

\newcommand{\restrict}[2]{{#1}|_{#2}}
\newcommand{\lfp}{\mathrm{lfp}}

\newcommand{\interswap}{\ensuremath{\sim}}
\newcommand{\interswapProcs}{\ensuremath{\sim}}
\newcommand{\interswapConc}{\equiv}

\stackMath
\DeclareMathOperator*{\MixCh}{\Sum}

\def \ifempty#1{\def\temp{#1} \ifx\temp\empty }

\newcommand{\snd}[3]{#1\triangleright#2\sendOp#3}
\newcommand{\rcv}[3]{#2\triangleleft#1\recOp#3}

\newcommand{\msgFromTo}[3]{#1\!\to\!#2\!:\!#3}
\newcommand{\msgFromToNS}[3]{#1\to#2:#3}

\newcommand{\alphOp}{\operatorname{\textit{alph}}}

\newcommand{\pref}{\operatorname{pref}}

\newcommand{\wproj}{{\ensuremath{\Downarrow}}}

\def\mmerge{\mathrel{\ThisStyle{\stretchrel*{\ooalign{\raise0.2\LMex\hbox{$\SavedStyle\sqcap$}\cr \raise-0.2\LMex\hbox{$\SavedStyle\sqcap$}}}{\sqcap}}}}
\def\mmmerge{\mathrel{\ThisStyle{\stretchrel*{\ooalign{\raise0.6\LMex\hbox{$\SavedStyle\sqcap$}\cr \raise0.2\LMex\hbox{$\SavedStyle\sqcap$}\cr \raise-0.2\LMex\hbox{$\SavedStyle\sqcap$}}}{\sqcap}}}}

\newcommand{\union}{\cup}
\newcommand{\inters}{\cap}
\newcommand{\Union}{\bigcup}

\newcommand{\dunion}{\uplus}

\DeclarePairedDelimiter\card{\lvert}{\rvert}

\providecommand{\implies}{\Rightarrow}

 \providecommand{\Coloneqq}{\mathrel{\mathop{::}}=} \newcommand{\is}{\coloneq}

\def\congr{\equiv}
\def\precongr{\sqsubseteq}

\newcommand{\reach}{\operatorname{reach}}

\newcommand{\from}{\colon}

 \def\st{.\,}  
\newcommand{\ie}{i.e.~}
\newcommand{\eg}{e.g.~}

\def\grammOr{\hspace{3pt}\mid\hspace{3pt}}
\def\grammIs{\Coloneqq}

\makeatletter
\begingroup
\catcode`\|=\active \gdef\@grammar@bar{\catcode`\|=\active \def|{\grammOr}}
\endgroup

\newcommand{\gramm}[1]{\begingroup
	\def\is{\grammIs}\@grammar@bar #1\endgroup }

\newenvironment{grammar}{\begin{equation*}\def\is{& \grammIs }\@grammar@bar \aligned }
	{\endaligned \end{equation*}\aftergroup\ignorespaces }
\makeatother

\newcommand{\hole}{\hbox{-}}

\usetikzlibrary{arrows.meta}

\tikzset{
  trans/.style={
    draw,-{Stealth[round]}, semithick, shorten >= 1pt,
  },
  init/.style={initial by arrow},
  final/.style={accepting},
  initial text={},
  initial distance=2ex,
  every initial by arrow/.style={trans},
}

\tikzset{
  sem/.style={
every state/.style = {semnode},
	 every edge/.style = {semarrow},
	 every loop/.style = {semarrow},
  }
}

\tikzset{
  semnode/.style={
    thick,
    draw,
    minimum size=1ex,
    shape=circle,
    font=\scriptsize,
    inner sep=2pt,
  },
  semarrow/.style={
    trans,
    font=\scriptsize,
    draw,
    pos=.4,
}
}

\newcommand{\adjustfigure}[1][\small]{\centering#1\columnwidth=\linewidth \belowdisplayskip=0pt\belowdisplayshortskip=0pt\abovedisplayskip=0pt\abovedisplayshortskip=0pt}

\DeclareMathOperator{\sendOp}{!}
\DeclareMathOperator{\recOp}{?}

\newcommand{\seq}{\,.\,}
\newcommand{\cat}{\cdot}

\newcommand{\starred}{\mathrel{\vphantom{\to}^*}}

\newcommand{\pn}[1]{\mathsf{#1}}

\newcommand{\restr}{\upnu}  \newcommand{\prefixPi}{\pi}  \newcommand{\zero}{\mathbf{0}}
\newcommand{\labelAndType}[2]{#1(#2)} \newcommand{\labelAndMsg}[2]{#1\langle#2\rangle}
\newcommand{\labelAndVar}[2]{#1(#2)}

\newcommand{\meta}[1]{{\color{gray}#1}}
\newcommand{\types}{\;{\color{blue}\vdash}\;}
\newcommand{\typesSFd}{\;{\color{blue}\Vdash_{\tiny{SF}}}\;}
\newcommand{\typesSFs}{\;{\color{blue}\vdash_{\tiny{SF}}}\;}
\newcommand{\redContext}{\mathbb{C}}
\newcommand{\hasType}{:}
\newcommand{\typingContextOne}{\Theta}
\newcommand{\typingContextTwo}{\Lambda}
\newcommand{\typingContextThree}{\Omega}
\newcommand{\typingContextCat}{\mathrel{\color{red!60!black}\shortmid}}
\newcommand{\EndState}{\operatorname{end}}
\newcommand{\SessionName}{\mathcal{S}}

\newcommand{\Parallel}{\prod}

\newcommand{\Sum}{\Sigma}

\newcommand{\redto}{\longrightarrow}

\newcommand{\Defs}{\mathcal{D}}
\newcommand{\ProcIds}{\mathcal{Q}}

\newcommand{\procReductionProcName}{\textsc{PR-$\pn{Q}$}\xspace}
\newcommand{\procReductionContext}{\textsc{PR-ctx}\xspace}
\newcommand{\procReductionExc}{\textsc{PR-exc}\xspace}
\newcommand{\procReductionCongr}{\textsc{PR-$\precongr$}\xspace}

\newcommand{\procTypingProcDefEmpty}{\textsc{PT-def-$\emptystring$}\xspace}
\newcommand{\procTypingProcDef}{\textsc{PT-def}\xspace}
\newcommand{\procTypingProcName}{\textsc{PT-$\pn{Q}$}\xspace}
\newcommand{\procTypingZero}{\textsc{PT-$\zero$}\xspace}
\newcommand{\procTypingEnd}{\textsc{PT-end}\xspace}

\newcommand{\procTypingMixCh}{\textsc{PT-$\MixCh$}\xspace}
\newcommand{\procTypingParallel}{\textsc{PT-$\parallel$}\xspace}
\newcommand{\procTypingRestr}{\textsc{PT-$\restr$}\xspace}

\newcommand{\typingReductionMixCh}{\textsc{TR-$\MixCh$}\xspace}

\newcommand{\Elang}[2]{[#1]^\exists_{#2}}
\newcommand{\Alang}[2]{[#1]^\forall_{#2}}

\newcommand{\EQUIVALENCE}{\ensuremath{\mathit{EQUIVALENCE}}}
\newcommand{\INCLUSION}{\ensuremath{\mathit{INCLUSION}}}

\newcommand{\PSPACE}{\mbox{PSPACE}\xspace}
\newcommand{\EXPSPACE}{\mbox{EXPSPACE}\xspace}
\newcommand{\EXPTIME}{\mbox{EXPTIME}\xspace} \usepackage{newunicodechar}
\newunicodechar{∃}{\ensuremath{\exists}}
\newunicodechar{∀}{\ensuremath{\forall}}
\newunicodechar{θ}{\ensuremath{\theta}}
\newunicodechar{τ}{\ensuremath{\tau}}
\newunicodechar{φ}{\ensuremath{\varphi}}
\newunicodechar{ξ}{\ensuremath{\xi}}
\newunicodechar{ζ}{\ensuremath{\zeta}}
\newunicodechar{ψ}{\ensuremath{\psi}}
\newunicodechar{π}{\ensuremath{\pi}}
\newunicodechar{α}{\ensuremath{\alpha}}
\newunicodechar{β}{\ensuremath{\beta}}
\newunicodechar{γ}{\ensuremath{\gamma}}
\newunicodechar{δ}{\ensuremath{\delta}}
\newunicodechar{ε}{\ensuremath{\varepsilon}}
\newunicodechar{κ}{\ensuremath{\kappa}}
\newunicodechar{λ}{\ensuremath{\lambda}}
\newunicodechar{μ}{\ensuremath{\mu}}
\newunicodechar{ρ}{\ensuremath{\rho}}
\newunicodechar{σ}{\ensuremath{\sigma}}
\newunicodechar{ω}{\ensuremath{\omega}}
\newunicodechar{Γ}{\ensuremath{\Gamma}}
\newunicodechar{Φ}{\ensuremath{\Phi}}
\newunicodechar{Δ}{\ensuremath{\Delta}}
\newunicodechar{Σ}{\ensuremath{\Sigma}}
\newunicodechar{Π}{\ensuremath{\Pi}}
\newunicodechar{∑}{\ensuremath{\Sigma}}
\newunicodechar{∏}{\ensuremath{\Pi}}
\newunicodechar{Θ}{\ensuremath{\Theta}}
\newunicodechar{Ω}{\ensuremath{\Omega}}
\newunicodechar{⇒}{\ensuremath{\Rightarrow}}
\newunicodechar{⇐}{\ensuremath{\Leftarrow}}
\newunicodechar{⇔}{\ensuremath{\Leftrightarrow}}
\newunicodechar{→}{\ensuremath{\rightarrow}}
\newunicodechar{←}{\ensuremath{\leftarrow}}
\newunicodechar{↔}{\ensuremath{\leftrightarrow}}
\newunicodechar{¬}{\ensuremath{\neg}}
\newunicodechar{∧}{\ensuremath{\land}}
\newunicodechar{∨}{\ensuremath{\lor}}
\newunicodechar{≠}{\ensuremath{\neq}}
\newunicodechar{≡}{\ensuremath{\equiv}}
\newunicodechar{∼}{\ensuremath{\sim}}
\newunicodechar{≈}{\ensuremath{\approx}}
\newunicodechar{≥}{\ensuremath{\geq}}
\newunicodechar{≤}{\ensuremath{\leq}}
\newunicodechar{≫}{\ensuremath{\gg}}
\newunicodechar{≪}{\ensuremath{\ll}}
\newunicodechar{∅}{\ensuremath{\emptyset}}
\newunicodechar{⊆}{\ensuremath{\subseteq}}
\newunicodechar{⊂}{\ensuremath{\subset}}
\newunicodechar{∩}{\ensuremath{\cap}}
\newunicodechar{⋂}{\ensuremath{\cap}}
\newunicodechar{∪}{\ensuremath{\cup}}
\newunicodechar{⋃}{\ensuremath{\cup}}
\newunicodechar{⊎}{\ensuremath{\uplus}}
\newunicodechar{∈}{\ensuremath{\in}}
\newunicodechar{∉}{\ensuremath{\not\in}}
\newunicodechar{⊤}{\ensuremath{\top}}
\newunicodechar{⊥}{\ensuremath{\bot}}
\newunicodechar{₀}{\ensuremath{_0}}
\newunicodechar{₁}{\ensuremath{_1}}
\newunicodechar{₂}{\ensuremath{_2}}
\newunicodechar{₃}{\ensuremath{_3}}
\newunicodechar{₄}{\ensuremath{_4}}
\newunicodechar{₅}{\ensuremath{_5}}
\newunicodechar{₆}{\ensuremath{_6}}
\newunicodechar{₇}{\ensuremath{_7}}
\newunicodechar{₈}{\ensuremath{_8}}
\newunicodechar{₉}{\ensuremath{_9}}
\newunicodechar{⁰}{\ensuremath{^0}}
\newunicodechar{¹}{\ensuremath{^1}}
\newunicodechar{²}{\ensuremath{^2}}
\newunicodechar{³}{\ensuremath{^3}}
\newunicodechar{⁴}{\ensuremath{^4}}
\newunicodechar{⁵}{\ensuremath{^5}}
\newunicodechar{⁶}{\ensuremath{^6}}
\newunicodechar{⁷}{\ensuremath{^7}}
\newunicodechar{⁸}{\ensuremath{^8}}
\newunicodechar{⁹}{\ensuremath{^9}}
\newunicodechar{𝔹}{\ensuremath{\mathbb{B}}}
\newunicodechar{ℝ}{\ensuremath{\mathbb{R}}}
\newunicodechar{ℕ}{\ensuremath{\mathbb{N}}}
\newunicodechar{ℂ}{\ensuremath{\mathbb{C}}}
\newunicodechar{ℚ}{\ensuremath{\mathbb{Q}}}
\newunicodechar{𝕋}{\ensuremath{\mathbb{T}}}
\newunicodechar{𝕏}{\ensuremath{\mathbb{X}}}
\newunicodechar{ℤ}{\ensuremath{\mathbb{Z}}}
\newunicodechar{✓}{\ding{51}}
\newunicodechar{✗}{\ding{55}}
\newunicodechar{◊}{\ensuremath{\lozenge}}
\newunicodechar{□}{\ensuremath{\square}}
\newunicodechar{𝓐}{\ensuremath{\mathcal{A}}}
\newunicodechar{𝓑}{\ensuremath{\mathcal{B}}}
\newunicodechar{𝓒}{\ensuremath{\mathcal{C}}}
\newunicodechar{𝓓}{\ensuremath{\mathcal{D}}}
\newunicodechar{𝓔}{\ensuremath{\mathcal{E}}}
\newunicodechar{𝓕}{\ensuremath{\mathcal{F}}}
\newunicodechar{𝓖}{\ensuremath{\mathcal{G}}}
\newunicodechar{𝓗}{\ensuremath{\mathcal{H}}}
\newunicodechar{𝓘}{\ensuremath{\mathcal{I}}}
\newunicodechar{𝓙}{\ensuremath{\mathcal{J}}}
\newunicodechar{𝓚}{\ensuremath{\mathcal{K}}}
\newunicodechar{𝓛}{\ensuremath{\mathcal{L}}}
\newunicodechar{𝓜}{\ensuremath{\mathcal{M}}}
\newunicodechar{𝓝}{\ensuremath{\mathcal{N}}}
\newunicodechar{𝓞}{\ensuremath{\mathcal{O}}}
\newunicodechar{𝓟}{\ensuremath{\mathcal{P}}}
\newunicodechar{𝓠}{\ensuremath{\mathcal{Q}}}
\newunicodechar{𝓡}{\ensuremath{\mathcal{R}}}
\newunicodechar{𝓢}{\ensuremath{\mathcal{S}}}
\newunicodechar{𝓣}{\ensuremath{\mathcal{T}}}
\newunicodechar{𝓤}{\ensuremath{\mathcal{U}}}
\newunicodechar{𝓥}{\ensuremath{\mathcal{V}}}
\newunicodechar{𝓦}{\ensuremath{\mathcal{W}}}
\newunicodechar{𝓧}{\ensuremath{\mathcal{X}}}
\newunicodechar{𝓨}{\ensuremath{\mathcal{Y}}}
\newunicodechar{𝓩}{\ensuremath{\mathcal{Z}}}
\newunicodechar{…}{\ensuremath{\ldots}}
\newunicodechar{∗}{\ensuremath{\ast}}
\newunicodechar{⊢}{\ensuremath{\vdash}}
\newunicodechar{⊧}{\ensuremath{\models}}
\newunicodechar{′}{\ensuremath{'}}
\newunicodechar{″}{\ensuremath{''}}
\newunicodechar{‴}{\ensuremath{'''}}
\newunicodechar{∥}{\ensuremath{\|}}
\newunicodechar{⊕}{\ensuremath{\oplus}}
\newunicodechar{⁺}{\ensuremath{^+}}
\newunicodechar{⊇}{\ensuremath{\supseteq}}
\newunicodechar{∘}{\ensuremath{\circ}}
\newunicodechar{∙}{\ensuremath{\cdot}}
\newunicodechar{⋅}{\ensuremath{\cdot}}
\newunicodechar{≈}{\ensuremath{\approx}}
\newunicodechar{×}{\ensuremath{\times}}
\newunicodechar{∞}{\ensuremath{\infty}}
\newunicodechar{⊑}{\ensuremath{\sqsubseteq}}
 \RequirePackage{tikz}
\usetikzlibrary{calc,scopes}

\usetikzlibrary{fit,positioning}
\usetikzlibrary{arrows.meta}
\usetikzlibrary{chains}
\usetikzlibrary{automata}

\pgfdeclarelayer{background}
\pgfsetlayers{background,main}

\tikzset{
  strut sized/.style = {
    execute at end node={\strut}
  },
  over/.style={remember picture,overlay},
  tight/.style={
    inner sep=0pt,
    outer sep=0pt,
  },
  clear/.style={
    preaction={draw=white,line width=2.5pt,-, shorten <=4pt}
  },
  shorten/.style={
    shorten >=#1,shorten <=#1,
  },
  shorten/.default=2pt,
  exclude from bounding box/.style={
    execute at begin scope={\useasboundingbox;},
  },
  between/.style args={#1 and #2}{
    at = ($(#1)!0.5!(#2)$),
  },
}

\tikzset{
  trans/.style={
    draw,-{Stealth[round]}, semithick, shorten >= 1pt,
  },
}

\tikzset{
  hmsc/.style={
    event sep/.store in=\hmsceventsep,
    event sep=1em,
    clear blockguys/.code={\gdef\blockguys{}},
    add blockguy/.code={\xdef\blockguys{\blockguys(##1)(##1-end)}},
    add block node/.code={\xdef\blockguys{\blockguys(##1)}},
    block style/.style={
      draw,
      rounded corners=4pt,
      inner sep=7pt,
      minimum size=10pt,
      node contents={},
    },
    empty block/.style={
      block style,
    },
    block/.style={
      clear blockguys,
add blockguy/.list={##1},
fit={\blockguys},
      block style,
    },
    msc/.style={
      name prefix=##1,
      clear blockguys,
execute at begin scope={
        \gdef\blockguys{}
      },
      execute at end scope={
        \node[
          block style,
          msc node style/.try,
          fit={\blockguys},
          name=,
        ];
      },
      start chain=participants going right,
      head/.append style={on chain=participants},
      final/.style={msc node style/.append style={accepting}},
      init/.style={msc node style/.append style={initial}},
    },
    participant/.style={
      add blockguy={##1},
      start chain=##1 going {below=\hmsceventsep of \tikzchainprevious},
      every join/.style={semithick,gray},
      every on chain/.style={join},
      head/.append style={
        name=##1,
      },
      participant head/.style={participant name={}},
      execute at end scope={
        \node[skip=3pt] {};
      },
    },
    participant name/.style={node contents={##1}},
    head/.style={
inner sep=0pt,append after command={(\tikzlastnode) ++(0,-2pt) node[skip,outer sep=3pt]},
      participant head/.try,
    },
    no event/.style={
      on chain,
      tight,
      minimum size=0pt,
    },
    event/.style={
      on chain,
      inner sep=0pt,
      outer sep=0pt,
minimum size=2pt,
      fill=black,
      circle,
    },
    skip/.default=0pt,
    skip/.style={
      on chain,
      yshift=##1,
      tight,
      node contents={},
    },
    messages/.style={
      thick,
      -{[scale=.7]Triangle},
      shorten=1pt,
      every edge/.append style={clear},
    },
    messagesMC/.style={
      thick,
      {[scale=.7,open]Triangle}-{[scale=.7]Triangle},
      shorten=1pt,
      every edge/.append style={clear},
    },
    msg/.style={
      above,
      fill=white,
      inner sep=1pt,
      rounded corners=.5ex,
      font={\tiny},
    },
    anonymous/.style={
      participant name/.style={node contents={}},
    },
    eventless/.style={
      event/.style={
        on chain,
        inner sep=0pt,
        outer sep=0pt,
        minimum size=0pt,
      },
    },
    on grid,
node distance=5ex,
    initial where=above,
  },
}
 
\ifoptionfinal
{\usepackage[disable]{todonotes}}
{\usepackage{todonotes}}

\newcommand\FS[1]{\todo[color=white,inline]{\textcolor{blue}{FS, #1}}}

\usepackage{hyperref}
\usepackage[capitalise]{cleveref}
\usepackage[shortlabels]{enumitem}

\crefformat{section}{\S#2#1#3} \crefmultiformat{section}{\S#2#1#3}{ and~\S#2#1#3}{, \S#2#1#3}{, and~\S#2#1#3} \crefrangeformat{section}{\S#3#1#4 to~\S#5#2#6}  

\setlist[itemize]{leftmargin=*,topsep=0.3ex,itemsep=0.1ex} 
\setlist[enumerate]{leftmargin=*,topsep=0.3ex,itemsep=0.1ex}

\theoremstyle{remark}
\newtheorem{remark}{Remark}[section] 

\AtBeginDocument{}

\begin{document}

\title
    [Global Protocols under Rendezvous Synchrony: From Realizability to Type Checking]
    {
Global Protocols under Rendezvous Synchrony: \\From Realizability to Type Checking
    }
	
	\author{Elaine Li}
	\orcid{0000-0003-0173-4498}
	\email{efl9013@nyu.edu}
	\affiliation{
		     \institution{New York University}
		     \country{USA}
		 }
	\author{Felix Stutz}
	\orcid{0000-0003-3638-4096}
	\email{felix.stutz@uni.lu}
	\affiliation{
		\institution{University of Luxembourg}
		\country{Luxembourg}
	}

\begin{abstract}
    Global protocol specifications are the starting point of top-down verification methodologies, and serve as a blueprint for synthesizing local specifications that guarantee the correctness of distributed implementations. 
In this work, we study global protocol specifications targeting distributed processes that communicate via rendezvous synchrony. 
We obtain the following positive results for the synchronous realizability problem: 
(a) realizability is decidable for global protocols over a transitive concurrency alphabet in \mbox{2-\EXPTIME} in the size of the protocol, and in 3-\EXPTIME in the size of the alphabet, and
(b) realizability is decidable in \EXPTIME for global protocols that \emph{unambiguously} represent their trace language. 
Unambiguous global protocols encompass fragments of directed and sender-driven choice studied in prior work. 
Further, our reductions admit the corollary that the synchronous verification problem is solvable with the same complexity, where a candidate realization is included as part of the input. 
We then prove a negative result stating that synchronous realizability is undecidable in general. 
Finally, we propose a type system to check $\pi$-calculus processes against local specifications in the form of synchronous communicating state machines. 
Our type system is the first to support processes with local mixed choice in the presence of session interleaving and~delegation.  

 	\end{abstract}

\maketitle
	
	\section{Introduction}
\label{sec:introduction}

Global protocols are an abstraction for specifying message-passing protocols that describes the behavior of all participants jointly from a birds-eye view. 
Due to their attractive simplicity, global protocols enjoy industry adoption as part of the ITU standard \cite{z120-standard}, are studied in the form of high-level message sequence charts \mbox{(HMSCs)} \cite{DBLP:conf/sdl/MauwR97,
DBLP:conf/ac/GenestMP03,DBLP:conf/acsd/GenestM05,DBLP:conf/concur/GazagnaireGHTY07,DBLP:journals/tosem/RoychoudhuryGS12, 
DBLP:journals/tse/AlurEY03, 
DBLP:journals/tcs/Lohrey03, 
DBLP:conf/concur/AlurY99,DBLP:conf/mfcs/MuschollP99, 
DBLP:conf/stacs/Morin02,
DBLP:journals/jcss/GenestMSZ06}, 
 and prominently featured in programming language frameworks including multiparty session types (MSTs)~\cite{DBLP:conf/popl/HondaYC08, DBLP:conf/concur/BocchiHTY10, DBLP:conf/tgc/BocchiDY12, DBLP:journals/jlp/ToninhoY17, DBLP:journals/pacmpl/00020HNY20, DBLP:conf/cav/LiSWZ23, DBLP:journals/pacmpl/HinrichsenBK20, DBLP:journals/lmcs/HinrichsenBK22, DBLP:journals/pacmpl/JacobsHK23} and choreographic programming~\cite{DBLP:journals/tcs/Cruz-FilipeM20,DBLP:conf/ecoop/GiallorenzoMPRS21,DBLP:journals/pacmpl/HirschG22,DBLP:journals/corr/abs-2111-03484,DBLP:journals/corr/abs-2303-00924}. 
 Implementations of such frameworks can be found in more than a dozen programming languages. 

A central feature of global protocols is their separation of concerns between protocol logic and communication. 
Purely as a specification formalism, global protocols do not commit to a particular network model. 
Consider the example protocol $\protocol$ depicted in \cref{fig:receiver-power-sync-protocol}, involving three participants $\procA, \procB$ and $\procC$. 
The protocol $\protocol$ simply captures the intention that all participants exchange the same message value, either $\val_1$ or $\val_2$, which is decided upon by participant $\procA$. 

Global protocols are designed to be ultimately implemented by a collection of distributed processes communicating with one another over a network. 
A typical workflow takes a global protocol and generates local specifications, one for each participant, that govern the behavior of local implementations. 
Figures \ref{fig:receiver-power-sync-p}, 
\ref{fig:receiver-power-sync-q}, and 
\ref{fig:receiver-power-sync-r} depict local specifications for each participant in $\protocol$ in the form of finite automata, 
 whose transitions are labeled with send and receive events, \eg $\snd{\procA}{\procB}{\val_1}$ denotes $\procA$ sending $\val_1$ to $\procB$. \cref{fig:receiver-power-sync-procs} depicts each participant's implementation as a $\pi$-calculus process. 

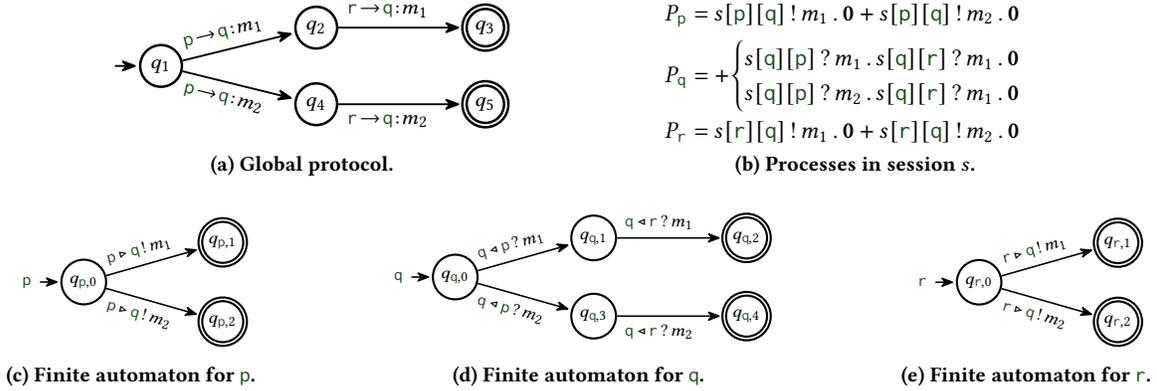
\begin{figure*}[t]
\begin{subfigure}[b]{0.32\textwidth}
\centering
\resizebox{1\textwidth}{!}{
    \begin{tikzpicture}[sem, node distance=0.08cm and 1.4cm]

  \node[state, initial left] (p1) {$q_{1}$};
  \node[state, above right = of p1] (p2) {$q_{2}$}; 
  \node[state, final, right = of p2] (p3) {$q_{3}$}; 
  \node[state, below right = of p1] (p4) {$q_{4}$}; 
  \node[state, final, right = of p4] (p5) {$q_{5}$}; 

  \path 
  (p1) edge [sloped,above] node {$\msgFromTo{\procA}{\procB}{\val_1}$} (p2)
  (p1) edge [sloped,below] node {$\msgFromTo{\procA}{\procB}{\val_2}$} (p4)
  (p2) edge [above] node {$\msgFromTo{\procC}{\procB}{\val_1}$} (p3)
  (p4) edge [below] node {$\msgFromTo{\procC}{\procB}{\val_2}$} (p5)
  ;

\end{tikzpicture}
 }
    \caption{Global protocol.}
\label{fig:receiver-power-sync}
    \label{fig:receiver-power-sync-protocol}
\end{subfigure}
\hspace{1.5cm}
\begin{subfigure}[b]{0.32\textwidth}
\begin{align*}
P_\procA & = 
    s[\procA][\procB] \sendOp \val_1 \seq \zero 
    +
    s[\procA][\procB] \sendOp \val_2 \seq \zero
\\
P_\procB & = 
    {+} \begin{cases}
    s[\procB][\procA] \recOp \val_1 \seq 
    s[\procB][\procC] \recOp \val_1 \seq \zero
    \\ 
    s[\procB][\procA] \recOp \val_2 \seq 
    s[\procB][\procC] \recOp \val_1 \seq \zero
    \end{cases}
\\
P_\procC & = 
    s[\procC][\procB] \sendOp \val_1 \seq \zero 
    +
    s[\procC][\procB] \sendOp \val_2 \seq \zero
\end{align*}
\vspace{-4ex}
\caption{Processes in session $s$.}
    \label{fig:receiver-power-sync-procs}
\end{subfigure}
\\
\vspace{2.5ex}
\begin{subfigure}[b]{0.33\textwidth}
\centering
    \begin{tikzpicture}[sem, node distance=0.08cm and 1.4cm]

  \node[state, initial left, initial text={$\procA$}] (p0) {$q_{\procA, 0}$};
  \node[state, final, above right = of p0] (p1) {$q_{\procA, 1}$}; 
  \node[state, final, below right = of p0] (p2) {$q_{\procA, 2}$}; 

  \path 
  (p0) edge [sloped,above] node {$\snd{\procA}{\procB}{\val_1}$} (p1)
  (p0) edge [sloped,below] node {$\snd{\procA}{\procB}{\val_2}$} (p2)
  ;

\end{tikzpicture}
     \caption{Finite automaton for $\procA$.}
    \label{fig:receiver-power-sync-p}
\end{subfigure}
\begin{subfigure}[b]{0.33\textwidth}
\centering
    \begin{tikzpicture}[sem, node distance=0.08cm and 1.4cm]

  \node[state, initial left, initial text={$\procB$}] (p0) {$q_{\procB, 0}$};
  \node[state, above right = of p0] (p1) {$q_{\procB, 1}$}; 
  \node[state, final, right = of p1] (p2) {$q_{\procB, 2}$}; 
  \node[state, below right = of p0] (p3) {$q_{\procB, 3}$}; 
  \node[state, final, right = of p3] (p4) {$q_{\procB, 4}$}; 

  \path 
  (p0) edge [sloped,above] node {$\rcv{\procA}{\procB}{\val_1}$} (p1)
  (p1) edge [sloped,above] node {$\rcv{\procC}{\procB}{\val_1}$} (p2)
  (p0) edge [sloped,below] node {$\rcv{\procA}{\procB}{\val_2}$} (p3)
  (p3) edge [sloped,below] node {$\rcv{\procC}{\procB}{\val_2}$} (p4)
  ;

\end{tikzpicture}
     \caption{Finite automaton for $\procB$.}
    \label{fig:receiver-power-sync-q}
\end{subfigure}
\begin{subfigure}[b]{0.33\textwidth}
\centering
    \begin{tikzpicture}[sem, node distance=0.08cm and 1.4cm]

  \node[state, initial left, initial text={$\procC$}] (p0) {$q_{\procC, 0}$};
  \node[state, final, above right = of p0] (p1) {$q_{\procC, 1}$}; 
  \node[state, final, below right = of p0] (p2) {$q_{\procC, 2}$}; 

  \path 
  (p0) edge [sloped,above] node {$\snd{\procC}{\procB}{\val_1}$} (p1)
  (p0) edge [sloped,below] node {$\snd{\procC}{\procB}{\val_2}$} (p2)
  ;

\end{tikzpicture}
     \caption{Finite automaton for $\procC$.}
    \label{fig:receiver-power-sync-r}
\end{subfigure}
\caption{
    A global protocol (\subref{fig:receiver-power-sync-protocol}), 
    local specifications as finite automata 
    (\subref{fig:receiver-power-sync-p}, 
     \subref{fig:receiver-power-sync-q}, and 
     \subref{fig:receiver-power-sync-r}),
    and implementations as $\pi$-calculus processes~(\subref{fig:receiver-power-sync-procs}). 
    }
\end{figure*}

Observe that none of the components in the workflow commit to a particular network model. 
If we allow $\procA, \procB$ and $\procC$ to communicate via rendezvous synchrony, then the participants collectively follow~$\protocol$, exhibiting exactly the desired behaviors and avoiding deadlocks. 
On the other hand, if we allow $\procA, \procB$ and $\procC$ to communicate asynchronously, through peer-to-peer FIFO channels, a violation of $\protocol$ occurs: participant $\procC$, who cannot infer the choice made by $\procA$, can independently send a message that does not match the message exchanged between $\procA$ and $\procB$. 

The choice of network model is of central concern to the \emph{realizability} problem, which asks whether a given global protocol is truly distributable, such that its participants avoid communication errors or system-wide deadlock, despite only having a partial view of the entire system. 
Realizability is defined relative to a target implementation and network model. 
Investigations of realizability in the literature primarily target asynchronous, reliable, peer-to-peer FIFO communication. 
Decidability and complexity results have been obtained for global protocols in the form of \mbox{HMSCs}~\cite{DBLP:journals/tse/AlurEY03, 
DBLP:journals/tcs/Lohrey03, 
DBLP:conf/concur/AlurY99,DBLP:conf/mfcs/MuschollP99, 
DBLP:conf/stacs/Morin02,DBLP:journals/jcss/GenestMSZ06}, and for global communicating labeled transition systems~\cite{DBLP:journals/pacmpl/LiSWZ25}, which encompass global types.

Recently there has been increased interest in synchronous communication models, both for their expressivity and their realizability problem~\cite{DBLP:conf/coordination/BarbaneraLT20,DBLP:journals/lmcs/BarbaneraLT23,DBLP:conf/lics/PetersY24,DBLP:conf/ppdp/GiustoLU25,NODBLPyet:journals/corr/abs-2511-22692}. 
\citet{DBLP:conf/lics/PetersY24} study the impact of local mixed choice on the expressivity of synchronously communicating $\pi$-calculus processes, and propose a type system for such processes. 
Local mixed choice allows a process to choose between receiving and sending a message. 
\citet{NODBLPyet:journals/corr/abs-2511-22692} propose a synchronous MST framework, 
in which typability of a process against a global protocol implies that the former correctly implements the latter. 
With this, they soundly approximate a variant of the synchronous realizability problem, which we call the synchronous verification problem, through restrictions on global protocols and typing rules. 
Of recent interest is \cite{DBLP:conf/ppdp/GiustoLU25}, which connects the realizability problem to complementability, and proposes the first algorithms for deciding synchronous realizability of various global type fragments, including sender-driven global types~\cite{DBLP:conf/concur/MajumdarMSZ21}. 
Sender-driven choice allows a sender to send messages to different receives in a choice. 

In this paper, we make the following contributions. Our decidability and complexity results are contextualized in \cref{fig:venn-diagram}. 
\begin{itemize} 
	\item We show that synchronous realizability is decidable for global protocols over a transitive concurrency alphabet in 2-\EXPTIME in the size of the protocol, and in 3-\EXPTIME in the size of the alphabet. To the best of our knowledge, this constitutes the first upper bound for the equivalence and inclusion problems of trace languages with a transitive concurrency alphabet. We additionally provide an exhaustive case analysis of transitive alphabets in terms of their participant set, message value set, and communication topology. 
	\item We show that synchronous realizability is decidable in \EXPTIME for global protocols that \emph{unambiguously} represent their trace language, via a reduction to the multiplicity equivalence problem for rational trace languages. 
	\item Our reductions allow us to conclude that the synchronous verification problem is solvable with the same complexity. 
	\item We show that synchronous realizability is undecidable in general. 
	\item We propose a type system for checking $\pi$-calculus processes against finite state automata, that is the first to simultaneously handle local mixed choice, session interleaving and delegation. 
\end{itemize}

Our results are inspired by two separate lines of work. 
The first line of work uses Mazurkiewicz trace theory to answer questions about message sequence charts (MSCs) and HMSCs.  
Specifically, reductions target \emph{recognizable} trace languages, which are a well-behaved fragment corresponding to sets of $I$-closed regular languages. 
Correspondingly, restrictions are proposed to ensure MSCs and HMSCs fall into the recognizable fragment, such as loop-connectedness and boundedness~\cite{DBLP:journals/tcs/Lohrey03, DBLP:conf/mfcs/MuschollP99, DBLP:conf/stacs/Morin02}. 
Kleene characterizations and MSO-definability characterizations have also been obtained~\cite{DBLP:conf/dlt/GenestMK04}. 
The second line of work~\cite{DBLP:conf/cav/LiSWZ23, DBLP:conf/esop/LiSW24, DBLP:conf/esop/StutzD25} is situated within the context of using global protocols to design programming language and verification frameworks. 
Deviating from the standard workflow, these works separate the realizability problem of global protocols from problem of synthesizing local specifications. 
This separation has enabled the target implementation models to be generalized in expressivity.

\begin{figure*}[t]
	\centering
	\scalebox{0.75}{
\colorlet{lightgray}{gray!50!}
		\colorlet{colLinesKnown}{lightgray}
		\colorlet{colTextKnown}{gray}
\begin{tikzpicture}[line width=1.5pt]

\def \rectangleRAT
    {(-1.5,-2) rectangle (16,8)}
\def \rectangleREC
    {(1.5,0.5) rectangle (14.5,7.5)}
\def \rectangleCC
    {(2.0,1.0) rectangle (12.25,7)}
\def \rectangleUR
    {(2.5,-1.5) rectangle (11.75,5.85)}

\def \ellipseItransitive {(4.0, 2.8) ellipse 
[x radius=4.7cm, y radius=4.0cm, rotate=0]}
\def \ellipseIempty{(5.0, 3.4) ellipse 
[x radius=2.1cm, y radius=2.1cm, rotate=0]}
\def \ellipsePleqThree{(5.6, 2.7) ellipse 
[x radius=0.8cm, y radius=0.8cm, rotate=0]}
\def \ellipseCD{(7.7, 1.8) ellipse 
[x radius=3.7cm, y radius=3.0cm, rotate=0]}
\def \ellipseSD{(7.8, 0.9) ellipse 
[x radius=2.4cm, y radius=1.3cm, rotate=-30]}

\draw [densely dotted] \rectangleRAT 
node [
    label={[below, xshift=-16.65cm, yshift=-0.2cm, align=left]
            rational \\ trace \\ languages}
    ] 
{};
\draw \rectangleREC
node [
    label={[xshift=-1.15cm, yshift=-1.1cm, align=center]
            star-connected \\ (\cref{subsec:unambiguous})},
    label={[xshift=-1.15cm, yshift=-1.55cm, align=center]
            \EXPSPACE},
    ] 
{};
\draw [colLinesKnown] \rectangleCC
node [
    label={[xshift=-1.85cm, yshift=-0.7cm, align=center, text=colTextKnown]
            commutation-closed (\cref{subsec:unambiguous})}, 
    label={[xshift=-1.85cm, yshift=-1.1cm, align=center, text=colTextKnown]
            \PSPACE}
     ]
{};
\draw \rectangleUR 
node [
    label={[below, xshift=-1.2cm, yshift=-0.2cm, align=center]
            unambiguous \\ (\cref{subsec:unambiguous})},
    label={[below, xshift=-1.2cm, yshift=-1.2cm, align=center]
            \EXPTIME},
    ] 
{};
\draw \ellipseItransitive 
node [
    label={[xshift=-3.7cm, yshift=-0.2cm, align=center]
            $I$ transitive \\ (\cref{subsec:transitive})}, 
    label={[xshift=-3.7cm, yshift=-0.75cm, align=center]
            2-\EXPTIME}, 
     ]
{};
\draw \ellipseIempty
node [
    label={[xshift=-0.0cm, yshift=0.9cm, align=center] 
            $I$ empty \\ (\cref{subsec:transitive})}, 
     ]
{};
\draw [colLinesKnown] \ellipsePleqThree 
node [
    label={[xshift=0.0cm, yshift=-0.35cm, color=colTextKnown, align=center]
            $\card{\Procs} \leq 3$ \\  (\cref{subsec:unambiguous})}
    ] 
{};
\draw [colLinesKnown] \ellipseCD 
node [
    label={[xshift=2.1cm, yshift=0.4cm, color=colTextKnown, align=center]non-mixed\\choice \\  (\cref{subsec:unambiguous})}, 
    label={[xshift=2.1cm, yshift=-0.2cm, color=colTextKnown, align=center]\PSPACE}
    ] 
{};
\draw [colLinesKnown] \ellipseSD 
node [
    label={[xshift=1.1cm, yshift=-1.6cm, color=colTextKnown, align=center]
            sender-driven \\  (\cref{subsec:unambiguous})} 
    ] 
{};

\end{tikzpicture} 	}
	\caption
	{
		Contextualization of our positive results. 
		Known results \cite{DBLP:conf/ppdp/GiustoLU25} are depicted in gray. 
		Sets denote trace languages (existentially) represented by regular string languages. 
		For every set without a complexity annotation, its complexity is the smallest of all supersets. 
	} 
	\label{fig:venn-diagram}
\end{figure*}
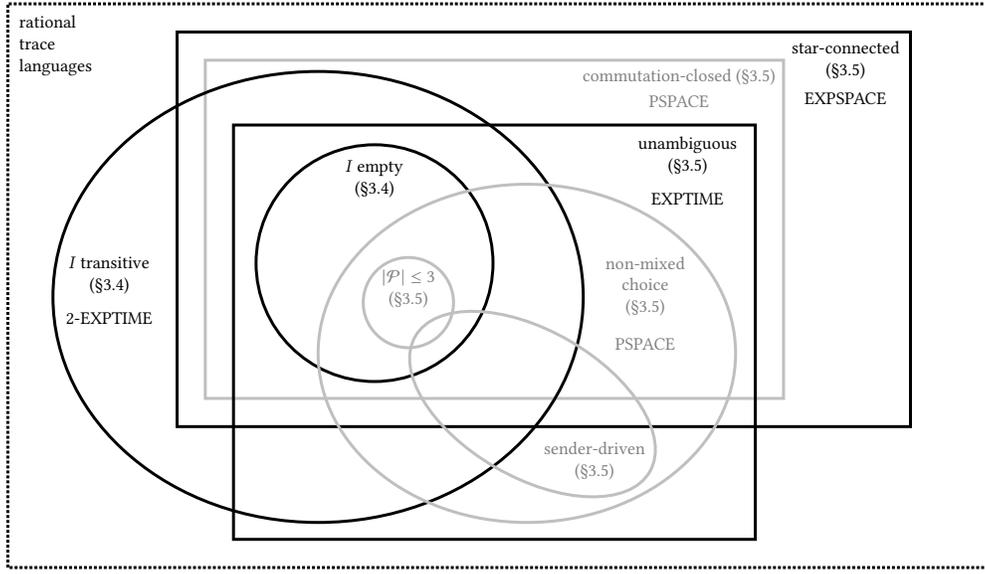

 	\section{Preliminaries} 

\paragraph{Words.}
Let $\Sigma$ be an alphabet.
$\Sigma^*$ denotes the set of finite words over $\Sigma$, $\Sigma^\omega$ the set of infinite words, and $\Sigma^\infty$ their union $\Sigma^* \cup \Sigma^\omega$.
$\Sigma^+$ denotes the set of non-empty finite words over $\Sigma$, \ie $\Sigma^+ = \Sigma^*\Sigma$. 
A word $u \in \Sigma^*$ is a \emph{prefix} of word $v \in \Sigma^\infty$, denoted $u \leq v$, if there exists $w \in \Sigma^\infty$ with $u \cdot w = v$;
we denote all prefixes of $u$ with $\pref(u)$, which we extend to sets of words $\pref(S)$ in the usual way. 
We denote by $\alphOp(u)$ the set of symbols in $u$.

\paragraph{Message Alphabets.}
Let $\Procs$ be a finite set of participants and $\MsgVals$ be a finite data domain. 
We define the \emph{synchronous alphabet} as the set of events $\AlphSync \is \set{ \msgFromTo{\procA}{\procB}{\val} \mid \procA,\procB ∈ \Procs, \procA \neq \procB \text{ and } \val ∈ \MsgVals}$ where $\msgFromTo{\procA}{\procB}{\val}$ denotes a message exchange of $\val$ between $\procA$ and $\procB$.
For each participant $\procA\in \Procs$, we define the alphabet
$\AlphSync_{\procA} = 
\set{\msgFromTo{\procA}{\procB}{\val} \mid \procB \in \Procs,\; \val \in \MsgVals }
\union
\set{\msgFromTo{\procB}{\procA}{\val} \mid \procB \in \Procs,\; \val \in \MsgVals }
$ of events that $\procA$ participates in. We define a homomorphism~$\wproj_{\AlphSync_\procA}$, where $x \wproj_{\AlphSync_\procA} = x$ if $x \in \AlphSync_{\procA}$, and $\emptystring$ otherwise.
We write $\Procs(w)$ to project the events in $w$ onto their participants, and $\MsgVals(w)$ to project the events in $w$ onto their messages.

\paragraph{Finite Automata.}
A \emph{finite automaton} 
over the alphabet $\Sigma$ is a tuple
$A = (Q, \Sigma, \delta, q_0, F)$,
where
$Q$ is a finite non-empty set of states, 
$\delta \from Q \times \Sigma \to Q$
is a transition relation , $q_0 \in Q$ is the initial state, and 
$F \subseteq Q$ is a set of final states.
The language recognized by $A$ is defined by $\lang(A) = \set{w \in \Sigma^* \mid \delta^*(q_0, w) \in F}$, where $\delta^*$ is the reflexive transitive closure of $\delta$. 
We use 
$q_1 \xrightarrow{\alpha} q_2$ 
to denote $q_1 \xrightarrow{w}\starred q_2$, and 
$q_1 \xrightarrow{w}\starred q_2$
to denote 
$q_2 \in \delta^*(q_1, w)$. 
We say that $A$ is \emph{deterministic} iff $\delta$ is a function and we abbreviate deterministic (resp.\ non-deterministic) finite automaton with DFA (resp.\ NFA). 
We say that $A$ is \emph{complete} iff $\delta$ is defined for every $q \in Q$ and $\alpha \in \Sigma$. 
The automaton $A$ is \emph{reduced} iff for every state in $q \in Q$, there exists $u \in \Sigma^*$ such that $\delta(q_0, u) = q$, and $v \in \Sigma^*$ such that $\delta(q, v) \in F$.  
Observe that a finite automaton always admits a language-preserving transformation into one that is deterministic and reduced.

\paragraph{Global Protocols.} We use reduced, deterministic finite automata (DFAs) over the synchronous alphabet~$\AlphSync$ to model global protocols, which we denote with $\protocol$.

\paragraph{Restricting Protocols to Participants.}
Given a global protocol $\protocol = (Q, \AlphSync, \delta, q_0, F)$, we define an NFA $\protocol_\procA$ for each participant $\procA$ via domain restriction to $\AlphSync_\procA$: 
\vspace{-1.5ex}
\[
\protocol_\procA \is
(Q, \AlphSync_\procA \dunion \set{\emptystring}, \delta_\procA, q_0, F) 
\text{ with } 
\delta_\procA \is
\set{q \xrightarrow{l \wproj_{\AlphSync_\procA}} q'
	\mid q' \in \delta(q, l)}
    \enspace .
\]

Realizability is defined relative to our target realization model, which is synchronous communicating state machines. 
Our realization model coincides with \emph{loosely cooperating systems} from \cite{DBLP:conf/fsttcs/CastellaniMT99}, and its semantics coincide with the synchronous runs of a communicating finite state machine 
in \cite{DBLP:conf/ppdp/GiustoLU25}, but differs from Zielonka's asynchronous automata~\cite{DBLP:journals/ita/Zielonka87}, also called synchronously communicating systems in \cite{DBLP:conf/fsttcs/CastellaniMT99}, in that individual finite automata do not read one another's states when synchronizing on an event. 

\paragraph{Synchronous Communicating State Machines.}
$\CSMabb{A} = \CLTS{\CSMabb{A}}$ is a \emph{synchronous communicating state machine} (SCSM) over $\Procs$ and~$\MsgVals$ if
$\CSMabb{A}_\procA$
is a DFA
over~$\AlphSync_\procA$ for every $\procA\in\Procs$, denoted by
$(Q_\procA, \AlphSync_\procA, \delta_\procA, q_{0, \procA}, F_\procA)$.
A~\emph{configuration} of $\CSMabb{A}$ is a global state $\vec{q} \in \prod_{\procA \in \Procs} Q_\procA$. 
We use $\vec{q}_\procA$ to denote the state of $\procA$ in $\vec{q}$.
The transition relation, denoted $\rightarrow$, is defined as follows: 
\begin{align*}
    & 
	\vec{q} \xrightarrow{\msgFromToNS{\procA}{\procB}{\val}} \pvec{q}' \text{ if }
	(\vec{q}_\procA, \msgFromTo{\procA}{\procB}{\val}, \pvec{q}'_\procA)\in\delta_\procA, \;
    (\vec{q}_\procB, \msgFromTo{\procA}{\procB}{\val}, \pvec{q}'_\procB)\in\delta_\procB, \\
    & 
    \text{and }
	\vec{q}_\procC = \pvec{q}'_\procC
    \text{ for every } 
    \procC
    \text{ with }
    \procA \neq \procC \neq \procB 
    \enspace .
\end{align*}

\vspace{-1ex}
As for finite automata, we use $\xrightarrow{l}\starred$ for the reflexive transitive closure and may omit the post-state. 
In the initial configuration $\vec{q}_0$, each participant's state is the initial state $q_{0,\procA}$ of $A_\procA$. 
A configuration $\vec{q}$ is said to be \emph{final} iff $\vec{q}_\procA$ is final for every $\procA \in \Procs$. 
Runs and traces are defined in the expected way and $\reach(\CSMabb{A})$ denotes the set of reachable configurations.\footnote{We use the word trace to refer to both execution traces of an SCSM and elements of a free partially commutative monoid, as the meaning is always clear from context.}
A run is \emph{maximal} if it is finite and ends in a final configuration. 
The language of $\CSMabb{A}$, denoted $\lang(\CSMabb{A})$, is defined as the set of maximal traces.
A configuration $\vec{q}$ is a \emph{deadlock} if it is not final and has no outgoing transitions.
An SCSM is \emph{deadlock-free} if no reachable configuration is a~deadlock.

\paragraph{Synchronous Product of an SCSM}
Given an SCSM $\CSMabb{A}$
over $\Procs$ and~$\MsgVals$, 
where $\CSMabb{A}_\procA = (Q_\procA, \AlphSync_\procA, \delta_\procA, q_{0, \procA}, F_\procA)$
is a DFA
over~$\AlphSync_\procA$ for every $\procA\in\Procs$, 
we define a DFA 
$Prod(\CSMabb{A}) = 
    (\Pi_{\procA \in \Procs} Q_\procA, \AlphSync, \delta, q_0, F)$ 
where 
\begin{itemize}
\item $q_0 = (q_{0,\procA})_{\procA\in \Procs}$, 
\item $F = \Pi_{\procA \in \Procs} F_\procA$, and 
\item 
$(\vec{q}, \msgFromTo{\procA}{\procB}{\val}, \pvec{q}') \in \delta$ 
\; iff \;
$(\vec{q}_\procA, \msgFromTo{\procA}{\procB}{\val}, \pvec{q}'_\procA) \in \delta_\procA$, \\ 
$\qquad 
(\vec{q}_\procB, \msgFromTo{\procA}{\procB}{\val}, \pvec{q}'_\procB) \in \delta_\procB$, and 
$\vec{q}_\procC = \pvec{q}'_\procC$ for all 
$\procC$ with 
$\procA \neq \procC \neq \procB$. 
\end{itemize}

\paragraph{Synchronous Protocol Semantics.}
The synchronous semantics of a protocol $\protocol$ is defined as the set of all words that are per-participant equal to some word in $\lang(\protocol)$: \[
\langsync(\protocol) \is \set{w \in \AlphSync^*~\mid~\exists w' \in \lang(\protocol).~\forall \procA\in \Procs.~w \wproj_{\AlphSync_\procA} = w' \wproj_{\AlphSync_\procA}} 
\]

\begin{definition}[Synchronous Realizability Problem]
	\label{def:sync-impl}
	An SCSM $\CSMabb{A}$ is a synchronous realization of protocol $\protocol$ if the following two properties hold:
	\begin{inparaenum}[(i)]
		\item \label{def:lts-implementability-protocol-fidelity}
		\emph{protocol fidelity}: $\lang(\CSMabb{A}) = \langsync(\protocol)$, and
		\item \label{def:lts-implementability-deadlock-freedom}
		\emph{deadlock freedom}: $\CSMabb{A}$ is deadlock-free.
	\end{inparaenum}
Given a protocol $\protocol$, the \emph{synchronous realizability problem} asks whether there exists an SCSM $\CSMabb{A}$ that is a synchronous realization of $\protocol$. 
\end{definition}

One can also ask the synchronous realization question for a given SCSM as input. 
We refer to this problem as the \emph{synchronous verification problem}.  \begin{definition}[Synchronous Verification Problem]
	\label{def:sync-impl}
	Given a protocol $\protocol$ and an SCSM $\CSMabb{A}$, the \emph{synchronous verification problem} asks whether $\CSMabb{A}$ synchronously realizes $\protocol$. 
\end{definition}

In the MST literature, the two problems are typically studied separately, with the latter reformulated as 
a form of subtyping. 
In \cite{DBLP:conf/esop/LiSW24}, it was shown that sound and complete solutions to the asynchronous realizability problem can be used to derive sound and complete solutions to the asynchronous protocol verification problem. 
Our solutions to realizability in the synchronous setting mirror this observation. 
Indeed, we show in \cref{sec:realizability} that the two problems share an upper bound in complexity. 

We may omit the term ``synchronous'' when clear from context. 
 	\section{Realizability}
\label{sec:realizability}

\subsection{Trace theory preliminaries} 
\label{subsec:trace-theory-prelim} 
We begin with preliminaries from trace theory.

\paragraph{Traces and Trace Languages.} 
Let $\Sigma$ be a finite alphabet and $C$ be an irreflexive, symmetric relation on $\Sigma$, called the \emph{concurrency} relation, or the \emph{independence} relation. 
The pair $(\Sigma, C)$ is called a \emph{concurrent alphabet}. 
The reflexive and symmetric relation $D = \Sigma \times \Sigma \setminus C$ is called the \emph{dependence} relation. 
The relations $C$ and $D$ extend to words in the expected way: $(u,v) \in C$ iff $\alphOp(u) \times \alphOp(v) \subseteq C$ and $(u,v) \in D$ iff 
$(\alphOp(u) \times \alphOp(v)) \inters D \neq \emptyset$. 
Given $x \in \Sigma^*$, we denote $C(x) = \set{w \in \Sigma^* \mid (w,x) \in C}$.

The concurrent alphabet $(\Sigma, C)$ defines a \emph{free partially commutative monoid}, denoted $M(\Sigma, C)$, also called a \emph{trace monoid}, 
which 
is the quotient structure induced by $\equiv_C$, the reflexive, transitive closure of $C$ extended to words. A trace on concurrent alphabet $(\Sigma, C)$ is an element $t \in M(\Sigma, C)$, namely a congruence class of words. We use $[w]_C$ to denote the trace that contains $w \in \Sigma^*$. 

\paragraph{Rational, Recognizable, and Unambiguous Trace Languages.}
Let $L \subseteq \Sigma^*$ be a language and $(\Sigma, C)$ be a concurrent alphabet. 
Language~$L$, as what we may call \emph{string language}, can define a trace language in several ways. The trace language \emph{existentially} defined by $L$, denoted $\Elang{L}{C}$, is the set $\set{t \in M(\Sigma, C) \mid t \inters L \neq \emptyset}$. 
The trace language \emph{universally} defined by $L$, denoted $\Alang{L}{C}$, is the set $\set{t \in M(\Sigma, C) \mid t \subseteq L}$. 
If $\Elang{L}{C} = \Alang{L}{C} = T$, then we say that $T$ is \emph{consistently} defined by $L$, \ie $L = \bigcup_{t \in T} t$. 
Moreover, if $L$ is regular, then $\Elang{L}{C}$ is a existentially regular trace language, and $\Alang{L}{C}$ is a universally regular trace language. 

By convention, we drop the $\exists$ subscript, using $[L]_C$ to denote $\Elang{L}{C}$. 
When clear from context, we further drop the $C$, using $[L]$ to denote $\Elang{L}{C}$. 

Conversely, given a trace language $T \subseteq M(\Sigma, C)$, the set $lin(T) = \bigcup_{t \in T} \varphi^{-1}(t)$ is the \emph{linearization of T}, where $\varphi$ is the canonical morphism from $\Sigma^*$ to $M(\Sigma, C)$. 

Let $T \subseteq M(\Sigma, C)$ be a a trace language. 
We say that $T$ is \emph{rational} iff there exists a regular language $L \subseteq \Sigma^*$ such that $[L]_C = T$. 
We say that $T$ is \emph{recognizable} iff $lin(T)$ is a regular language over $\Sigma$. 
We say $T$ is \emph{unambiguously rational}, or \emph{unambiguous} for short, iff there exists a regular language $L$ such that $T = [L]_C$, and moreover $w \in L$ implies $\card{lin([w]_C) \inters L} = 1$. 
The latter condition is equivalent to $t \in M(\Sigma, C)$ implies $\card{lin(t) \inters L} \leq 1$, or $t \in T$ implies $\card{lin(t) \inters L} = 1$. 
Note that the definition of existentially regular and rational coincide, while the definition of recognizable and consistently regular coincide. 

Recognizable trace languages admit multiple algebraic characterizations. 
One such characterization is that of a star-connected rational expression, where all languages appearing under Kleene star are connected, \ie ruling out concurrent iteration. 
Of particular interest to us is the notion of a \emph{lexicographic representation} of a trace language. 
Let $<$ be a total order on $\Sigma$.
Then, $<$ induces a total, lexicographic order on $\Sigma^*$, defined as: 
\begin{inparaenum}[(i)]
\item for every $u \in \Sigma^*$, $\emptystring < u$, and
\item for every $u, v \in \Sigma^*$, $a, b \in \Sigma$, $au < bv$ if and only if either $a < b$, or $a = b$ and $u < v$. 
\end{inparaenum}
The \emph{lexicographic representation} of a trace language $T \subseteq M(\Sigma, C)$ is the set 
$Lex(T) = \set{u \in lin(T) \mid \forall v \in [u].~u \leq v}$. 
Members of $Lex(T)$ are called lexicographically \emph{minimal} words with respect to $C$ and $<$. 

\begin{remark} Rational trace languages were introduced as the least class of subsets of $M(\Sigma, C)$ closed under the usual set operations of union, product and Kleene star. Recognizable trace languages were introduced as the subset of $M(\Sigma, C)$ recognizable by finite automata. Unambiguous trace languages admit an algebraic definition in terms of unambiguous union, product and Kleene star, and are also called \emph{1-sequential} or \emph{deterministic} trace languages. 
For the purposes of exposition, we treat the characterizations of each class in relation to regular string languages directly as definitions. 
For a comprehensive treatment, we refer the reader to \cite{DBLP:books/ws/95/DR1995}. 
\end{remark}

\paragraph{Equivalence, Inclusion and Multiplicity Equivalence} 
Let $A_1, A_2$ be two finite automata recognizing the languages $L_1, L_2 \subseteq \Sigma^*$. 
The $\EQUIVALENCE(\Sigma, I)$ problem asks whether $[L_1] = [L_2]$. 
The $\INCLUSION(\Sigma, I)$ problem asks whether $[L_1] \subseteq [L_2]$. 
Last, the $\sharp \EQUIVALENCE(\Sigma,I)$ problem asks whether $\sharp L_1 = \sharp L_2$, where $\sharp L$ is a function $M(\Sigma,I) \rightarrow \Nat$ such that for every trace $t \in M(\Sigma, I)$, $\sharp L(t) = \card{\set{x \in L \mid [x] = t}}$, where $\Nat$ denotes the set of natural numbers. 
Thus, $\sharp L_1 = \sharp L_2$ when every trace has the same number of representatives in $L_1$ and $L_2$. 
Observe that if $\sharp L_1 = \sharp L_2$, then $[L_1] = [L_2]$, whereas the converse only holds when both $[L_1]$ and $[L_2]$ are unambiguous trace languages witnessed by $L_1$ and $L_2$ respectively.

\subsection{Global protocols and their realizations as trace languages} 

Observe that a global protocol naturally defines a rational trace language over concurrent alphabet $(\AlphSync, \interswapConc)$, where $\interswapConc$ is defined as:

\begin{definition} 
	\label{def:interswapConc}
	Given a synchronous alphabet $\AlphSync$ defined over a set of participants $\Procs$ and a set of message labels $\MsgVals$, the symmetric binary relation $\interswapConc$ is defined as 
	$\set{(x_1, x_2) \mid \Procs(x_1) \inters \Procs(x_2) = \emptyset}$. 
\end{definition}

Clearly, $\interswapConc$ defines an equivalence relation over $\AlphSync^*$, capturing pairs of execution traces that are equivalent up to a choice of scheduling. 
We prove that a global protocol's synchronous semantics exactly coincides with the closure of its language under $\interswapConc$ and that the semantics of SCSMs is closed under $\interswapConc$. 

\begin{restatable}{proposition}{protocolSemanticsEqualsClosureEquiv} 
	\label{prop:equiv-equivalence}
	\label{prop:global-rational-trace-language} 
	Let $\protocol$ be a global protocol. 
	Then, it holds that 
$\langsync(\protocol) = [\lang(\protocol)]_\interswapConc$ and 
	$\langsync(\protocol)$ is a rational trace language over the concurrent alphabet $(\AlphSync, \interswapConc)$, existentially defined by $\lang(\protocol)$. 
\end{restatable}

\begin{proposition} 
	\label{prop:scsm-interswap-closed} 
	\label{prop:scsm-recognizable-trace-language} 
	Let $\CSMabb{A}$ be an SCSM and let $w$ be a trace of $\CSMabb{A}$. 
	Then, for all $w' \in [w]_\interswapConc$, it holds that $w'$ is a trace of $\CSMabb{A}$ and 
	$\lang(\CSMabb{A})$ is a recognizable trace language over the concurrent alphabet $(\AlphSync, \interswapConc)$. 
\end{proposition} 

\subsection{Canonical SCSMs} 
It is well-known that if a global protocol is realizable, its \emph{canonical} candidate realization can always serve as existential witness to realizability~\cite{DBLP:journals/tse/AlurEY03,extendedversion,DBLP:conf/ppdp/GiustoLU25}. 
We define the \emph{canonical SCSM} of a global protocol and state facts about its canonicity. 
The results in this section are not novel and are included only to keep our presentation self-contained.

\begin{definition}[Canonical SCSM]
	\label{def:canonical-scsm}
	Let $\protocol = (Q, \AlphSync, \delta, q_0, F)$ be a protocol. 
	For each participant $\procA \in \Procs$,  we define a finite automaton $\CSMabb{A}_{\procA} =
	\bigl(
	Q_{\procA},
	\AlphSync_\procA,
	\delta_{\procA},
	q_{0, \procA},
	F_{\procA}
	\bigr)$, where 
	\begin{itemize}
		\item $ \delta(q, a) \is
		\set{s' \in Q
			\mid
			\exists s \in q,
			s \xrightarrow{a} \xrightarrow{\emptystring}\mathrel{\vphantom{\to}^*} s' \in \delta_\procA
		},
		$
		for
		every
		$q \subseteq Q$ and
		$a \in \AlphSync_{\procA}$\item $q_{0, \procA} \is
		\set{s \in Q \mid
			s_0 \xrightarrow{\emptystring} \mathrel{\vphantom{\to}^*} s \in \delta_\procA}$, \item $Q_{\procA} \is \lfp_{\set{q_{0,\procA}}}^\subseteq \lambda S.\, S \union \set{ \delta(s,a) \mid s \in S \land a \in \AlphSync_{\procA}} \setminus \set{\emptyset}$
		, 
		\item $\delta_{\procA} \is \restrict{\delta}{Q_{\procA} \times \AlphSync_{\procA}}$, and 
		\item $F_{\procA} \is
		\set{q \in Q_{\procA}
			\mid q \inters F \neq \emptyset}
            \enspace .$
		\end{itemize} 
		Note that $\delta_\procA$ enforces that $\CSMabb{A}_\procA$ is deterministic and reduced. 
		We define $\CSMabb{A}$ to be the canonical SCSM of $\protocol$. 
\end{definition}

Canonical SCSMs satisfy the following local language property. 

\begin{restatable}{lemma}{canonicalSCSMLocalLanguageProperty}
	Let $\protocol$ be a global protocol and $\CSMabb{A}$ be its canonical SCSM. 
	Then, for every $\procA \in \Procs$, the finite automaton $\CSMabb{A}_\procA$ satisfies: \\
	\begin{inparaenum}[(i)]
		\item \label{def:canonicity-finite-words} 
		$\lang(\CSMabb{A}_\procA) = \lang(\protocol) \wproj_{\AlphSync_\procA}$, and
		\item \label{def:canonicity-prefixes} 
		$\pref(\lang(\CSMabb{A}_\procA)) = \pref(\lang(\protocol) \wproj_{\AlphSync_\procA})$.
	\end{inparaenum}
\end{restatable}

It is easy to see that by construction, the canonical SCSM of a global protocol includes at least the global protocol's specified behaviors. 

\begin{lemma}\label{lm:canonical-implementation-incl-protocol-semantics} 
	Let $\protocol$ be a global protocol and let $\CSMabb{A}$ be its canonical SCSM. 
	Then, $\langsync(\protocol) \subseteq \lang(\CSMabb{A})$. 
\end{lemma}

We relate global protocols to their canonical SCSMs below. 
\begin{lemma} 
	\label{lm:realizable-iff-canonical-scsm} 
	Let $\protocol$ be a global protocol. 
	Then, $\protocol$ is realizable if and only if its canonical SCSM realizes it. 
\end{lemma} 

\cref{lm:realizable-iff-canonical-scsm} and \cref{lm:canonical-implementation-incl-protocol-semantics} together imply that deciding  realizability of a global protocol amounts to checking deadlock-freedom of the canonical SCSM, and checking language inclusion of the canonical SCSM into global protocol semantics. 
Observe further that each SCSM is in fact a finite automaton, due to the absence of (unbounded) message channels. 
We define the language preserving finite automaton of an SCSM below. 

\begin{definition}[Canonical synchronous product]
	\label{def:sclts-to-dfa} 
	Let $\protocol$ be a global protocol and let $\CSMabb{A}$ be its canonical SCSM. 
	The canonical synchronous product of $\protocol$, denoted $Prod(\protocol)$ is the synchronous product of $\CSMabb{A}$. 
\end{definition} 

We characterize realizability of a global protocol in terms of its canonical synchronous product.

\begin{lemma} 
	\label{lm:realizable-minimal} 
	Let $\protocol$ be a global protocol.
	Then, $\protocol$ is realizable if and only if $Prod(\protocol)$ is reduced and satisfies $\lang(Prod(\protocol)) \subseteq \langsync(\protocol)$. \end{lemma} 

A keen reader may have observed that in light of \cref{lm:realizable-minimal}, to decide realizability, it suffices to decide recognizability.

Unfortunately, we have the following, due to Sakarovitch~\cite{DBLP:conf/latin/Sakarovitch92}: 
\begin{theorem}
	Recognizability of rational trace languages is undecidable in general. 
\end{theorem}

The reduction chosen is from the universality problem for two-tape one-way non-deterministic automata over a unary alphabet, which Ibarra showed to be undecidable in~\cite{DBLP:journals/jacm/Ibarra78}. 
The constructed rational trace language in fact satisfies a stronger property: either it is \emph{empty}, or it is not recognizable. 
This stronger property establishes universality if and only if recognizability. 
The construction, in the form of an NFA, can be encoded using a synchronous alphabet with four participants.
Note, however, that the same construction does not work for undecidability of realizability. 
This is because encoding invalid solutions to a Post Correspondence Problem set requires concurrent iteration, which immediately renders the protocol unrealizable. In \cref{sec:undecidability}, we detail a different reduction that avoids concurrent iteration.

In the following two sections, we combine \cref{lm:realizable-minimal} with our novel trace-theoretic characterizations in \cref{prop:global-rational-trace-language} and \cref{prop:scsm-recognizable-trace-language} to reduce the realizability problem to different problems over trace languages. 
All proofs can be found in \cref{app:decidability}. 

\subsection{The transitive case} 
\label{subsec:transitive} 

Trace languages are classically studied through the lens of their concurrency relation.  
In particular, it was shown that classes of rational trace languages over concurrent alphabets, whose concurrency relation
is \emph{transitive}, enjoy many decidability~properties. 

In this section, we transfer these decidability properties to global protocols. 
In \cref{subsubsec:case-analysis-transitive}, we provide an exhaustive case analysis of synchronous alphabets with a transitive concurrency relation. In \cref{subsubsec:dec-transitive}, we propose an algorithm for deciding realizability of global protocols with a transitive concurrent alphabet, and in \cref{subsubsec:complexity-transitive}, we analyze its complexity. 
While our algorithm is based on an existing decidability proof~\cite{DBLP:journals/ita/AalbersbergW86} and is thus not novel, to the best of our knowledge, no upper bound has previously been provided for the equivalence and inclusion problems over rational trace languages. 
In \cref{subsubsec:flaw-transitive}, we analyze a flaw in the original proof claiming the decidability of rational trace language equivalence in the transitive setting. 

\subsubsection{Subsets of $\mathit{\AlphSync}$ with transitive $\interswapConc$} 
\label{subsubsec:case-analysis-transitive} 
Let $\protocol$ be a global protocol, and let $\AlphSync_{\protocol}$ be the subset of $\AlphSync$ whose symbols appear in $\protocol$. 
Let $\interswapConc_{\protocol}$ be the concurrency relation $\interswapConc$ (\cref{def:interswapConc}) restricted to $\AlphSync_{\protocol}$. 
The transitivity of $\interswapConc_{\protocol}$ depends both on the size of $\Procs$ and $\MsgVals$, as well as the communication topology in $\protocol$. 
Note that in the special case when $\interswapConc_{\protocol}$ is empty and thus trivially transitive, realizability is decidable in PSPACE by checking equivalence and reducedness of finite state automata directly.

When $\card{\Procs} \leq 3$, $\interswapConc_\protocol$ is empty. 
When $\card{\Procs} \geq 4$, we first case split on whether the communication topology is \emph{complete}, \ie where all pairs of participants communicate in the global protocol. 

When $\card{\Procs} = 4$ and complete, we further case split on the size of~$\MsgVals$. 
Let $\Procs = \set{\procA, \procB, \procC, \procD}$. 
If $\card{\MsgVals} = 1$, then $\interswapConc_\protocol$ is transitive:
\[
	\set{(\msgFromTo{\procA}{\procB}{\val}, \msgFromTo{\procC}{\procD}{\val}), 
			 (\msgFromTo{\procA}{\procC}{\val}, \msgFromTo{\procB}{\procD}{\val}),
			 (\msgFromTo{\procA}{\procD}{\val}, \msgFromTo{\procB}{\procC}{\val})}
    \enspace .
\]
When a pair of participants exchanges more than one message value, the following pattern breaks transitivity: 
\[
	\msgFromTo{\procA}{\procB}{\val_1} \interswapConc_{\protocol} \msgFromTo{\procC}{\procD}{\val_1} \interswapConc_{\protocol} \msgFromTo{\procA}{\procB}{\val_2} \not\interswapConc_{\protocol}
	\msgFromTo{\procA}{\procB}{\val_1}
    \enspace .
\]

\noindent When $\card{\Procs} \geq 5$, $\interswapConc_\protocol$ is not transitive since the pattern above is contained. 

In the case that the communication topology is \emph{not complete}, when $\card{\Procs} = 4$ and $\card{\MsgVals} > 1$, transitive alphabets are characterized by the following property. 

\begin{proposition}
	Let $\Gamma_\protocol$ be a subset of $\AlphSync$, defined over $\card{\Procs} = 4$ and $\card{\MsgVals} > 1$. 
	Then, $\interswapConc_\protocol$ is transitive if and only if  $\card{\set{\msgFromTo{\procA}{\procB}{\val} \mid \val \in \MsgVals} \inters  \AlphSync_{\protocol}} \geq 1$ implies 
	$\card{\set{\msgFromTo{\procC}{\procD}{\val} \mid \val \in \MsgVals} \inters  \AlphSync_{\protocol}} = 0$ for all $\procA, \procB, \procC, \procD \in \Procs$ with 
	$\set{\procA, \procB} \inters \set{\procC, \procD} = \emptyset$. 
\end{proposition}
Note that if two distinct pairs $(\procA, \procC)$ and $(\procA, \procB)$ 
both exchange more than one message, then $\AlphSync_\protocol$ can only additionally contain $\msgFromTo{\procA}{\procD}{\val}$ and $\msgFromTo{\procB}{\procC}{\val}$. 
Moreover, there must exist one ``coordinator'' participant that communicates with all other participants in $\Procs$. 

Finally, we mention
two 
important instances of non-transitive~$\interswapConc_\protocol$:  the following synchronous alphabet with $\card{\Procs} = 4$ and $\card{\MsgVals} = 2$ is isomorphic to the ``boomerang'' concurrency alphabet used in a variety of undecidability results, for example to show that the inclusion problem for rational trace languages is undecidable in~general~\cite{DBLP:books/ws/95/DR1995}: 
\[
\AlphSync_{abc} = \set{\msgFromTo{\procA}{\procB}{\val_1}, \msgFromTo{\procC}{\procD}{\val_1}, \msgFromTo{\procA}{\procB}{\val_2}}
\enspace .
\]
The following synchronous alphabet with $\card{\Procs} = 4$ and $\card{\MsgVals} = 2$ is isomorphic to the ``rectangle'' concurrency alphabet 
that suffices to encode any concurrent alphabet of the form $(\Sigma_1 \cup \Sigma_2, \Sigma_1 \times \Sigma_2)$~\cite{DBLP:conf/icalp/BertoniMS82}, whose rational trace languages can be expressed as multi-tape one-way non-writing automata~\cite{DBLP:journals/jcss/FischerR68}: 
\[
\AlphSync_{abcd} = \set{\msgFromTo{\procA}{\procB}{\val_1}, \msgFromTo{\procC}{\procD}{\val_1}, \msgFromTo{\procA}{\procB}{\val_2}, \msgFromTo{\procC}{\procD}{\val_2}}
\enspace .
\]

\subsubsection{Decidability}
\label{subsubsec:dec-transitive}
Our decidability result is stated below. 

\begin{restatable}{theorem}{realizabilityTransitiveDecidable}
	\label{thm:realizability-transitive-decidable} 
	Synchronous realizability is decidable for all global protocols $\protocol$ over $\AlphSync$ whose concurrency relation $\interswapConc_{\protocol}$ is transitive.
\end{restatable}

With $\card{\Procs} \leq 3$ having an empty concurrency relation, and with our undecidability reduction in \cref{sec:undecidability} requiring 5 participants, the remaining  $\card{\Procs} = 4$ case is partially addressed by the following corollary of \cref{thm:realizability-transitive-decidable}. 

\begin{corollary}
	Synchronous realizability is decidable for all global protocols with $\Procs = \set{\procA, \procB, \procC, \procD}$ when either $\card{\MsgVals} = 1$, or $\card{\MsgVals} > 1$ and $\card{\set{\msgFromTo{\procA}{\procB}{\val} \mid \val \in \MsgVals} \inters  \AlphSync_{\protocol}} \geq 1$ implies 
	$\card{\set{\msgFromTo{\procC}{\procD}{\val} \mid \val \in \MsgVals} \inters  \AlphSync_{\protocol}} = 0$ for all $\procA, \procB, \procC, \procD \in \Procs$ with 
	$\set{\procA, \procB} \inters \set{\procC, \procD} = \emptyset$. 
\end{corollary}

\newcommand{\inclproblem}{\mathcal{I}}
\subsubsection{Deciding synchronous realizability in 2-\EXPTIME}
\label{subsubsec:complexity-transitive}
Our complexity result is stated below. 
\begin{restatable}{theorem}{transitiveExptime} 
	\label{thm:concurrent-exptime} 
    Synchronous realizability is decidable for all global protocols $\protocol$ over $\AlphSync$ whose concurrency relation $\interswapConc_{\protocol}$ is transitive in 2-\EXPTIME in the size of $\protocol$, and 3-\EXPTIME in the size of $\AlphSync$.
\end{restatable} 

Our upper bound is obtained via a reduction to the following problem $\inclproblem$: given a transitive concurrent alphabet $(\Sigma, C)$ and $L_1, L_2 \subseteq \Sigma^*$, does $\Alang{L_1}{C} \subseteq \Elang{L_2}{C}$?

Our proposed algorithm to solve $\inclproblem$ is excavated from the proof of Lemma 2.7 \cite{DBLP:journals/ita/AalbersbergW86}, stating that a concurrent alphabet is of type $\exists = \forall$ if and only if it is transitive. 
A concurrent alphabet is of type $\exists = \forall$ iff the class of existentially regular trace languages over $(\Sigma, C)$ and the class of universally regular trace languages over $(\Sigma, C)$ coincide.  

We restate and prove Lemma 2.7 from \cite{DBLP:journals/ita/AalbersbergW86} below.

\begin{restatable}{lemma}{twoPointSeven}
	A concurrent alphabet $(\Sigma, C)$ is of type $\exists = \forall$ if and only if $C$ is transitive. 
\end{restatable} 

In the following, we analyze the complexity of a series of constructions used to prove Lemma 2.7 in \cite{DBLP:journals/ita/AalbersbergW86}. 
We focus on describing the constructions algorithmically and analyzing their complexity. 
For proofs of correctness, we refer the reader to the original paper. 

Recall that the Parikh image of a language $L$ over 
$\Delta = \set{a_1, 
\ldots, a_m}$, denoted $\psi(L)$, is the image of $L$ under the mapping $\psi \from \Delta^* \to \mathbb{N}^m$, defined as 
$\psi(w) = (|w|_{a_1}, 
\ldots, |w|_{a_m})$, 
where $|w|_a$ denotes the number of occurrences of symbol $a$ in $w$. 

The following lemma constructs for each existentially regular trace language an equivalent universally regular trace language over the same concurrent alphabet. 

\begin{restatable}{lemma}{twoPointEight}
	\label{lm:my-lemma-2.8}
	Let $(\Sigma, C)$ be a concurrent alphabet and let $L$ be a regular string language over $\Sigma$ given by an NFA $M$. 
	Then, a regular string language $K$ such that $\Elang{L}{C} = \Alang{K}{C}$ can be constructed with time and size complexity doubly exponential in the size of $M$ and triply exponential in the size of $\Sigma$. 
\end{restatable}

The next lemma generalizes the correspondence between individual languages to sets of languages over the same concurrent alphabet. 

\begin{restatable}{lemma}{twoPointFour}
	\label{lm:my-lemma-2.4}
	Let $(\Sigma, C)$ be a concurrent alphabet and let $L_1, \ldots, L_n$ be $n$ regular string languages over $\Sigma$, given by $n$ NFAs $M_1, \ldots, M_n$. 
	Then, $n$ regular string languages $K_1, \ldots, K_n$ such that for every subset of indices $s \subseteq \set{1, \ldots, n}$, 
	$\Elang{\bigcup_{i \in s} L_i}{C} = \Alang{\bigcup_{i \in s} K_i}{C}$ can be constructed in time and size complexity doubly exponential in the size of $M_i$ and triply exponential in the size of $\Sigma$. 
\end{restatable}

The following lemma describes the complexity of the automata construction in Lemma 2.5 of \cite{DBLP:journals/ita/AalbersbergW86}. 
Intuitively, given a regular language $L$ over the disjoint union of alphabets $\Sigma_1, \Sigma_2$, represented by an NFA $A$, we construct an NFA over a new alphabet of triples containing pairs of states of $A$, in addition to $\set{1,2}$. 
The new NFA accepts words in $L$ in an alternating fashion, reading maximal substrings of symbols from $\Sigma_1$ and $\Sigma_2$ in turn. 
Every triple $(s, s', i)$ in the new alphabet represents the Parikh image of the language defined by setting the initial state to $s$ and the unique final state to $s'$ in $A$, intersected with $\Sigma_i^*$. 
In other words, $(s, s', i)$ represents the words that can reach $s'$ from $s$ in $A$ on only symbols from $\Sigma_i$. 
Since all symbols in $\Sigma_i$ commute, there is a bijection between semi-linear sets and languages over $\Sigma_i^*$. 
The construction is described in detail in the following lemma, with an analysis of its complexity. 

\begin{restatable}{lemma}{twoPointFive}
	\label{lm:my-lemma-2.5}
	Let $A = (Q, \Sigma, \delta, q_0, F)$ be a complete DFA recognizing $L$. Let $(\Sigma_1, \Sigma_2)$ be a partition of $\Sigma$. 
	Let $\Delta_1 = Q \times Q \times \set{1}$, $\Delta_2 = Q \times Q \times \set{2}$, and $\Delta = \Delta_1 \cup \Delta_2$. 
	Let $K' = \set{(s_1, s_2, i_1) \ldots (s_n, s_{n+1}, i_n) \in \Delta \mid n \geq 1, s_1 = q_0, s_{n+1} \in F \text{ and } i_j \neq i_{j+1} \text{ for } 1 \leq j \leq n-1}$.  
	If $\emptystring \in L$, let $K = K' \cup \set{\emptystring}$, otherwise let $K = K'$. 
	Let $\eta : \Delta^* \rightarrow 2^{\Sigma^*}$ be defined as 
	$\eta((s,s',i)) = \set{w \in \Sigma_i^+ \mid \delta(s, w) = s'}$. 
	A DFA recognizing $K$ can be constructed with a constant number of states, and $\Delta$ is polynomial in the size of $Q$. 
	For each $(s, s', i) \in \Delta$, a finite automaton representation of $\eta((s, s', i))$ can be constructed using $|Q|^2$ states. 
	Moreover, $K$ and $\eta$ satisfy the following properties: 
    \begin{inparaenum}[(1)]
        \item $K$ is a regular string language, 
        \item $\eta$ is a regular substitution, 
        \item $\eta(K) = L$, 
        \item $\eta(\Delta_1) \subseteq \Sigma_1^+$ and $\eta(\Delta_2) \subseteq \Sigma_2^+$, and 
        \item $K \inters \Delta^*\Delta_1\Delta_1\Delta^* = \emptyset$ and $K \inters \Delta^*\Delta_2\Delta_2\Delta^* = \emptyset$. 
    \end{inparaenum}
\end{restatable}

The following lemma describes how to combine two trace languages over disjoint concurrent alphabets, and intuitively provides the induction step for the proof of \cref{thm:forall-incl-exists-exptime}. 

\begin{restatable}{lemma}{twoPointSix}
	\label{lm:my-lemma-2.6}
	Let $(\Sigma_1, C_1), (\Sigma_2, C_2)$ be two concurrent alphabets of type $\exists = \forall$, where $\Sigma_1$ and $\Sigma_2$ are disjoint. 
	Let $L$ be a regular language over $\Sigma_1 \cup \Sigma_2$ given by a DFA $A = (Q, \Sigma_1 \cup \Sigma_2, \delta, q_0, F)$. 
	Then, a DFA~$M$ such that $\Elang{L}{C_1 \cup C_2} = \Alang{\lang(M)}{C_1 \cup C_2}$ can be constructed in time and size doubly exponential in $Q$ and triply exponential in $\Sigma$. \end{restatable}

\begin{restatable}{theorem}{forallInclExistsExptime}
	\label{thm:forall-incl-exists-exptime}
	The following problem is decidable in 2-\EXPTIME in the size of 
$M_1$ and $M_2$ 
    and 3-\EXPTIME in the size of $\Sigma$: \\
	Input: Concurrent alphabet $(\Sigma, C)$ with $C$ transitive, and DFAs $M_1, M_2$. \\
	Problem: $\Alang{\lang(M_1)}{C} \subseteq \Elang{\lang(M_2)}{C}$? 
\end{restatable}

As a corollary, we obtain the following. 
\begin{corollary} 
	The synchronous verification problem for global protocols is decidable in 2-\EXPTIME in the size of $\protocol$ and the input SCSM, and 3-\EXPTIME in the size of $\Sigma$. 
\end{corollary}

\subsubsection{Error in Lemma 5.2 of \cite{DBLP:conf/icalp/BertoniMS82}}
\label{subsubsec:flaw-transitive}
Lemma 5.2 proposes the following construction for deciding equivalence of rational trace languages over transitive concurrency alphabets. 
The decidability result is stated in Theorem 5.2 of \cite{DBLP:conf/icalp/BertoniMS82}. 
Let $(\Sigma, C)$ be a concurrency alphabet such that $C$ is transitive, and let $\Sigma$ be partitioned into $\Sigma_1, \ldots, \Sigma_n$ by~$C$. 
Let $M = (Q, \Sigma, \delta, q_0, F)$ be a DFA. 
Define 
$N(M) = 
\set{
    (q',j,q'') 
        \mid 
    \exists~\alpha, \sigma', z, \sigma'', k \in \Sigma^* \st \sigma', \sigma'' \notin \Sigma_j \land z \in \Sigma_j^* \text{ for some } 1 \leq j \leq n, \delta(q_0, \alpha\sigma') = q'' \land \delta(q'', \sigma''k) \in F
}$. 
Intuitively, $N(M)$ captures pairs of states $(q', q'')$ together with an alphabet $\Sigma_j$ such that there exists a run of $M$ that reads a ``pure'' word from $\Sigma_j^*$, sandwiched by two symbols not from $\Sigma_j$, that reaches a final state.

Let $L_{q'q''}$ be defined as the language accepted by the DFA obtained from $M$ by setting $q'$ to be the initial state, and $q''$ to be the only final state. 
Let $P_{q'jq''}$ be defined as the Parikh language defined by $[L_{q'q''} \inters \Sigma_j^*]_C$. 
Because $L_{q'q''} \inters \Sigma_j^*$ is permutation-closed, it is uniquely represented by its Parikh image. 

Lemma 5.2 claims that for two finite automata $M_1$ and~$M_2$ recognizing languages $L_1$ and $L_2$, the following statements are equivalent: 
\begin{enumerate}[(a)]
	\item \label{lemma5.2-9-a}
    $[L_1]_C = [L_2]_C$ 
	\item \label{lemma5.2-9-b}
    $\forall j, q_1', q_2', q_1'', q_2''.~(q_1',j,q_1'') \in N(M_1) \land (q_2',j,q_2'') \in N(M_2)$ \\ 
    $\implies P_{q_1'jq_1''} = P_{q_2'jq_2''}$ 
\end{enumerate} 

Let $\Sigma = {a_1, a_2, b}$ and let $C$ be $\set{(a_1, a_2), (a_2, a_1)}$. 
Consider $L = b \cdot \set{a_1, a_2} \cdot b \cdot \set{a_1, a_3}$.
Then, $[L]_C = [L]_C$ is trivially satisfied, yet the automaton recognizing $L$ does not satisfy (b). 
Consider further 
$L_1 = b \cdot \set{a_1, a_2} \cdot b$, and 
$L_2 = b \cdot \set{a_1, a_2} \cdot b \cdot \set{a_1, a_2}$. 
$[L_1]_C \neq [L_2]_C$, yet (b) holds for the automata recognizing $L_1$ and $L_2$ respectively. 
Thus, neither does 
\ref{lemma5.2-9-a}
imply
\ref{lemma5.2-9-b},
nor 
\ref{lemma5.2-9-b}
imply
\ref{lemma5.2-9-a}. 

\subsection{The unambiguous case} 
\label{subsec:unambiguous}

Unambiguous trace languages are originally defined algebraically using unambiguous restrictions of the $\cup$, $\cdot$ and $\Kleenestar$ operators defined over trace languages~\cite{DBLP:journals/iandc/BruschiPS94}. 
\begin{definition}
	\label{def:unambiguous-regex} 
	Let $(\Sigma, C)$ be a concurrent alphabet and let $T_1, T_2 \subseteq M(\Sigma, C)$. 
	Then $T_1 \sqcup T_2$ is defined as $T_1 \cup T_2$ iff $T_1 \inters T_2 = \emptyset$, 
	$T_1 \circ T_2$ is defined as $T_1 \cdot T_2$ iff for all $x, z \in T_1$ and $y, w \in T_2$, 
	$xy = zw \implies x = z \land y = w$, 
	and finally $T^@$ is defined as $T^*$ iff $T$ is the basis of a free submonoid $T^*$ of $M(\Sigma, C)$. 
	The class of unambiguous trace languages over $(\Sigma, C)$ is the least class of subsets of $M(\Sigma, I)$ containing finite sets and closed with respect to $\sqcup$, $\circ$ and $(-)^@$.
\end{definition}

This definition 
naturally also defines unambiguous global protocols over $(\AlphSync, \interswapConc)$. 
However, when we switch to the view of trace languages defined by regular languages, unambiguous trace languages are defined as the class of languages such that there \emph{exists} a regular language that unambiguously represents it. 
Given a trace language $[L]_C$ defined by a regular language $L$, it may be unambiguous even if $L$ does not unambiguously represent it. 
Sakarovitch~\cite{DBLP:journals/tcs/Sakarovitch87} showed that the class of rational trace languages and unambiguous trace languages over a concurrency alphabet coincide if and only if the concurrency relation is transitive. 
This means that for every rational trace language represented as $[L]_C$, one can effectively construct $L'$ such that $L'$ unambiguously represents $[L]_C$. 
However, such a construction is not effective in the general case. 

Nonetheless, if a trace language is defined by a regular language that unambiguously represents it, this guarantees its membership in the class of unambiguous rational trace languages. 
In the following, we say that $\protocol$ is an \emph{unambiguous global protocol}, or $\protocol$ is \emph{unambiguous}, iff $\protocol$ unambiguously represents $[\lang(\protocol)]_\interswapConc$, namely for every trace $t \in M(\AlphSync, \interswapConc)$, $t$ has either 0 or 1 representatives in~$\lang(\protocol)$. 

In \cref{subsubsec:choice-unambiguous}, we show that previously proposed restrictions on \emph{branching choice} in global protocols are sufficient conditions for unambiguity. 
In \cref{subsubsec:dec-unambiguous}, we show that realizability is decidable for unambiguous global protocols. 
The key fact we prove is that for a recognizable trace language given as a finite state automaton recognizing all of its linearizations, an unambiguous representative can be computed in linear time. 
In \cref{subsubsec:complexity-unambiguous}, we state the complexity of realizability, invoking an upper bound on the multiplicity equivalence problem for unambiguous trace languages.

\subsubsection{Choice restrictions} 
\label{subsubsec:choice-unambiguous} 

In the asynchronous setting, mixed choice is key to showing that realizability is undecidable in general~\cite{DBLP:journals/tcs/Lohrey03, DBLP:phd/dnb/Stutz24}. 
For protocols with a more restrictive form of choice, called \emph{sender-driven choice}~\cite{DBLP:conf/concur/MajumdarMSZ21}, realizability is co-NP-complete~\cite{DBLP:journals/pacmpl/LiSWZ25}. 
Naturally, the notion of sender-driven choice was investigated in the synchronous setting, and \citet{NODBLPyet:10.1145/3756907.3756918} showed that synchronous realizability for sender-driven choice global types is decidable in PSPACE. 
They further relaxed sender-driven choice to \emph{commutation-determinism}, and showed that realizability remains in PSPACE. 
While sender-driven choice requires a dedicated sender among all outgoing transitions from a single state in a protocol, commutation-determinism only requires that every pair of states share a participant. We state all three definitions below.  

\begin{definition}
Let $\protocol = (Q, \AlphSync, \delta, q_0, F)$ be a global protocol. 
We say $\protocol$ is \emph{mixed choice} iff there exists $q \in Q$ with $q \xrightarrow{a} q'$ and $q \xrightarrow{b} q''$, such that $a \interswapConc b$. 
We say $\protocol$ is \emph{non-mixed choice} iff for all $q \in Q$ with $q \xrightarrow{a} q'$ and $q \xrightarrow{b} q''$, it holds that $a \not\interswapConc b$. 
We say $\protocol$ is \emph{sender-driven} iff for all $q \in Q$ with $q \xrightarrow{a} q'$ and $q \xrightarrow{b} q''$, $a = \msgFromTo{\procA}{\procB}{\val}$ and $b = \msgFromTo{\procA}{\procC}{\val'}$ for some $\procA, \procB, \procC, \val, \val'$. 
\end{definition}

Note that the definition of commutation-deterministic~\cite{DBLP:conf/ppdp/GiustoLU25} coincides with non-mixed choice, of which sender-driven choice is a subclass. 
We show that non-mixed choice places a protocol 
into the class of unambiguous global protocols. 
\begin{restatable}{lemma}{nonMixedChoiceUnambiguous}
	Every non-mixed choice global protocol defines an unambiguous trace language over the concurrent alphabet $(\AlphSync, \interswapConc)$. 
\end{restatable}

However, unambiguous global protocols are not necessarily non-mixed choice. Thus, non-mixed choice is only a sufficient~condition, as indicated by the strict containment in \cref{fig:venn-diagram}. 

\begin{remark} 
	There exist unambiguous mixed choice global protocols, such as the following: 
	\[
		\set{\msgFromTo{\procA}{\procB}{\val} \cdot 
		\msgFromTo{\procC}{\procD}{\val} \cdot 
		\msgFromTo{\procA}{\procB}{\val_1}}
		\cup 
		\set{\msgFromTo{\procC}{\procD}{\val} \cdot 
		\msgFromTo{\procA}{\procB}{\val} \cdot 
		\msgFromTo{\procA}{\procB}{\val_2}}
		\enspace . 
	\]
\end{remark}

\subsubsection{Decidability} 
\label{subsubsec:dec-unambiguous} 

Our decidability result is stated below. 
\begin{restatable}{theorem}{unambiguousDecidable} 
	\label{thm:unambiguous-dec} 
	Synchronous realizability of unambiguous global protocols is decidable. 
\end{restatable} 

A key observation towards \cref{thm:unambiguous-dec} is that all recognizable trace languages are unambiguous trace languages. 
The proof is immediate via the characterization of recognizable trace languages in terms of regular lexicographic representations. 
We additionally provide a bound on the size of the unambiguous representative. 
\begin{restatable}{lemma}{recIsUnambiguous}
	\label{lm:rec-is-unambiguous} 
	Let $T$ be a recognizable trace language over concurrent alphabet $(\Sigma, C)$.
	Then, $T$ is an unambiguous trace language. 
	Moreover, given an automaton that recognizes $lin(T)$, an unambiguous regular language generating $T$ can be computed in linear time. 
\end{restatable}

Given \cref{lm:rec-is-unambiguous}, when the global protocol unambiguously represents its trace language, protocol fidelity amounts to checking the equivalence of two unambiguous rational trace languages. 
Unlike in \cref{subsec:transitive}, where we made use of the definition of canonicity to reduce protocol fidelity to language inclusion, here we reduce realizability to equivalence directly. 
Indeed, inclusion cannot be used, because of the following negative result. 

\begin{theorem}[\textnormal{Theorem\,5.4.16 of} \cite{DBLP:books/ws/95/DR1995}]
	Inclusion is undecidable for unambiguous rational trace languages in general. 
\end{theorem}

Equivalence, in contrast, is decidable thanks to the fact that two unambiguous rational trace languages are multiplicity equivalent if and only if they are equivalent. 
Multiplicity equivalence can in turn be decided by a reduction to the equality of two rational formal power series in one variable, for more details see~\cite{DBLP:books/ws/95/DR1995}. 

\subsubsection{Complexity} 
\label{subsubsec:complexity-unambiguous} 
The complexity of the synchronous realizability and verification problems for unambiguous global protocols is stated below. 
\begin{restatable}{theorem}{unambiguousExptime}
	Synchronous realizability of unambiguous global protocols is decidable in \EXPTIME. 
\end{restatable} 
\begin{corollary} 
	The synchronous verification problem for unambiguous global protocols is decidable in \EXPTIME. 
\end{corollary}

\begin{remark} 
    \label{rem:recognizable-PSPACE}
	We conclude this section by mentioning the case of \emph{recognizable} trace languages. 
When $\lang(\protocol) = lin([\lang(\protocol)]_\interswapConc)$, realizability is decidable in PSPACE, and this case coincides with the class of \emph{commutation-closed} global types from \cite{DBLP:conf/ppdp/GiustoLU25}, or $I$-closed automata from \cite{DBLP:journals/ita/Zielonka87}. 
	When $\protocol$ is a star-connected regular expression, realizability is decidable in EXPSPACE thanks to the exponential construction from~ \cite{DBLP:conf/mfcs/MuschollP99}. 
\end{remark} 

\subsection{Undecidability} 
\label{sec:undecidability}

In this section, we show that synchronous realizability is undecidable in general. 
We adapt the undecidability proof of asynchronous realizability for 
global protocols, given in Theorem 9.1 of ~\cite{DBLP:phd/dnb/Stutz24}, which is in turn inspired by Theorem 3.4 of~\cite{DBLP:journals/tcs/Lohrey03}. 
In the proof, the acceptance problem for Turing machines is reduced to realizability. 
The proof constructs a 
protocol $\protocol_{\TM}$, that is asynchronously realizable if and only if there is no accepting computation for Turing machine~$\TM$. 
We show that in fact, $\protocol_{\TM}$ is synchronously realizable if and only if $\protocol_{\TM}$ is asynchronously realizable. 
Additionally, the encoding $\protocol_{\TM}$ is \sinkfinal, meaning that final states have no outgoing transitions. 

\begin{definition}[Sink finality] 
Let $A = (Q, \Sigma, \delta, q_0, F)$ be a finite automaton.
We say $A$ is \sinkfinal if for every final state $q \in F$, there is no $l \in \Sigma$ and $q' \in Q$ 
with $q \xrightarrow{l} q'$.
\end{definition}

We introduce the required preliminaries for the asynchronous setting, and defer full definitions and proofs to \cref{app:undecidability}. 

\subsubsection*{Alphabets and CSMs for asynchronous setting.}

A synchronous event can be split into a send and receive event, yielding \emph{asynchronous events}: 
$
\AlphAsync_{\procA} = 
\set{\snd{\procA}{\procB}{\val} \mid \procB \in \Procs,\; \val \in \MsgVals }
\union
\set{\rcv{\procB}{\procA}{\val} \mid \procB \in \Procs,\; \val \in \MsgVals }
$
and 
$
\AlphAsync = \Union_{\procA \in \Procs}
$. 
The homomorphism 
$
\SyncToAsync(\msgFromTo{\procA}{\procB}{\val})
\is
\snd{\procA}{\procB}{\val}. \,
\rcv{\procA}{\procB}{\val}
$
maps the synchronous alphabet to its asynchronous counterpart. 
The event $\snd{\procA}{\procB}{\val}$ denotes participant $\procA$ sending a message $\val$ to $\procB$,
and $\rcv{\procB}{\procA}{\val}$ denotes participant $\procA$ receiving a message $\val$ from $\procB$.
An (asynchronous) \emph{communicating state machine} (CSM) $\CSMabb{A} = \CLTS{\CSMabb{A}}$ consists of a DFA $\CSMabb{A}_\procA$ for each participant over $\AlphAsync_\procA$, connected by peer-to-peer FIFO channels. 
It is straightforward to obtain a finite automaton $\makeAsync{A}$ over $\AlphAsync_\procA$ from a finite automaton $A$ over~$\AlphSync_\procA$: one applies $\SyncToAsync(\hole)$ to the label of each transition and projects the result onto the asynchronous alphabet of $\procA$.
This extends to SCSMs as expected: for a SCSM $\CSMabb{B}$, its asynchronous CSM is denoted by $\makeAsync{\CSMabb{B}}$. 
We say a CSM or SCSM $\CSMabb{A}$ is \sinkfinal if $\CSMabb{A}_\procA$ is \sinkfinal for each $\procA \in \Procs$. 
The asynchronous semantics for CSMs is denoted by $\lang(\CSMabb{A})$ for CSM~$\CSMabb{A}$.  
The asynchronous semantics of global protocols is defined in terms of reorderings induced by peer-to-peer FIFO communication, and is denoted by $\langasync(\protocol)$ for protocol $\protocol$. 

\begin{definition}[Asynchronous Realizability Problem] \label{def:async-realizability}
A CSM $\CSMabb{A}$ is an asynchronous realization of protocol $\protocol$ if the following two properties hold:
\begin{inparaenum}[(i)]
\item \label{def:lts-implementability-protocol-fidelity}
\emph{protocol fidelity}: $\lang(\CSMabb{A}) = \langasync(\protocol)$, and
\item \label{def:lts-implementability-deadlock-freedom}
\emph{deadlock freedom}: $\CSMabb{A}$ is deadlock-free.
\end{inparaenum}
Given a protocol $\protocol$, the \emph{asynchronous realizability problem} asks whether there exists a CSM $\CSMabb{A}$ that is an asynchronous realization of $\protocol$. 
\end{definition}

We want to prove that $\protocol_{\TM}$ is synchronously realizable if and only if $\protocol_{\TM}$ is asynchronously realizable. 
The ``if'' direction follows directly from \cite[Thm.\,5.2]{NODBLPyet:10.1145/3756907.3756918}. 
The ``only if'' direction is significantly more involved. 

The (regular) language of encoding $\protocol_{\TM}$ satisfies an interesting property: every sent message is immediately acknowledged. 
We formalize this notion as follows. 

\begin{definition}
	We say a word $w \in \AlphSync^*$ is \emph{immediately acknowledging} if 
	$w = w_1 \ldots w_n$ such that  
	if $w_i = \msgFromTo{\procA}{\procB}{\val}$ for some $1 \leq i < n$,
	then $w_{i+1} = \msgFromTo{\procB}{\procA}{\val}$. 
	We say a finite automaton $A$ over $\AlphSync$ is \emph{immediately acknowledging} if every word in $\lang(A)$ is immediately acknowledging. 
An SCSM $\CSMabb{A}$ is \emph{immediately acknowledging} if
    for every $\procA \in \Procs$,  
    $\CSMabb{A}_\procA$ is immediately acknowledging. 
\end{definition}

Being immediately acknowledging means that for any sent message, its receiver has the power to prevent its sender from taking further actions, in both the asynchronous and synchronous setting. 
While this is a useful property, immediately acknowledging alone is not a sufficient condition for a synchronously realizable protocol to also be asynchronously realizable. 
The protocol depicted in \cref{fig:receiver-power-sync} serves as a counterexample: the key discrepancy is that in the asynchronous setting, sends are always enabled, whereas in the synchronous setting, both the sender and the receiver in a message exchange have the power to veto the exchange. 
Protocols that do not provide receivers this additional veto power are asynchronously realizable if they are synchronously realizable. 
We capture this intuition in a sufficient condition below, which states that for every message that can be sent from a reachable CSM configuration, either the receiver can receive the message immediately, or can receive any message value from the sender, after executing an immediately acknowledging sequence of events. 
We provide an intermediate definition towards defining our sufficient~condition. 

\begin{definition}Let $A$ be an immediately acknowledging finite automaton over $\AlphSync_\procA$ and let $q$ be a state in $A$. 
	We call a state $q'$ \emph{send-ack-reachable} if there is 
    $w \in (\Union_{\procA \neq \procB \in \Procs, \val \in \MsgVals} \set{
        \msgFromTo{\procA}{\procB}{\val} \cat 
        \msgFromTo{\procB}{\procA}{\val} 
        })^*$
	such that 
	$q \xrightarrow{w}\starred q'$. 
	We say $q$ is maximal wrt.\ send-ack-reachability if there is no state~$q'$ that is send-ack-reachable from $q$ with $q' \neq q$. 
\end{definition}

We define the sufficient condition on SCSMs rather than protocols, as the proof reasons about a canonical SCSM of a protocol. 
In the following definition, $\MsgVals(\CSMabb{A}, \procA, \procB)$ denotes all message values sent from $\procA$ to $\procB$ in a CSM $\CSMabb{A}$. 

\begin{definition}[Condition \condX]
	An SCSM $\CSMabb{A}$ satisfies Condition~\condX if 
	for every 
	$\vec{q} \in \reach(\CSMabb{A})$  and 
	$\procA, \procC \in \Procs$, 
	one of the following holds 
	if $\vec{q}_\procA \xrightarrow{\msgFromToNS{\procA}{\procC}{\val}}$: 
	\begin{enumerate}[label=(\alph*)]
		\item $\vec{q}_\procC \xrightarrow{\msgFromToNS{\procA}{\procC}{\val}}$, or 
		\item for every maximal send-ack-reachable $q'_\procC$, \\
		it holds that 
		$s'_\procC \xrightarrow{\msgFromToNS{\procA}{\procC}{\val}} $
		for every $m \in \MsgVals(\CSMabb{A}, \procA, \procC)$. 
	\end{enumerate}
\end{definition}

We show that Condition \condX is a sufficient condition on the realization of a protocol to preserve realizability from the synchronous to the asynchronous setting. 
Interestingly, Condition \condX is only required to show deadlock freedom. 

\begin{restatable}{lemma}{syncRealizationAndConditionXyieldAsyncRealization}
	\label{lm:sync-realization-and-Condition-X-yield-async-realization}
	Let $\protocol$ be an immediately acknowledging protocol and let $\CSMabb{B}$ be a sink-final immediately acknowledging SCSM that synchronously realizes $\protocol$. 
    Then, $\langasync(\protocol) = \lang(\CSMabb{\makeAsync{B}})$. 
    If $\CSMabb{B}$ satisfies Condition $\condX$, 
	then $\CSMabb{\makeAsync{B}}$ is asynchronously deadlock-free. 
\end{restatable}

Last, we show that the canonical SCSM for $\protocol_\TM$ satisfies Condition~\condX when $\protocol_\TM$ is realizable. 

\begin{restatable}{lemma}{canonicalRealizationSatConditionX}
	\label{lm:canonical-realization-sat-Condition-X}
	If its canonical SCSM $\CSMabb{A}$ synchronously realizes $\protocol_\TM$, then $\CSMabb{A}$ satisfies Condition~\condX. 
\end{restatable}

\cref{lm:sync-realization-and-Condition-X-yield-async-realization,lm:canonical-realization-sat-Condition-X}
are used to show the ``only if'' direction of the equivalence, namely that $\protocol_\TM$ is synchronously realizable implies that $\protocol_\TM$ is asynchronously realizable. With this equivalence, we obtain our main result. 

\begin{restatable}{theorem}{syncRealizabilityUndecidable}
	\label{thm:sync-realizability-undecidable}
	Synchronous realizability of sink-final global protocols is undecidable. 
\end{restatable}

Because the encoding can be represented as a high-level message sequence chart, we obtain the following: 

\begin{corollary}
Synchronous realizability of high-level message sequence charts is undecidable. 
\end{corollary}

     \section{Type system} \label{sec:type-system}

In this section, we present a type system that checks mixed choice processes against SCSMs, which thus act as interface between global protocols and processes.
In fact, this makes our type system applicable to any SCSM, which does not need to have been obtained from or even be related to a global protocol. 
Our type system adapts the one from \cite{DBLP:conf/esop/StutzD25} 
(which is inspired by \cite{DBLP:journals/pacmpl/ScalasY19}) from the asynchronous to the synchronous setting and extends it by mixed choice. It is the first type system with support for mixed choice processes that addresses session interleaving and delegation.

\begin{definition}
Processes
and process definitions are defined by the following grammar:
\begin{grammar}
 c \is
        x
    |   s[\procA]
\\
 \prefixPi \is 
        c[\procB] \sendOp \labelAndMsg{l}{c}
    |   c[\procB] \recOp \labelAndVar{l}{y}
    \\
 P \is
        \zero
    |   P_1 \parallel P_2
    |   (\restr s \hasType \CSMabb{A}) \, P
    |   \MixCh_{i \in I} \prefixPi_i \seq P_i
    |   \pn{Q}[\vec{c}]
    \\
\Defs \is
    \bigl(\pn{Q}[\vec{x}] =
    \MixCh_{i \in I} c[\procB_i] \, \prefixPi_i \seq P_i\bigr); \; \Defs
    | \emptystring
\end{grammar}

A term $c$ is either a session endpoint $s[\procA]$ or a variable $x$ (which at runtime will resolve to a session endpoint). 
As is standard, $\zero$ represents explicit termination and $\parallel$ is parallel composition. 
The term $(\restr s \hasType \CSMabb{A})$ restricts a session~$s$ where SCSM $\CSMabb{A}$ provides the intended session behavior, an annotation solely used for type checking and not for reductions. 
We have mixed choice ($\MixCh$) so processes can choose between 
sending ($\sendOp \labelAndMsg{l}{c}$) or 
receiving ($\recOp \labelAndVar{l}{y}$) a message. 
We assume $\card{I} > 0$ and omit $\MixCh$ if $\card{I} = 1$.   
With $\pn{Q}[\vec{c}]$, we specify a process call: $\pn{Q}$ is the process identifier and $\vec{c}$ its parameters.  
The process definitions are given by $\Defs$: for instance 
$
    \pn{Q}[\vec{x}] =
    \MixCh_{i \in I} \prefixPi_i \seq P_i
$, providing guarded process definitions. 
With $\Defs(\pn{Q}, \vec{c})$, we denote the unfolding of the process definition, \ie 
$\bigl( \MixCh_{i \in I} \prefixPi_i \seq P_i \bigr) [\vec{c} / \vec{x}]$. 
\end{definition}

Structural congruence is a standard way to define if processes shall be considered equivalent, \eg parallel composition with $\zero$ can be removed (or added). 

\begin{definition}
Structural congruence, denoted by $\congr$, is defined using the following rules:\footnote{Note that $\interswapConc$ is also used to denote the concurrency relation in \cref{sec:realizability}. The intended interpretation of $\interswapConc$ is always clear from context.}

\begin{itemize}
 \item $P_1 \parallel P_2
        \congr
        P_2 \parallel P_1$
 \item $(P_1 \parallel P_2) \parallel P_3
        \congr
        P_1 \parallel (P_2 \parallel P_3)$
 \item $P \parallel \zero
        \congr
        P$
 \item $(\restr s \hasType \CSMabb{A}) \, (\restr s' \hasType \CSMabb{B}) \, P
        \congr
        (\restr s' \hasType \CSMabb{B}) \, (\restr s \hasType \CSMabb{A}) \, P$
 \item $(\restr s \hasType \CSMabb{A}) \, (P_1 \parallel P_2)
        \congr
        P_1 \parallel (\restr s \hasType \CSMabb{A}) \, P_2$,
        if $s$ is not free in $P_1$
\end{itemize}

We define structural precongruence $\precongr$ for processes as the smallest precongruence relation that
includes $\congr$ and
    $(\restr s \hasType \CSMabb{A}) \, \zero
    \precongr
    \zero$.
\end{definition}

When we let processes reduce, we allow descending into the structure of the process by the means of reduction contexts.

\begin{definition}
We define reduction contexts as follows:
\begin{grammar}
    \redContext \is
        \redContext \parallel P
    |   P \parallel \redContext
    |   (\restr s \hasType \CSMabb{A}) \, \redContext
    |   [\,] \enspace .
\end{grammar}

We define the reduction rules in \cref{fig:reduction-rules}. 
The rule \procReductionExc lets a composition of two processes exchange a message: the label is the same and a value is sent and received. 
The rule \procReductionProcName unfolds a process definition. The rule \procReductionContext allows us to descend for reductions using contexts. 
Last, \procReductionCongr allows us to consider structurally precongruent processes for reductions.
\end{definition}

\begin{figure}[t]
  \adjustfigure[\scriptsize]
  \begin{mathpar}

  \inferrule*[right=\procReductionExc]{
          I \inters J = \emptyset \\
          k \in I  \\
          k' \in J  \\
\prefixPi_k = s[\procA][\procB] \sendOp \labelAndMsg{l_k}{v_k} \\
\prefixPi_{k'} = s[\procB][\procA] \recOp \labelAndVar{l_{k'}}{y_{k'}} \\
          l_k = l_{k'} 
  }{
          \MixCh_{i \in I} \prefixPi_i \seq P_i
          \parallel
          \MixCh_{j \in J} \prefixPi_j \seq P_j
              \redto
          P_k
          \parallel
          P_{k'}[v_k / y_{k'}]
  }

  \inferrule*[right=\procReductionProcName]{
          \Defs(\pn{Q}, \vec{c})
          \parallel
          P
          \redto
          P'
  }{
          \pn{Q}[\vec{c}]
          \parallel
          P
          \redto
          P'
  }

  \inferrule*[right=\procReductionContext]{
          P \redto P'
  }{
          \redContext[P] \redto \redContext[P']
  }

  \inferrule*[right=\procReductionCongr]{
          P_1 \precongr P'_1
          \\
          P'_1 \redto P'_2
          \\
          P'_2 \precongr P_2
  }{
          P_1 \redto P_2
  }

\end{mathpar}
  \caption{Reduction rules for processes.}
  \label{fig:reduction-rules}
\end{figure}

\begin{definition}
The process definition typing context $\typingContextOne$ is a function that maps process identifiers to the types of their parameters, \ie 
$\typingContextOne \from \ProcIds \to \vec{L}$.
A \emph{syntactic typing context} is given by this grammar:
\begin{grammar}
    \typingContextTwo \is
        \typingContextTwo, s[\procA] \hasType L
    |   \typingContextTwo, x \hasType L
    |   \emptyset 
        \enspace .
\end{grammar}
A syntactic typing context is a \emph{typing context} if every element has at most one type and we only consider the latter in this work. 
Typing contexts are considered equivalent up to reordering and, thus, may be treated as mappings.
Notation $\set{\typingContextTwo_i}_{i \in I}$ denotes splitting $\typingContextTwo$ into $\card{I}$ typing contexts.
\end{definition}

\begin{figure}[t]
  \adjustfigure[\small]
  \begin{mathpar}
\inferrule*[right=\procTypingProcDefEmpty]{
}{
      \types
      \emptystring \hasType \typingContextOne
  }

  \inferrule*[right=\procTypingProcDef]{
      \typingContextOne
      \typingContextCat
      \vec{x} \hasType \vec{L}
      \types
      P \\
\typingContextOne(\pn{Q}) = \vec{L} 
      \\
      \types
      \Defs \hasType \typingContextOne
  }{
      \types
      (\pn{Q}[\vec{x}] = P); \Defs \hasType \typingContextOne
  }

  \inferrule*[right=\procTypingProcName]{
      \typingContextOne(\pn{Q}) = \vec{L} }{
      \typingContextOne
          \typingContextCat
          \vec{c}  \hasType \vec{L}
\types
          \pn{Q}[\vec{c}]
  }

  \inferrule*[right=\procTypingZero]{
}{
      \typingContextOne
      \typingContextCat
      \emptyset
          \types
      \zero
  }

  \inferrule*[right=\procTypingEnd]{
      \typingContextOne
      \typingContextCat
      \typingContextTwo \types P \\
\EndState(q)
}{
      \typingContextOne
      \typingContextCat
      c \hasType q,
      \typingContextTwo
          \types
      P
  }

  \inferrule*[right=\procTypingParallel]{
      \typingContextOne
          \typingContextCat
          \typingContextTwo_1
          \types
          P_1 \\
\typingContextOne
          \typingContextCat
          \typingContextTwo_2
          \types
          P_2 \\
  }{
      \typingContextOne
          \typingContextCat
          \typingContextTwo_1,
          \typingContextTwo_2
          \types
          P_1 \parallel P_2
  }

  \inferrule*[right=\procTypingMixCh]{
\delta(q) =
      \set{(\msgFromTo{\procA}{\procB_i}{\labelAndType{l_i}{L_i}}, q_i) \mid i \in I}
      \dunion
      \set{(\msgFromTo{\procB_j}{\procA}{\labelAndType{l_j}{L_j}}, q_j) \mid j \in J}
      \\
\meta{\forall i \in I \st}
      \typingContextOne \typingContextCat
          \typingContextTwo , c \hasType q_i,
          \set{c_j \hasType L_j}_{j \in I\setminus \set{i}}
          \types P_i 
      \\
      \meta{\forall j \in J \st}
      \typingContextOne \typingContextCat
          \typingContextTwo,
          c \hasType q_j,
          \set{c_i \hasType L_i}_{i \in I}, 
          y_j \hasType L_j
          \types P_j \\
  }{
      \typingContextOne
      \typingContextCat
      \typingContextTwo,
      c \hasType q,
      \set{c_i \hasType L_i}_{i \in I}
      \\
          \types
      \MixCh_{i \in I} c[\procB_i] \sendOp \labelAndMsg{l_i}{c_i} \seq P_i
      + 
      \MixCh_{j \in J} c[\procB_j] ? \labelAndVar{l_j}{y_j} \seq P_j
  }

\inferrule*[right=\procTypingRestr]{
\vec{q} \in \reach(\CSMabb{A})
      \\
      \typingContextTwo_s =
          \set{s[\procA] \hasType \vec{q}_\procA}_{\procA \in \ProcsOf{\CSMabb{A}}}
      \\
\typingContextOne
          \typingContextCat
          \typingContextTwo,
          \typingContextTwo_s
      \types
      P
}{
      \typingContextOne
          \typingContextCat
          \typingContextTwo
          \types
          (\restr s \hasType \CSMabb{A})\, P
  }
  \end{mathpar}
  \caption{
    Typing rules for processes.
}
  \label{fig:proc-typing}
\end{figure}

In any system, for two labels and types 
$\labelAndType{l_1}{L_1}$ and 
$\labelAndType{l_2}{L_2}$, we assume that 
$l_1 = l_2$ implies that $L_1 = L_2$. 
Other synchronous MST frameworks ensure this when projecting global types, through the merge operator. 
Some also circumvent the underlying problem by splitting branching decisions from sending and receiving values. 
In practice, such types can be common base types (like booleans, integers) or a channel endpoint $s[\procA]$ with the goal of delegation. 
It is well-known how to add base types to such type systems (and orthogonal) so, following \cite{DBLP:journals/pacmpl/ScalasY19}, we focus on channel endpoints.

\begin{definition}
We give the typing rules for processes in \cref{fig:proc-typing}. 
For any state $q$, $\EndState(q)$ holds if $q$ is final and has no outgoing transitions. The rules \procTypingProcDefEmpty and \procTypingProcDef are used for process definitions and allow us to reason about the types of their parameters. 
This information is required to type process definitions, as in \procTypingProcName. 
With \procTypingZero, we type the process $\zero$ but only with an empty second typing context, which can be achieved through \procTypingEnd: one can drop $c \hasType q$ if $\EndState(q)$ holds. 
The rule \procTypingMixCh lets us type processes for which we partition the choices into sending and receiving next. 
With \procTypingParallel, parallel compositions of processes and their typing contexts are split and then typed individually. 
When using \procTypingRestr, type bindings for a session are added and the remaining process should be typed using also those. 
\end{definition}

\begin{definition}
We define the reductions for typing contexts in a way that 
they precisely match the semantics of SCSMs: 
{ \small 
\begin{mathpar}
\inferrule*[right=\typingReductionMixCh]{
    q_1
        \xrightarrow{\msgFromToNS{\procA}{\procB}{\labelAndType{l}{L}}}
    q_3
    \\
    q_2
        \xrightarrow{\msgFromToNS{\procA}{\procB}{\labelAndType{l}{L}}}
    q_4
}{
    s[\procA] \hasType q_1,
    s[\procB] \hasType q_2,
        \typingContextTwo
        \redto
    s[\procA] \hasType q_2,
    s[\procB] \hasType q_4,
        \typingContextTwo
}
\end{mathpar}
}
\end{definition}

\begin{proposition}
\label{prop:typing-reduction-cong}
\label{lm:typing-reduction-cong} 
Let $\typingContextTwo_1, \typingContextTwo'_1$, and $\typingContextTwo_2$ be typing contexts. 
Whenever 
    $
        \typingContextTwo_1
        \redto
        \typingContextTwo'_1
    $
holds,
it holds that 
    $
        \typingContextTwo_1,
            \typingContextTwo_2
        \redto
        \typingContextTwo_1',
            \typingContextTwo_2
    $.
\end{proposition}

Now, we can state (and prove) our main properties. 
\emph{Subject reduction} says that typing contexts can match any step a process (typed using this context) can take and the new typing context can type the new process. 

\begin{restatable}[Subject Reduction]{theorem}{subjectReduction}
\label{thm:subject-reduction}
Let $P$ be a process with a set of sessions~$\SessionName$.
If 
\begin{enumerate}[label=\textnormal{(\arabic*)}]
 \item \label{lm:subject-reduction-assumption-1} $\types \Defs \hasType \typingContextOne$,
 \item \label{lm:subject-reduction-assumption-2} $\typingContextOne
           \typingContextCat
           \typingContextTwo
       \types
       P$ 
       with
        $\typingContextTwo = \hat{\typingContextTwo}, \set{\typingContextTwo_s}_{s \in \SessionName}$,
\item \label{lm:subject-reduction-assumption-3} for every $s \in \SessionName$, it holds that
       there is $
\vec{q}\in \reach(\CSMabb{A}_s)
   $ such that $\typingContextTwo_s =
        \set{s[\procA] \hasType \vec{q}_\procA}_{\procA \in \ProcsOf{\CSMabb{A}}}, and 
        $
\item \label{lm:subject-reduction-assumption-4} $P \redto P'$,
\end{enumerate}
then there exists
$\typingContextTwo'$
with
$\typingContextTwo
\redto
\typingContextTwo'$
or
$\typingContextTwo
=
\typingContextTwo'$
such that
$\typingContextOne
    \typingContextCat
    \typingContextTwo'
\types
P'$.
\end{restatable}

The second property, \emph{progress}, considers the opposite direction: if a typing context types a process and can take a step, then the process can match this step. 
As is standard, we restrict progress to single sessions because of potential interdependencies.
Orthogonal techniques have been studied to overcome this restriction which are out of scope for this work \cite{DBLP:conf/concur/Kobayashi06,DBLP:journals/mscs/CoppoDYP16}.

To achieve this restriction, we define $\typesSFd$ to be $\types$ but without the rule \procTypingRestr. 
We define $\typesSFs$ for processes as follows:
{ \small 
\begin{mathpar}
\inferrule*[right=\procTypingRestr ']{
\meta{\forall \procA \in \ProcsOf{\CSMabb{A}} \st}
    \forall c \hasType q' \in \typingContextTwo_\procA \st
    \EndState(q')
    \\
    \vec{q} \in \reach(\CSMabb{A})
    \\
\meta{\forall \procA \in \ProcsOf{\CSMabb{A}} \st}
        \typingContextOne
        \typingContextCat
        \typingContextTwo_\procA,
        s[\procA] \hasType \vec{q}_\procA
        \typesSFd
        Q_\procA
}{
    \typingContextOne
        \typingContextCat
        \set{\typingContextTwo_\procA}_{\procA \in \ProcsOf{\CSMabb{A}}}
\overset{\vec{q}}{\typesSFs}
        (\restr s \hasType \CSMabb{A})\,
        (\Parallel_{\procA \in \ProcsOf{\CSMabb{A}}} Q_\procA)
}
\end{mathpar}
}

\vspace{-2ex}
\noindent Note the use of 
$\overset{\vec{q}}{\typesSFs}$
instead of $\types$ 
in the conclusion:  
we augment the typing rule with an easy way to obtain the SCSM configuration.

\begin{restatable}[Progress]{theorem}{progressThm}
\label{lm:session-fidelity}
Let $\CSMabb{A}$ be a sink-final deadlock-free SCSM. Let $P$ be a process. 
We assume that 
\begin{enumerate}[label=\textnormal{(\arabic*)}]
 \item \label{sf-cond-1}
        $\typesSFs \Defs \hasType \typingContextOne$,
 \item \label{sf-cond-2}
        $
            \typingContextOne
                \typingContextCat
                \typingContextTwo
                \overset{\vec{q}}{\typesSFs}
                (\restr s \hasType \CSMabb{A})\,
                P
        $, and
 \item \label{sf-cond-3}
       $
       \vec{q}
       \rightarrow
       \pvec{q}'
       $
       for some
       $\pvec{q}'$. 
\end{enumerate}
Then, 
there is
$P'$
with
$P \redto P'$ such that
        $
            \typingContextOne
                \typingContextCat
                \typingContextTwo
\overset{\pvec{q}'}{\typesSFs}
                (\restr s \hasType \CSMabb{A})\,
                P'
        $.
\end{restatable}
The proofs for both results can be found in \cref{app:type-system}. 
There, we also provide a proof sketch for deadlock freedom of typed processes as a consequence of progress and subject reduction. 
     \section{Related work}
\label{sec:related-work}

We structure our discussion of related work around the two most closely related formalisms: HMSCs and MSTs. 

\subsection{High-level message sequence charts}
High-level message sequence charts (HMSCs) have been studied extensively in academia 
~\cite{DBLP:conf/sdl/MauwR97,
DBLP:conf/ac/GenestMP03,DBLP:conf/acsd/GenestM05,DBLP:conf/concur/GazagnaireGHTY07,DBLP:journals/tosem/RoychoudhuryGS12} 
and were defined in an industry standard~\cite{z120-standard}. 
Interestingly, most HMSC research considered asynchronous communication models. 
In terms of realizability, two problems have been considered: 
weak and safe realizability~\cite{DBLP:journals/tcs/Lohrey03,DBLP:journals/tcs/AlurEY05}.
Both require protocol fidelity; the former permits deadlocks while the latter does not. 
Lohrey showed both realizability problems to be undecidable in general \cite{DBLP:journals/tcs/Lohrey03}. To make the problem tractable, various restrictions for HMSCs have been identified. 
We discuss restrictions that restore decidability of (safe) realizability without modifying the protocol's behavior. 
Bounded/regular HMSCs \cite{DBLP:conf/mfcs/MuschollP99,DBLP:conf/concur/AlurY99} always have realizations with bounded channels, yielding finite-state systems.
This is achieved by requiring the communication topology of every loop to be strongly connected. 
When this condition is weakened to weakly connected, one obtains c-HMSCs \cite{DBLP:conf/stacs/Morin02} or globally-cooperative HMSCs~\cite{DBLP:journals/jcss/GenestMSZ06}. 
For both regular and globally-cooperative HMSCs, \citet{DBLP:journals/tcs/Lohrey03} proved the realizability problem to be \EXPSPACE-complete.
\citet{DBLP:journals/tcs/Lohrey03} also introduced so-called $\mathcal{I}$-closed HMSCs where $\mathcal{I}$ is the independence relation $\interswapConc$ (\cref{def:interswapConc}) lifted to atomic MSCs, and proved the realizability problem to be \PSPACE-complete.

As mentioned in \cref{sec:introduction}, Mazurkiewicz trace theory has proven fruitful for obtaining solutions to other decision problems over HMSCs. 
Most works transfer results about recognizable trace languages. 
\citet{DBLP:conf/mfcs/MuschollP99} show that the race and confluence problem for MSC graphs is undecidable in general, but EXPSPACE-complete for locally synchronized graphs corresponding to recognizable trace languages. 
Asynchronous FIFO protocol semantics are not definable as the closure of an irreflexive, symmetric binary alphabetic relation in general. 
The authors circumvent this by prohibiting the communication pattern 
$(\msgFromTo{\procA}{\procB}{\val})^*\!$, which intuitively requires counting unreceived messages in a FIFO channel. 
Recognizability is in turn enforced by requiring MSCs to be loop-connected, which is similar to star-connectedness of regular expressions that was shown to characterize recognizability~\cite{DBLP:books/ws/95/DR1995}. 
Muscholl and Peled additionally show that loop-connectedness of an NFA is co-NP-complete. 
Interestingly, we mirror the observation made in \cite{DBLP:conf/mfcs/MuschollP99} that complexity results in trace theory are few and far between. 

A related problem of deciding whether the language recognized by a finite automaton is $I$-closed was shown to be PSPACE-complete in \cite{DBLP:journals/tcs/PeledWW98}. 
The $I$-closure problem asks, given a finite automaton over some alphabet and a concurrency relation, whether the language of the automaton is closed under the equivalence relation induced by the concurrency relation. 
\citet{DBLP:conf/stacs/Morin02} studies star-connected HMSCs, a restriction that again ensures recognizability, and yields that realizability is immediately decidable. 
\citet{DBLP:conf/dlt/GenestMK04} prove a Kleene theorem connecting monadic second-order logic, (asynchronously) communicating state machines and globally-cooperative compositional MSC graphs. 
In light of the above results, we transfer decidability results from recognizable trace languages to realizability in our setting (\cref{subsec:unambiguous}), but do not claim it as a novel contribution. 
In our work, we studied aspects in some sense orthogonal to recognizability: unambiguity, and the transitivity of the independence alphabet. 
In particular, the connection to unambiguous trace languages was inspired by syntactic choice restrictions inherited by multiparty session types from binary session types, called \emph{directed choice}, that was later generalized to \emph{sender-driven choice} by \citet{DBLP:conf/concur/MajumdarMSZ21} and to \emph{commutation-determinism} by \citet{DBLP:conf/ppdp/GiustoLU25}.

\subsection{Multiparty session types} MST frameworks guarantee the absence of communication errors using types. 
The first MST framework targeted asynchronously communicating processes~\cite{DBLP:conf/popl/HondaYC08}, but a synchronous variant followed soon afterwards \cite{DBLP:journals/entcs/BejleriY09}. 
MST frameworks typically address the realizability problem through a combination of syntactic restrictions on global and local types, a projection operator that computes local types from global types, and a type system that further rules out undesirable behaviors from local implementations. 
This hybrid approach trades completeness for efficiency, soundly approximating realizability via the notion of projectability. 
It is difficult and not meaningful to compare algorithms that solve different problems such as projectability and realizability. 
We focus our discussion on 
\begin{inparaenum}
\item comparing our results to other decision procedures for MST realizability, 
\item \label{rw:classical-MST}
      top-down MST frameworks,
\item \label{rw:bottom-up}
      bottom-up MST frameworks with mixed choice, and 
\item \label{rw:synt-MST-choraut}
      synthetic MSTs and choreography automata. 
\end{inparaenum}
Note that 
(\ref{rw:classical-MST})
is the classical MST approach while 
(\ref{rw:bottom-up})
and 
(\ref{rw:synt-MST-choraut})
are more recent developments.

\paragraph{Deciding realizability} 
To the best of our knowledge, \cite{DBLP:conf/ppdp/GiustoLU25} is the only work that provides solutions to the synchronous realizability problem of global types. 
We discussed their contributions in \cref{sec:introduction}, explained how their fragments relate to ours in \cref{subsec:unambiguous}, and  illustrated their relation to our contributions in \cref{fig:venn-diagram}. 
For the asynchronous setting, \cite{DBLP:conf/cav/LiSWZ23} took a similar approach to our work: separating the concerns of realizability and type checking. 
They provided the first complete projection operator, deciding the realizability for sender-driven sink-final global protocols. 
The same authors later showed this problem to be co-NP-complete \cite{DBLP:journals/pacmpl/LiSWZ25}. Their CSM realizations can be type-checked with the type system proposed by \cite{DBLP:conf/esop/StutzD25}. 

\paragraph{Top-down MST frameworks}
In terms of graph structure, syntactic multiparty session types correspond to tree-shaped, \sinkfinal finite automata with loops only between leaves and their ancestors. 
\citet{DBLP:conf/esop/StutzD25} showed that only sink finality restricts expressivity. 
Additional restrictions on branching choice are imposed on global types, local types and processes. 
Our work shows that as far as decidability of realizability is concerned, the role of choice is to ensure unambiguity of the global protocol. 
Our local processes are liberated from choice restrictions. 
We limit our discussion to synchronous MST frameworks and refer the reader to \cite{DBLP:journals/csur/HuttelLVCCDMPRT16} and \cite{DBLP:books/sp/24/Yoshida24} for a survey. 
Most synchronous MST frameworks 
\cite{DBLP:journals/entcs/BejleriY09,DBLP:conf/icdcit/YoshidaG20,DBLP:conf/itp/TiroreBC23,DBLP:series/lncs/YoshidaH24,DBLP:journals/pacmpl/UdomsrirungruangY25,DBLP:conf/itp/EkiciKY25}
impose directed choice on global types: in each choice, there is a unique sender and receiver pair, with distinct messages across choices. 
On a local level, a matching restriction is imposed: processes can either send or receive, with a unique communication partner in each case. 
In such classical frameworks, a partial function called projection operator implicitly attempts to solve realizability problem, \ie if the projection for each participant is defined then the global type is realizable. 
Behind the scenes, a merge operator is used to ``merge'' branches if participants do not take part in a branching. 
Most of these frameworks employ the so-called plain merge operator which does not allow participants to have different behavior in branches if they are not part of the decision-making, prohibiting learning about branching later.  
\cite{DBLP:journals/pacmpl/UdomsrirungruangY25} employs the so-called full merge operator which improves the situation but cannot unfold recursion for merging for instance. 
It can still happen that realizable global types have an undefined projection and it was shown that the merge-based approach is inherently incomplete \cite{DBLP:conf/ecoop/Stutz23}. 
Note that all these works 
do not support sender-driven or mixed choice, neither for global types, local types, nor processes. 

\paragraph{Bottom-up MST framework with mixed choice}
While aforementioned works follow the classical ``top-down'' approach, \cite{DBLP:journals/pacmpl/ScalasY19} forgoes global types and establishes properties of local types through model checking, coined the ``bottom-up'' approach. 
Their type system only guarantees operational correspondence between local types and processes in the absence of a global type. 
Our type system indirectly takes inspiration from \cite{DBLP:journals/pacmpl/ScalasY19} but extends it with mixed choice and types processes against deadlock-free SCSMs obtained from a global protocol. 
\citet{DBLP:conf/lics/PetersY24} also forgo global types, but their type system supports mixed choice for local types and processes. 
However, their type system only types a single session and hence no session interleaving or delegation is considered. 
In session types, subtyping provides flexibility for how to implement the behavior in processes.
Usually, the folklore is that a sender can send fewer messages and a~receiver can receive more messages than specified. 
In the presence of generalized choice, this is not necessarily the case -- at least not in the asynchronous setting \cite{DBLP:conf/esop/LiSW24}. 
\citet{DBLP:conf/lics/PetersY24} provide subtyping that splits the mixed choice into a send and a receive part and applies the above folklore for each part individually. 
Following our decoupling discipline, our type system does not provide rules for subtyping, rather, a variant of subtyping that preserves the full protocol semantics is addressed via the synchronous verification problem.

\paragraph{Synthetic MSTs and Choreography Automata}
A recent work by~\citet{NODBLPyet:journals/corr/abs-2511-22692} forgoes local --not global-- types, and addresses a variant of the synchronous verification problem, type checking processes against a global type directly. 
Their type system only supports single sessions and as such, their processes are very close to our realization model of SCSMs, but they do not allow receivers to receive from different senders. 
Protocol semantics in \cite{NODBLPyet:journals/corr/abs-2511-22692} is defined plainly as the word language of a labeled transition system (LTS), and not as the closure of a regular language under an equivalence relation.   
Note that when the LTS is finite-state (and only then the type system is known to be decidable), this means that its semantics is simply a regular language. 
Thus, both the realizability and the verification problem can be decided (soundly and completely) in \PSPACE (see \cref{rem:recognizable-PSPACE})
-- for a more expressive realization model -- and is in fact also \PSPACE-hard. 
Subsequently, the realization can be used to type check local processes. 
The authors take a different, hybrid approach, however, soundly approximating the verification problem through a combination of well-behavedness conditions on global protocols and typing rule restrictions. 
One such condition requires the LTS to be closed under~$\interswapConc$, but only conditionally. 
Unfortunately, the realization model's semantics is closed under $\interswapConc$, so any LTS that is not closed under $\interswapConc$ is trivially unrealizable. 
\citet{DBLP:conf/coordination/BarbaneraLT20} introduced choreography automata as a global protocol specification, whose semantics, like \cite{NODBLPyet:10.1145/3756907.3756918}, are not closed under any relation. Hence choreography automata semantics are also regular languages, and the same characterizations for deciding realizability and verification apply. 
\citet{NODBLPyet:10.1145/3756907.3756918} approximate the realizability problem through well-sequencedness and well-branchedness conditions on choreography automata and observed that well-sequencedness entails closure under~$\interswapConc$.
Notably, the proposed realizability conditions were shown to be flawed for both the asynchronous \cite{DBLP:journals/corr/abs-2107-03984} and the synchronous~\cite{DBLP:journals/lmcs/BarbaneraLT23} setting.

     \section{Discussion}

We conclude by discussing possible extensions to two separate and orthogonal infinite settings. 

Note that all trace-theoretic results invoked in this work define languages as sets of finite words.  
The first extension concerns infinite semantics of global protocols and their realizations, \ie when $\langsync(\protocol) \subseteq \AlphSync^* \union \AlphSync^\infty$ for a finite $\AlphSync$.  
In this work, global protocol and SCSM semantics are subsets of $\AlphSync^*\!$.  
Global protocols were given an infinite semantics by the authors of  \cite{DBLP:conf/concur/MajumdarMSZ21}, which amounts to a Büchi acceptance condition with all states accepting. 
Independently, peer-to-peer communicating state machines give rise to a definition of infinite words, which are simply defined as all infinite execution traces. 
Li and Wies~\cite{DBLP:conf/itp/LiW25} showed that the infinite word semantics for global protocols introduced in \cite{DBLP:conf/concur/MajumdarMSZ21} and adopted in subsequent works~\cite{DBLP:conf/cav/LiSWZ23, DBLP:conf/esop/LiSW24, DBLP:conf/ecoop/Stutz23} in fact does not match the infinite word semantics of the intended realization model, and thus proposed a revised infinite word semantics. 
\citet{DBLP:conf/ecoop/Stutz23} showed that reduced global protocols enjoy guarantees about its infinite behavior ``for free'', a fact that transfers immediately to our setting.

The second extension concerns global protocols with infinitely many states and transitions, \ie when $Q$ and $\AlphSync$ are infinite. 
Such global protocols are typically represented symbolically using logical predicates, as studied in \cite{DBLP:journals/pacmpl/LiSWZ25, DBLP:conf/cav/LiSWZ25,DBLP:conf/ecoop/VassorY24,DBLP:journals/pacmpl/00020HNY20}. 
\citet{DBLP:journals/pacmpl/LiSWZ25, DBLP:conf/cav/LiSWZ25} decide asynchronous realizability for symbolic protocols using a semantic characterization that is in essence a (co-)reachability reduction. The generality of this characterization enables it to uniformly derive decision procedures for both finite and infinite-state global protocols. 
It is less clear how the ``direct'' approach taken by our work extends to  infinite-state protocols represented symbolically. 
On the other hand, our direct approach immediately gives us decision procedures and upper bounds for the \emph{verification} problem. 
We leave such an exploration to future work. 
 	
    \phantomsection\label{paper-last-page}

    \clearpage
    \appendix
    \section{Additional Material for \cref{sec:realizability} -- Decidability}
\label{app:decidability}

\canonicalSCSMLocalLanguageProperty*
\begin{proof}
	Both inclusions of Claim (\ref{def:canonicity-prefixes}) follow by induction on prefix length and the determinism of $\CSMabb{A}_{\procA}$. 
	Claim (\ref{def:canonicity-finite-words}) follows from \ref{def:canonicity-prefixes} and the definition of final states $F_\procA$. 
\end{proof}

\protocolSemanticsEqualsClosureEquiv*
\begin{proof}
	Recall that $\interswapProcs$ is defined as for all participants $\procA \in \Procs$, 
	$w_1 \wproj_{\AlphSync_\procA} = w_2 \wproj_{\AlphSync_\procA}$. 
	The only if direction follows from the fact that $\interswapConc$ preserves each participant's total ordering of events. 
	The if direction follows from the fact that each participant's total ordering of events uniquely defines a message sequence chart (MSC), so we have $msc(w_1) = msc(w_2)$, of which $w_1$ and $w_2$ are both linearizations. 
	The proof follows from the fact that all linearizations of the same MSC are related by $\interswapConc$. We refer the reader to \cite{DBLP:journals/tcs/Lohrey03} for definitions. 
\end{proof}

\nonMixedChoiceUnambiguous*
\begin{proof} 
	Let $\protocol = (Q, \AlphSync, \delta, q_0, F)$ be a global protocol. 
	Suppose by contradiction that a trace $t \in M(\AlphSync, \interswapConc)$ has two representatives \mbox{$w_1 \neq w_2$} in $\lang(\protocol)$. 
	Note that $\lang(\protocol)$ refers to the regular language defined directly by $\protocol$, and not the synchronous semantics of $\protocol$. 
	Then, there exists two runs of $\protocol$ accepting $w_1$ and $w_2$ respectively. 
	Because $w_1 \neq w_2$, there exists a state $s \in Q$ at which the accepting runs for each diverge, and each word takes a different transition. 
	Suppose w.l.o.g.\ that the transitions are labeled $\msgFromTo{\procA}{\procB}{\val}$ and \mbox{$\msgFromTo{\procA}{\procC}{\val'}$}. 
	Then, it follows that $w_1 \wproj_{\AlphSync_\procA} \neq w_2 \wproj_{\AlphSync_\procA}$. 
	However, $t = [w_1] = [w_2]$. 
	We reach a contradiction via \cref{prop:equiv-equivalence}. 
\end{proof}

\realizabilityTransitiveDecidable*
\begin{proof}
	Let $\protocol$ be a global protocol over $\AlphSync$ and let $Prod(\protocol)$ be its canonical synchronous product. 
	Deadlock-freedom of $Prod(\protocol)$, \ie checking whether $Prod(\protocol)$ is reduced, is decidable in PSPACE. 
	$[\lang(Prod(\protocol))]_\interswapConc = [\lang(\protocol)]_\interswapConc$ is decidable because $\AlphSync_{\protocol}$ is transitive. 
\end{proof}

\twoPointSeven*
\begin{proof}
	Because $C$ is transitive, it is the disjoint union of a finite number of complete concurrency relations.
	Let $C$ be partitioned into cliques $C_1, C_2, \ldots, C_n$ whose alphabets $\Sigma_1, \Sigma_2, \ldots, \Sigma_n$ form a partition of $\Sigma$. 
	The proof proceeds by induction on $n$, with Lemma 2.3 \cite{DBLP:journals/ita/AalbersbergW86} covering the base case, and Lemma 2.6 \cite{DBLP:journals/ita/AalbersbergW86} covering the inductive step. Lemma 2.3 states that all concurrent alphabets $(\Sigma, C)$ where $C$ is complete are of type $\exists = \forall$. 
	Lemma 2.6 states that if $(\Sigma_1, C_1)$ is of type $\exists = \forall$ and $(\Sigma_2, C_2)$ is of type $\exists = \forall$, and $\Sigma_1$ is disjoint from $\Sigma_2$, then $(\Sigma_1 \cup \Sigma_2, C_1 \cup C_2)$ is of type $\exists = \forall$. 
\end{proof}

\twoPointEight*
\begin{proof}
	We construct $K = \overline{\psi^{-1}(\psi(\Sigma^*) - \psi(L))}$.
	The Parikh image of a regular language represented by an NFA can be computed in polynomial time in the size of the NFA, and exponential time in the size of its alphabet, with the size of the semi-linear set's description matching the running time~\cite{DBLP:conf/lics/KopczynskiT10}. 
	Complementation of a semi-linear set can be achieved in doubly exponential time in the size of its representation~\cite{DBLP:journals/corr/abs-1708-06460}. 
	Given a semi-linear set, a regular language whose Parikh image it is can be constructed in polynomial time in the size of the semi-linear set. 
	Finally, complementation of regular languages represented as NFAs can be done in linear time in the size of the NFA. 
	For the proof that $\Elang{L}{C} = \Alang{L}{C}$, see Lemma 2.3 of \cite{DBLP:journals/ita/AalbersbergW86}. 
\end{proof}

\twoPointFour*
\begin{proof}
	First, we construct one $K'_s$ for each $s \subseteq \set{1, \ldots, n}$ such that 
	$\Elang{\bigcup_{i \in s} L_i}{C} = \Alang{K'_s}{C}$. 
	Since there are $2^n$ subsets of $\set{1, \ldots, n}$, and the union of two NFAs remains polynomial in their size, each $K'_s$ can be constructed in doubly exponential time in the size of the input NFA, altogether the $K'_s$'s can be constructed in doubly exponential time. 
	We define $K_i = \bigcap_{i \in s, s \subseteq \set{1, \ldots, n}} K_s$. 
	For the proof of correctness of the languages $K_i$ thus constructed, we refer the reader to the proof of Lemma 2.4 in \cite{DBLP:journals/ita/AalbersbergW86}.
\end{proof} 

\twoPointFive*
\begin{proof}
	The size of $K$ is clear by construction. 
	To compute elements of $\Delta$ represented by finite automata requires $2*|Q|^2$ automata intersection operations. 
	Each intersection operation results in a finite automaton with $|Q|^2$ states. 
	To verify properties (1) through (5), we refer the reader to the proof of Lemma 2.5 in \cite{DBLP:journals/ita/AalbersbergW86}. 
\end{proof}

\twoPointSix*
\begin{proof}
	Towards constructing $M$, we first use \cref{lm:my-lemma-2.5} to construct from $L$ an intermediate automaton $K$ over an alphabet $\Delta$ of size $2*|Q|$. 
	Let $\Delta = \Delta_1 \cup \Delta_2$, and let $\Delta_1 = \set{a_1, \ldots, a_n}$ and $\Delta_2 = \set{b_1, \ldots, b_n}$. 
	Note that $\Delta_1$ and $\Delta_2$ are the same size by construction. 
	For each $a \in \Delta_1$ (respectively $b \in \Delta_2$), we compute a finite automata representation of $\eta(a)$ (respectively $\eta(b)$), which by \cref{lm:my-lemma-2.5} is of size $|Q|^2$. 
	For the $n$ languages $\eta(a_1), \ldots, \eta(a_n)$ represented by finite automata, we use the construction in \cref{lm:my-lemma-2.4} to find a sequence of $n$ languages that universally represent them, and do the same for the $n$ languages $\eta(b_1), \ldots, \eta(b_n)$. 
	The entire construction takes doubly exponential time in $|Q|$. 
	To see the correctness of the construction, we refer the reader to the proof of Lemma 2.6 in~\cite{DBLP:journals/ita/AalbersbergW86}. 
\end{proof} 

\forallInclExistsExptime*

\begin{proof}
	First, observe that:
	\begin{align*}
		\Alang{\lang(M_1)}{C} \subseteq \Elang{\lang(M_2)}{C} \iff \\
		\Alang{\lang(M_1)}{C} \setminus \Elang{\lang(M_2)}{C} = \emptyset \iff \\
		\Alang{\lang(M_1) \setminus \lang(M_2)}{C} = \emptyset \iff \\
		\Elang{N}{C} = \emptyset, \text{ where } 
		\Alang{\lang(M_1) \setminus \lang(M_2)}{C} = \Elang{N}{C} \iff \\ 
		N = \emptyset
	\end{align*}
	Let $n$ be the number of states in the DFA recognizing $\lang(M_1) - \lang(M_2)$, and let $m = \card{\Sigma}$. 
	Because $C$ is transitive, it is the disjoint union of a finite number of complete concurrency relations.
	Let $C$ be partitioned into cliques $C_1, C_2, \ldots, C_n$ whose alphabets $\Sigma_1, \Sigma_2, \ldots, \Sigma_n$ form a partition of $\Sigma$. 
	Let $M$ be a DFA recognizing $K$, defined in \cref{lm:my-lemma-2.5}, generalized to more than 2 alphabets. 
	Then, $M$ has $nm$ states and $n^2m$ transitions. 
	We replace each transition $a$ of $M$ with $\eta(a)$, which by \cref{lm:my-lemma-2.6}, is of size doubly exponential in $n$ and triply exponential in $m$, to obtain an automaton $A_N$ whose state size remains doubly exponential in $n$ and triply exponential in $m$. 
	Then, $\lang(A_N) = N$ by \cref{lm:my-lemma-2.6}. 
	Emptiness checking is linear in the size of the input automaton, and thus the entire algorithm runs in time doubly exponential in $n$ and triply exponential in $m$. 
\end{proof} 

\transitiveExptime*
\begin{proof}
	Let $\protocol$ be a global protocol over $(\AlphSync, \interswapConc)$ and let $Prod(\protocol)$ be its canonical synchronous product whose language $\lang(Prod(\protocol))$ is a regular language that existentially and universally defines the same trace language. 
By \cref{lm:realizable-minimal}, $\protocol$ is realizable iff $Prod(\protocol)$ is reachable and 
	$\Alang{\lang(Prod(\protocol))}{\interswapConc} \subseteq \Elang{\lang(\protocol)}{\interswapConc}$. 
	By \cref{thm:forall-incl-exists-exptime}, this inclusion is decidable in 2-\EXPTIME in the size of $\protocol$, and 3-\EXPTIME in the size of $\Sigma$. 
\end{proof}

\unambiguousDecidable*
\begin{proof} 
	Let $\protocol$ be a global protocol and let $Prod(\protocol)$ be its canonical synchronous product. 
	Let $A$ denote the finite state automaton over $\AlphSync$ that recognizes the lexicographic representation of $Prod(\protocol)$. 
	By \cref{lm:rec-is-unambiguous} $A$ is a witness for the unambiguity of $\lang(Prod(\protocol))$. 
	Thus, realizability reduces to equivalence, which for unambiguous representatives reduces to multiplicity equivalence, which by Theorem 5.4.10 of~\cite{DBLP:books/ws/95/DR1995} is decidable. 
\end{proof} 

\recIsUnambiguous*
\begin{proof}
	By Proposition 6.3.4, \cite{DBLP:books/ws/95/DR1995}, every recognizable trace language admits a lexicographic representation that is regular. 
	Being $T$ recognizable, let $M$ be an automaton such that $\lang(M) = lin(T)$. 
	Let $N$ be the complement of the automaton recognizing the regular language 
	$\bigcup \set{\Sigma^*bC(a)a\Sigma* \mid (a,b) \notin C \text{ and } a < b}$. 
	Thus, $\lang(N)$ is the set of all lexicographically minimal words over $\Sigma$ according to $<$.
	Intersecting $M$ and $N$ results in an automaton linear in the size of $M$ that recognizes the lexicographic representation of $T$. 
	Because $<$ on words is total, the lexicographic representation of $T$ contains at most one representative for each trace in $T$, and thus unambiguously represents $T$. 
\end{proof}

\unambiguousExptime*
\begin{proof} 
	Recall that deadlock freedom of the canonical SCSM can be checked in PSPACE. 
	The algorithm for checking protocol fidelity is an immediate consequence of our reduction to multiplicity equivalence, which in turn can be checked by means of a small model property. 
	From the proof of Theorem 5.4.10~\cite{DBLP:books/ws/95/DR1995}, the model size is bounded by the sum of the state spaces of the finite automata representing the two input trace languages. 
	Let $\protocol = (Q, \AlphSync, \delta, q_0, F)$. 
	We first construct the canonical synchronous product $M$, whose state space is $\Procs^{\card{Q}}$, \ie exponential in the size of $\protocol$. 
	Let $m = \card{Q} + \Procs^{\card{Q}}$. 
	Next, we enumerate the set of all traces over $(\AlphSync, \interswapConc)$ up to length $m$, and for each trace $t \in T_m$, we check that $t$ is contained in $\lang(\protocol)$ if and only if $t$ is contained in $\lang(M)$. 
	Note that because $M$ is a consistently regular trace language, we can check membership in time linear in $\card{t}$. 
	On the other hand, $\protocol$ represents a rational trace language, and by Proposition 5.5.3~\cite{DBLP:books/ws/95/DR1995}, membership can be checked in time $\card{t}^\alpha$, where $\alpha$ is the size of the maximum clique in~$\interswapConc$. 
\end{proof} 

     \section{Additional Material for \cref{sec:realizability} -- Undecidability}
\label{app:undecidability}

Note that it is possible to define communicating state machines parametric in a choice of communication model (cf.\,\cite{NODBLPyet:10.1145/3756907.3756918}). 
Below, we define communication over peer-to-peer FIFO channels, which we refer to simply as CSMs, after its original definition in \cite{DBLP:journals/jacm/BrandZ83}. 

\paragraph{Communicating State Machine.}
\label{def:csm-formalisation}
$\CSMabb{A} = \CSM{\CSMabb{A}}$ is a CSM over~$\Procs$ and~$\MsgVals$ if
${A}_\procA$
is a finite automaton
over~$\Alphabet_\procA$ for every $\procA\in\Procs$, denoted by 
$(Q_\procA, \Alphabet_\procA, \delta_\procA, q_{0, \procA}, F_\procA)$.
Let 
$\prod_{\procA \in \Procs} q_\procA$ 
denote the set of global states and
\mbox{$\channels = \set{(\channel{\procA}{\procB}) \mid \procA,\procB\in \Procs, \procA\neq \procB}$}
denote the set of channels. 
A~\emph{configuration} of $\CSMabb{A}$ is a pair $(\vec{q}, \xi)$, where $\vec{q}\,$ is a global state and
$\xi : \channels \rightarrow \MsgVals^*$
is a mapping from each channel to a sequence of messages.
We use $\vec{q}_\procA$ to denote the state of $\procA$ in $\vec{q}$.
The CSM transition relation, denoted $\rightarrow$, is defined as follows. 
\begin{itemize}
	\item
	$(\vec{q},\xi) \xrightarrow{\snd{\procA}{\procB}{\val}} (\pvec{q}',\xi')$ if
	$(\vec{q}_\procA, \snd{\procA}{\procB}{\val}, \pvec{q}'_\procA)\in\delta_\procA$,
	$\vec{q}_\procC = \pvec{q}'_\procC$ for every participant $\procC \neq \procA$,
	$\xi'(\channel{\procA}{\procB}) =  \xi(\channel{\procA}{\procB})\cdot\val$ and $\xi'(c) = \xi(c)$ for every other channel $c\in \channels$.
	
	\item
	$(\vec{q},\xi) \xrightarrow{\rcv{\procA}{\procB}{\val}} (\pvec{q}',\xi')$ if
	$(\vec{q}_\procB, \rcv{\procA}{\procB}{\val}, \pvec{q}'_\procB)\in\delta_\procB$,
	$\vec{q}_\procC = \pvec{q}'_\procC$ for every participant $\procC \neq \procB$,
	$\xi(\channel{\procA}{\procB}) = \val\cdot \xi'(\channel{\procA}{\procB})$
	and $\xi'(c) = \xi(c)$ for every other channel $c\in \channels$.
\end{itemize}
In the initial configuration $(\vec{q}_0, \xi_0)$, each participant's state in $\vec{q}_0$ is the initial state $q_{0,\procA}$ of $A_\procA$, and $\xi_0 = \xi_\emptystring$, which maps each channel to~$\emptystring$.
A configuration $(\vec{q}, \xi)$ is said to be \emph{final} iff $\vec{q}_\procA$ is final for every~$\procA \in \Procs$ and $\xi = \xi_\emptystring$. 
Runs and traces are defined in the expected way. 
A~run is \emph{maximal} if it is finite and ends in a final configuration. The (asynchronous) language $\lang(\CSMabb{A})$ of the CSM $\CSMabb{A}$ is a set of words over $\AlphAsync^*$ defined as the set of maximal traces.
A configuration $(\vec{q}, \xi)$ is a \emph{deadlock} if it is not final and has no outgoing transitions.
A CSM is \emph{deadlock-free} if no reachable configuration is a deadlock. 

\paragraph{Asynchronous Protocol Semantics.}
To define the \emph{asynchronous} semantics of a protocol $\protocol$ we first map finite and infinite words of~$\protocol$ onto their asynchronous counterpart using $\SyncToAsync(\hole)$, thus obtaining a set of asynchronous words in which matching send and receive events are adjacent to each other. 

We define our protocol semantics as the set of \emph{channel-compliant} words \cite{DBLP:conf/concur/MajumdarMSZ21} that are closed under this notion of indistinguishability. 
Channel compliance characterizes words that respect FIFO order, i.e. receive events appear after their matching send event, and the order of receive events follows that of send events in each~channel. 
Channel compliance is the word-based analogue of \emph{p2p-linearizable}, used to describe message sequence charts. 

\begin{definition}[Channel compliance] 
	Let $w \in \AlphAsyncSubscript^\infty$. We say that $w$ is \emph{channel-compliant} if for all prefixes $w' \leq w$, for all $\procA\neq \procB \in \Procs$, 
	$\MsgVals(w' \wproj_{\rcv{\procA}{\procB}{\_}}) \leq 
	\MsgVals(w' \wproj_{\snd{\procA}{\procB}{\_}})$.  
\end{definition} 

The asynchronous semantics of a protocol is defined as follows. 

\begin{definition}[Asynchronous protocol semantics]
Let $\protocol$ be a global protocol. Its asynchronous semantics is given by:
	\begin{align*}
		\lang_{async}(\protocol) = \; & \{ w' \!\in\! \AlphAsync^* \mid \exists w \in \AlphAsync^*. \, w \!\in\! \SyncToAsync(\lang(\protocol)) 
		\\ & \hspace{2.5cm}
		\land w' \text{ is channel-compliant } \\ & \hspace{2.5cm}
		\land \forall\procA \!\in\! \Procs.~w' \wproj_{\AlphAsync_\procA} \!=\! w \wproj_{\AlphAsync_\procA}\} 
		\enspace . 
	\end{align*}
For our proofs, we use part of the condition as binary relation $\interswap$ on words over $\AlphAsync$.
Let $w \in \AlphAsync^*$ be a channel-compliant word and let $w' \in \AlphAsync^*$.
We define $w \interswap w'$ iff 
		$w'$ is channel-compliant 
		and 
        $\forall\procA \!\in\! \Procs.~w' \wproj_{\AlphAsync_\procA} \!=\! w \wproj_{\AlphAsync_\procA}$. 
\end{definition}

In addition, we use the asynchronous send events in the following proof, \ie 
$
\AlphAsync_! = \set{\snd{\procA}{\procB}{\val} \mid \procA \neq \procB \in \Procs \land \val \in \MsgVals}, 
$
which is a subset of $\AlphAsync$. 

\syncRealizationAndConditionXyieldAsyncRealization*
\begin{proof}
We enumerate the claims: 
\begin{enumerate}[label=(\alph*)]
\item \label{syncRealAndXAsyncReal:obligation-1}
      $\lang(\makeAsync{\CSMabb{B}}) = \langasync(\protocol)$ and
\item \label{syncRealAndXAsyncReal:obligation-2}
      $\makeAsync{\CSMabb{B}}$ is deadlock-free if $\CSMabb{B}$ satisfies Condition~\condX. 
\end{enumerate}
We might explicitly say asynchronous (resp.\ synchronous) run or asynchronously (resp.\ synchronously) deadlock-free -free even though the semantics of a CSM (resp.\ SCSM) is always asynchronous (resp.\ synchronous).
For both, we will make use of the following property, which we prove first. 

\textit{Property (i):}
For every asynchronous run 
\[ 
(\vec{q}_0, \xi_0) \xrightarrow{w_1} (\vec{q}_1, \xi_1) \xrightarrow{w_2} \ldots 
(\vec{q}_{n-1}, \xi_{n-1}) \xrightarrow{w_{n-1}} (\vec{q}_n, \xi_n)
\]
with $w = w_1 \ldots w_{n-1}$,
there is $u \interswap w$ and 
$u = u_1 \ldots u_{n-1}$ along with $0 < j < n$ 
such that 
$u_1 \ldots u_j \in (\SyncToAsync(\AlphSync))^*$, 
$u_{j+1} \ldots u_{n-1} \in \AlphAsync_{!}^*$
and 
the senders in the latter are different: 
for every $i,k$ such that $j < i < n$ and $j < k < n$ with 
$u_i = \snd{\procA}{\procB}{\val_1}$ and
$u_k = \snd{\procC}{\procD}{\val_2}$, 
if $i \neq k$ 
then $\procA \neq \procC$. 

\textit{Proof of Property (i).}
We prove this claim by induction on $n$.
For the induction base, let $n = 0$. It is trivial that the claim holds for the empty run and empty trace $\emptystring$. 
For the induction step, we assume that the claim holds for $n$ (as spelled out above) and prove it for $n+1$. 
Let 
{\small
\[ 
(\vec{q}_0, \xi_0) \xrightarrow{w_1} (\vec{q}_1, \xi_1) 
\ldots 
(\vec{q}_{n-1}, \xi_{n-1}) 
\xrightarrow{w_{n-1}} (\vec{q}_n, \xi_n)
\xrightarrow{w_{n}} (\vec{q}_{n+1}, \xi_{n+1})
\]
}

\noindent be the asynchronous run to consider. 
We do a case analysis on the shape of $w_n$, \ie if it is a send or receive event. 

Suppose that $w_n = \snd{\procA}{\procB}{\val}$. 
Then, we do not need to reorder $w$ and it suffices to show that $\procA$ does not occur as sender in $u_j \ldots u_{n-1}$. 
Towards a contradiction, assume that this was the case, \ie 
there is $j < i < n$ such that $u_i = \snd{\procA}{\procC}{\val'}$. 
Then, however, $\vec{q}_{n,\procA}$ would be waiting to receive the immediate acknowledgment, contradicting the fact that $\procA$ is the sender in $w_n$. 

Suppose that $w_n = \rcv{\procB}{\procA}{\val}$. 
We use the induction hypothesis to obtain the following run: 
{\small
\[ 
(\vec{q}_0, \xi_0) \xrightarrow{u_1} (\pvec{q}'_1, \xi'_1) 
\ldots 
(\vec{q}_{n-1}, \xi_{n-1}) 
\xrightarrow{u_{n-1}} (\vec{q}_n, \xi_n)
\xrightarrow{w_{n}} (\vec{q}_{n+1}, \xi_{n+1})
\]
}

\noindent By definition of CSM semantics, there is 
$j < k < n$ with $u_k = \snd{\procB}{\procA}{\val}$. 
We claim that there is no $j < i < n$ with 
$u_i = \snd{\procA}{\procC}{\val'}$. 

Towards a contradiction, assume there is such $i$. 
We do a case analysis if $i < k$ or $k < i$. 
(It is trivial that they cannot be the same.)
First, let $k < i$, \ie $\procA$ sent after $\procB$: then $\procA$ would wait for the immediate acknowledgment for $\val'$ by $\procC$ currently. Even if $\procC = \procB$, the above send event could not belong to the requested acknowledgment because it was sent before the actual message, yielding a contradiction.
Second, let $i < k$, \ie $\procB$ sent after $\procA$: still, $\procA$ waits for the acknowledgment (if not, $i \leq j$ because it was received); 
if $w_n$ is not this acknowledgment, we get a contradiction because $w_n$ is not possible; 
if $w_n$ is this acknowledgment, $\procB$ must have received $\val'$ before and then $j < i$, yielding a contradiction. 

Thus, we know that $w_n$ is $\procA$'s only event in $u_{j+1} \ldots u_{n-1} \cat w_n$.  
With the induction hypothesis, we also know that $u_k$ is $\procB$'s only (send) event in 
$u_{j+1} \ldots u_{n-1} \cat w_n$. 
Because of this, there are no dependencies and we can move $u_k$ and $w_n$ to the front of the send events using $\interswap$:
\[  
u_{j+1} \ldots u_{n-1} \cat w_n 
\; \interswap \;
u_k \cat w_n \cat u_{j+1} \ldots u_{k-1} \cat u_{k+1} \ldots u_{n-1} 
\]
Thus, 
\[
u_1 \ldots 
u_{n-1} \cat w_n 
\; \interswap \;
u_1 \ldots u_j \cat 
u_k \cat w_n \cat u_{j+1} \ldots u_{k-1} \cat u_{k+1} \ldots u_{n-1} 
\]
and $u_k \cat w_n$ basically belongs to the part that is already acknowledged. 
By the induction hypothesis, the remaining send events satisfy the requirement on different senders. 

\textit{End Proof of Property (i).}

We prove \ref{syncRealAndXAsyncReal:obligation-1} through two inclusions.
First, we show 
$\lang(\makeAsync{\CSMabb{B}}) \subseteq \langasync(\protocol)$.
Let $w \in \lang(\makeAsync{\CSMabb{B}})$. 
We show that $w \in \langasync(\protocol)$. 
By definition, there is an asynchronous run for $w$: 
\[ 
(\vec{q}_0, \xi_0) 
\xrightarrow{w}\starred
(\vec{q}, \xi_\emptystring)
\enspace .
\]
With Property (i), there is $u \in (\SyncToAsync(\AlphSync))^*$
such that $u \interswap w$ with $j = \card{w}$ because there are no remaining send events with empty channels. 
By the semantics of SCSMs, 
it follows that there is $u' \in \AlphSync^*$ with 
$\SyncToAsync(u') = u$ and $u \in \lang(\CSMabb{B})$.
By assumption, 
we have that $u' \in \langsync(\protocol)$.
Thus, by definition, it holds that 
$u = \SyncToAsync(u') \in \langasync(\protocol)$.
Together with $w \interswap u$, we have $w \in \langasync(\protocol)$, which concludes this case. 

Second, we show 
$\langasync(\protocol) \subseteq \lang(\makeAsync{\CSMabb{B}})$. 
Thus, let $w \in \langasync(\protocol)$. 
We show that $w \in \lang(\makeAsync{\CSMabb{B}})$. 
By definition, there is 
$w' \in \lang(\protocol)$ with $\SyncToAsync(w') \interswap w$. 
By assumption, we have 
$w' \in \lang(\CSMabb{B})$, giving rise to a synchronous run
\[
\vec{q}_0
\xrightarrow{w'}\starred
\vec{q} \enspace .
\]
It is easy to see that this run can be mimicked asynchronously: between every two configurations, one needs to add a configuration where all but one channel is empty. 
Thus, 
$w' \in \lang(\makeAsync{\CSMabb{B}})$. 
By closure of CSM semantics under $\interswap$, it holds that  
$w \in \lang(\makeAsync{\CSMabb{B}})$, which concludes this case. 

We prove \ref{syncRealAndXAsyncReal:obligation-2} by contraposition. 
We assume that $\makeAsync{\CSMabb{B}}$ has an asynchronous deadlock and show that $\CSMabb{B}$ then also has a synchronous deadlock. 
Let 
{ \small 
\[ 
(\vec{q}_0, \xi_0) \xrightarrow{w_1} (\vec{q}_1, \xi'_1) 
\ldots 
(\vec{q}_{n-1}, \xi_{n-1}) 
\xrightarrow{w_{n-1}} (\vec{q}_n, \xi_n)
\xrightarrow{w_{n}} (\vec{q}_{n+1}, \xi_{n+1})
\]
}

\noindent be an asynchronous run that ends in a deadlock. 

With Property (i), 
there is $u \interswap w$ and 
$u = u_1 \ldots u_{n-1}$ along with $0 < j < n$ 
such that 
$u_1 \ldots u_j \in (\SyncToAsync(\AlphSync))^*$, 
$u_{j+1} \ldots u_{n-1} \in \AlphAsync_!^*$
and 
the senders in the latter are different: 
for every $i,k$ with $j < i < n$ and $j < k < n$ with 
$u_i = \snd{\procA}{\procB}{\val_1}$ and
$u_k = \snd{\procC}{\procD}{\val_2}$, 
if $i \neq k$ 
then $\procA \neq \procC$. 

We show that $\vec{q}_j$ is a (reachable) synchronous deadlock in $\CSMabb{B}$. 
With the above reordering, it is obvious that $\vec{q}_j$ is synchronously reachable. 
Due to sink finality, $\vec{q}_j$ cannot be a final configuration because some participants can send messages so we need to extend the run to avoid a deadlock. 
Generally speaking, it could be that a synchronous run from $\vec{q}_j$ will never encounter the send events in 
$u_{j+1} \ldots u_{n-1}$
 because the respective receiver blocks the exchange. 
Here, Condition \condX prevents this though. 
Let $S$ be all the senders in 
$u_{j+1} \ldots u_{n-1}$. 
By Property (i), all senders are different. 
We instantiate Condition \condX for $\vec{q}_j$. 
For each sender $\procA \in S$ with send event $\snd{\procA}{\procB}{\val}$, the first case of Condition \condX can never apply. 
If it did, $\procB$ could receive, yielding a contradiction to having an asynchronous deadlock.
So the second case must apply for every sender, which implies that every possible receiver actually is obliged to send (including acknowledgment) before receiving another message.  
Hence, every sender is waiting for their receiver who themselves needs to send first and then to wait for their acknowledgment. 
This is also a deadlock in the synchronous setting, which is what we ought to show. 
\end{proof}

\canonicalRealizationSatConditionX*
\begin{proof}
Let $\TapeAlph$ be the tape alphabet and $\TMStates$ be the states of $\TM$ with 
\mbox{$\TapeAlph \inters \TMStates = \emptyset$}.
To prove the claim, we recall the necessary notation and intuition about the construction of $\protocol_\TM$ from \cite[Thm.\,9.1]{DBLP:phd/dnb/Stutz24} where the full construction can be found.

The encoding has five participants $\procA_1, \ldots, \procA_5$ who send configurations to each other.
Therefore, messages are from the set
$\TapeAlph \dunion
\set{
    \markNewRound,
    \markBeginConf,
    \markEndConf,
    \markEndLoop
}
\dunion Q$.

The following notation is used for conciseness: 
$\msgBackForth{\procA}{\procB}{\val}$
abbreviates
$\msgFromTo{\procA}{\procB}{\val} \cat
 \msgFromTo{\procB}{\procA}{\val}$.
Interactions are indeed only specified using
$\msgBackForth{\_}{\_}{\_}$, 
which is why the encoding is immediately acknowledging. 
These are also used to define regular expressions and complements thereof for which 
$\msgBackForth{\procA}{\procB}{\val}$
is considered as its individual symbols.
For a word $w = w_1 \ldots w_i$, 
the sequence 
$\msgBackForth{\procA}{\procB}{w_1}
 \cdots
 \msgBackForth{\procA}{\procB}{w_i}$
is abbreviated as 
$\msgBackForth{\procA}{\procB}{w}$. 

For words $C_1, D_1, C_2, D_2 \ldots, C_m, D_m \in (\TapeAlph \dunion \TMStates)^*$, let 
{
\small
\begin{align*}
   & w(C_1, D_1, C_2, D_2, \ldots, C_m, D_m)
   \is \quad
   \\ & \qquad  
    \msgBackForth{\procA_3}{\procA_2}{\markNewRound} \cat
    \msgBackForth{\procA_2}{\procA_1}{\markBeginConf} \cat
    \msgBackForth{\procA_2}{\procA_1}{C_1} \cat
    \msgBackForth{\procA_2}{\procA_1}{\markEndConf} \cat \invisibleEndLine
    \\ & \qquad
    \msgBackForth{\procA_3}{\procA_4}{\markNewRound} \cat
    \msgBackForth{\procA_4}{\procA_5}{\markBeginConf} \cat
    \msgBackForth{\procA_4}{\procA_5}{D_1} \cat
    \msgBackForth{\procA_4}{\procA_5}{\markEndConf} \cat \invisibleEndLine
    \\ & \qquad
    \msgBackForth{\procA_3}{\procA_2}{\markNewRound} \cat
    \msgBackForth{\procA_2}{\procA_1}{\markBeginConf} \cat
    \msgBackForth{\procA_2}{\procA_1}{C_2} \cat
    \msgBackForth{\procA_2}{\procA_1}{\markEndConf} \cat \invisibleEndLine
    \\ & \qquad
    \msgBackForth{\procA_3}{\procA_4}{\markNewRound} \cat
    \msgBackForth{\procA_4}{\procA_5}{\markBeginConf} \cat
    \msgBackForth{\procA_4}{\procA_5}{D_2} \cat
    \msgBackForth{\procA_4}{\procA_5}{\markEndConf} \cat
    \\ & \qquad \cdots \\ & \qquad
    \msgBackForth{\procA_3}{\procA_2}{\markNewRound} \cat
    \msgBackForth{\procA_2}{\procA_1}{\markBeginConf} \cat
    \msgBackForth{\procA_2}{\procA_1}{C_m} \cat
    \msgBackForth{\procA_2}{\procA_1}{\markEndConf} \cat \invisibleEndLine
    \\ & \qquad
    \msgBackForth{\procA_3}{\procA_4}{\markNewRound} \cat
    \msgBackForth{\procA_4}{\procA_5}{\markBeginConf} \cat
    \msgBackForth{\procA_4}{\procA_5}{D_m} \cat
    \msgBackForth{\procA_4}{\procA_5}{\markEndConf} \phantom{\cat}
    \enspace \hspace{-1ex} .
\end{align*}
}

\cref{fig:mixed-choice-encoding} illustrates the word 
$w(C_1, D_1, C_2, D_2, \ldots, C_m, D_m)$ 
as message sequence chart (MSC).
Intuitively, $\procA_2$ sends the sequence $C_i$ to~$\procA_1$ while
$\procA_4$ sends the sequence $D_i$ to~$\procA_5$.
Each sequence is started by a $\markBeginConf$-message and finished by a $\markEndConf$-message between the respective pair.
The participant $\procA_3$ starts each round by sending $\markNewRound$.

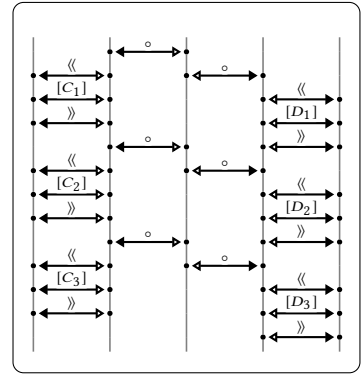
\begin{figure}[t]
\centering
    \begin{tikzpicture}[hmsc,baseline, node distance=7.5ex]

\begin{scope}[msc=s0]
  \begin{scope}[participant=p_1]
    \node[head];
    \node[no event] {};
    \node[event] {};
    \node[event] {};
    \node[event] {};
    \node[no event] {};
    \node[event] {};
    \node[event] {};
    \node[event] {};
    \node[no event] {};
    \node[event] {};
    \node[event] {};
    \node[event] {};
    \node[no event] {};
  \end{scope}

  \begin{scope}[participant=p_2]
    \node[head];
    \node[event] {};
    \node[event] {};
    \node[event] {};
    \node[event] {};
    \node[event] {};
    \node[event] {};
    \node[event] {};
    \node[event] {};
    \node[event] {};
    \node[event] {};
    \node[event] {};
    \node[event] {};
    \node[no event] {};
  \end{scope}

  \begin{scope}[participant=p_3]
    \node[head];
    \node[event] {};
    \node[event] {};
    \node[no event] {};
    \node[no event] {};
    \node[event] {};
    \node[event] {};
    \node[no event] {};
    \node[no event] {};
    \node[event] {};
    \node[event] {};
    \node[no event] {};
    \node[no event] {};
    \node[no event] {};
  \end{scope}

  \begin{scope}[participant=p_4]
    \node[head];
    \node[no event] {};
    \node[event] {};
    \node[event] {};
    \node[event] {};
    \node[event] {};
    \node[event] {};
    \node[event] {};
    \node[event] {};
    \node[event] {};
    \node[event] {};
    \node[event] {};
    \node[event] {};
    \node[event] {};
  \end{scope}

  \begin{scope}[participant=p_5]
    \node[head];
    \node[no event] {};
    \node[no event] {};
    \node[event] {};
    \node[event] {};
    \node[event] {};
    \node[no event] {};
    \node[event] {};
    \node[event] {};
    \node[event] {};
    \node[no event] {};
    \node[event] {};
    \node[event] {};
    \node[event] {};
  \end{scope}

  \draw[messagesMC]
    (p_3-2) edge node[msg]{$\markNewRound$} (p_2-2)
    (p_2-3) edge node[msg]{$\markBeginConf$} (p_1-3)
    (p_2-4) edge node[msg]{$[C_1]$} (p_1-4)
    (p_2-5) edge node[msg]{$\markEndConf$} (p_1-5)
    (p_3-6) edge node[msg]{$\markNewRound$} (p_2-6)
    (p_2-7) edge node[msg]{$\markBeginConf$} (p_1-7)
    (p_2-8) edge node[msg]{$[C_2]$} (p_1-8)
    (p_2-9) edge node[msg]{$\markEndConf$} (p_1-9)
    (p_3-10) edge node[msg]{$\markNewRound$} (p_2-10)
    (p_2-11) edge node[msg]{$\markBeginConf$} (p_1-11)
    (p_2-12) edge node[msg]{$[C_3]$} (p_1-12)
    (p_2-13) edge node[msg]{$\markEndConf$} (p_1-13)
(p_3-3) edge node[msg]{$\markNewRound$} (p_4-3)
    (p_4-4) edge node[msg]{$\markBeginConf$} (p_5-4)
    (p_4-5) edge node[msg]{$[D_1]$} (p_5-5)
    (p_4-6) edge node[msg]{$\markEndConf$} (p_5-6)
    (p_3-7) edge node[msg]{$\markNewRound$} (p_4-7)
    (p_4-8) edge node[msg]{$\markBeginConf$} (p_5-8)
    (p_4-9) edge node[msg]{$[D_2]$} (p_5-9)
    (p_4-10) edge node[msg]{$\markEndConf$} (p_5-10)
    (p_3-11) edge node[msg]{$\markNewRound$} (p_4-11)
    (p_4-12) edge node[msg]{$\markBeginConf$} (p_5-12)
    (p_4-13) edge node[msg]{$[D_3]$} (p_5-13)
    (p_4-14) edge node[msg]{$\markEndConf$} (p_5-14)
  ;
\end{scope}

\end{tikzpicture}
     \caption
    {MSC representation $\msc(w(C_1, D_1, C_2, D_2, C_3, D_3))$ \\
    where the double-sided arrow abbreviates two message interactions of which the direction with the filled tip goes first and
    $[C_1]$ denotes a sequence of such message interactions.}
    \label{fig:mixed-choice-encoding}
\end{figure}

Two languages $L_l$ and $L_r$ are defined, one for each branch of the protocol: 
\begin{align*}
 L_l \is \; &
    \set{
        \langasync(w(C_1, D_1, \ldots, C_m, D_m)) \mid 
        \\ & \hspace{8ex}
        m \geq 1, C_1, D_1, \ldots, C_m, D_m \in (\TapeAlph \dunion \TMStates)^*
    }
 \\
 L_r \is \; &
    L_l \setminus
    \set{
        \langasync(w(u_1, u_1, \ldots, u_m, u_m)) \mid
        \\ & \hspace{8ex}
        (u_1, \ldots, u_m) \text{ is an accepting computation} \enspace .
    }
\end{align*}
The definition of accepting computation is as expected but not important for this proof and hence omitted.
The following sequence of messages ensures that the protocol is \sinkfinal: 
\[
 w_{\mathit{end}} \is
    \msgBackForth{\procA_3}{\procA_2}{\markEndLoop} \cat
    \msgBackForth{\procA_2}{\procA_1}{\markEndLoop} \cat
    \msgBackForth{\procA_3}{\procA_4}{\markEndLoop} \cat
    \msgBackForth{\procA_4}{\procA_5}{\markEndLoop}
\]
It is appended to $L_l$ and $L_r$ to obtain 
$L'_l  \is \set{w \cat w_{\mathit{end}} \mid w \in L_l}$
and
$L'_r  \is \set{w \cat w_{\mathit{end}} \mid w \in L_r}$
.

Ultimately, the protocol $\protocol_\TM$ is constructed such that 
\[
\langasync(\protocol_{\TM}) =
    \set{\msgFromTo{\procA_2}{\procA_3}{\lbl{l}} \cat w \mid w \in L'_l}
    \dunion
    \set{\msgFromTo{\procA_2}{\procA_3}{\lbl{r}} \cat w \mid w \in L'_r}
    \enspace .
\]
By definition of $L'_l$ and $L'_r$, every word ends with $w_{\mathit{end}}$ so $\protocol_{\TM}$ is \sinkfinal.
It is easy to see that only $\procA_2$ and $\procA_3$ will ever know about the initial branching choice. 
In \cite[Thm.\,9.1]{DBLP:phd/dnb/Stutz24} it was proven that $\protocol_\TM$ is asynchronously realizable if and only if $L_l = L_r$.
We claim that synchronous realizability implies $L_l = L_r$.
(While the other direction likely also holds, we do not need it for this proof.)
\FS{@EL: Too informal?}
Towards a contradiction, assume that $L_l \neq L_r$. 
Participant $\procA_4$ never knows about the choice and sends sequences of configurations to $\procA_5$ (interrupted by communication with $\procA_3$). 
By construction, it holds that $L_l \wproj_{\AlphSync_{\procA_4}} \neq L_r \wproj_{\AlphSync_{\procA_4}}$ though, so $\procA_4$ will send unspecified messages in one of the branches, yielding a deadlock. 
In fact, this will happen in the branch for $L_l$ because $L_l \supseteq L_r$ by construction. 

For our proof, equality of both languages is very helpful: 
by assumption, $\protocol_\TM$ is synchronously realizable so we know that $L_l = L_r$ and we only need to show Condition $\condX$ holds for the canonical SCSM of $\protocol_\TM$ where
\[
\langasync(\protocol_{\TM}) = 
    \set{\msgFromTo{\procA_2}{\procA_3}{\lbl{l}} \cat w \mid w \in L'_l}
     \dunion 
    \set{\msgFromTo{\procA_2}{\procA_3}{\lbl{r}} \cat w \mid w \in L'_l}
    \enspace .
\]
For $L'_l$, we generate the finite automata of the canonical SCSM for participants 
$\procA_1$, $\procA_2$, and $\procA_3$. 
The ones for $\procA_5$ and $\procA_4$ are analogous to the ones for $\procA_1$ and $\procA_2$ respectively. 
Once shown that Condition \condX holds for these finite automata, it is easy to see that adding the initial choice by $\procA_2$ and $\procA_3$ preserves Condition~\condX. 
\cref{fig:mixed-choice-canonical-p1} shows the finite automata.  
Note that we specify messages and their acknowledgments jointly. 
Thus, technically, we would introduce one additional intermediate state per message exchange. 

\begin{figure}[t]
\centering
    \begin{subfigure}[b]{1\columnwidth}
    \centering
    \begin{tikzpicture}[sem, node distance=1.5cm]

  \node[state, initial above] (p11) {$q_{1,1}$};
  \node[state, left = of p11] (p12) {$q_{1,2}$}; 
  \node[state, final, right = of p11] (p13) {$q_{1,3}$}; 

  \path 
  (p11) edge [bend right, above] node {$\msgBackForth{\procA_2}{\procA_1}{\markBeginConf}$} (p12)
  (p12) edge [loop left] node [yshift=.7mm] {$\meta{\forall \val \in \TMStates \dunion \TapeAlph.} \; \msgBackForth{\procA_2}{\procA_1}{\val}$} (p12)
  (p12) edge [bend right, below] node {$\msgBackForth{\procA_2}{\procA_1}{\markEndConf}$} (p11)
  (p11) edge [above] node {$\msgBackForth{\procA_2}{\procA_1}{\markEndLoop}$} (p13)
  ;

\end{tikzpicture}
     \caption{Finite automaton for $\procA_1$.}
    \end{subfigure}
    \begin{subfigure}[b]{1\columnwidth}
    \centering
    \begin{tikzpicture}[sem, node distance=1.5cm]

  \node[state, initial above] (p21) {$q_{2,1}$};
  \node[state, left = of p21] (p22) {$q_{2,2}$}; 
  \node[state, below = of p22] (p23) {$q_{2,3}$}; 
  \node[state, right = of p21] (p24) {$q_{2,4}$}; 
  \node[state, final, below = of p24] (p25) {$q_{2,5}$}; 

  \path 
  (p21) edge [above] node {$\msgBackForth{\procA_3}{\procA_2}{\markNewRound}$} (p22)
  (p22) edge [left] node {$\msgBackForth{\procA_2}{\procA_1}{\markBeginConf}$} (p23)
  (p23) edge [loop left] node [yshift=.7mm] {$\meta{\forall \val \in \TMStates \dunion \TapeAlph.} \; \msgBackForth{\procA_2}{\procA_1}{\val}$} (p23)
  (p23) edge [sloped, below] node {$\msgBackForth{\procA_2}{\procA_1}{\markEndConf}$} (p21)
  (p21) edge [above] node {$\msgBackForth{\procA_3}{\procA_2}{\markEndLoop}$} (p24)
  (p24) edge [left] node {$\msgBackForth{\procA_2}{\procA_1}{\markEndLoop}$} (p25)
  ;

\end{tikzpicture}
     \caption{Finite automaton for $\procA_2$.}
    \end{subfigure}
    \begin{subfigure}[b]{1\columnwidth}
    \centering
    \begin{tikzpicture}[sem, node distance=1.5cm]

  \node[state, initial above] (p31) {$q_{3,1}$};
  \node[state, left = of p31] (p32) {$q_{3,2}$}; 
  \node[state, right = of p31] (p33) {$q_{3,3}$}; 
  \node[state, final, right = of p33] (p34) {$q_{3,4}$}; 
  \node[left = of p32, xshift=0.9cm] (pEmpty) {}; 

  \path 
  (p31) edge [bend right, above] node {$\msgBackForth{\procA_3}{\procA_2}{\markNewRound}$} (p32)
  (p32) edge [bend right, below] node {$\msgBackForth{\procA_3}{\procA_4}{\markNewRound}$} (p31)
(p31) edge [above] node {$\msgBackForth{\procA_3}{\procA_2}{\markEndLoop}$} (p33)
  (p33) edge [above] node {$\msgBackForth{\procA_3}{\procA_4}{\markEndLoop}$} (p34)
  ;

\end{tikzpicture}
     \caption{Finite automaton for $\procA_3$.}
    \end{subfigure}
    \caption
    {
    Canonical SCSM for $\protocol_\TM$ (if realizable):
    finite automata for 
    $\procA_1$, $\procA_2$, and $\procA_3$ where  
    $\msgFromTo{\procA_1}{\procA_2}{\val}$
    is abbreviation for  
    $
    \msgFromTo{\procA_1}{\procA_2}{\val}
    \cat 
    \msgFromTo{\procA_2}{\procA_1}{\val}
    $. 
    }
    \label{fig:mixed-choice-canonical-p1}
\end{figure}
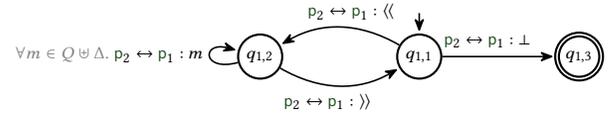
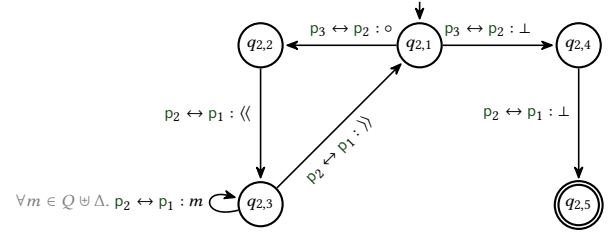
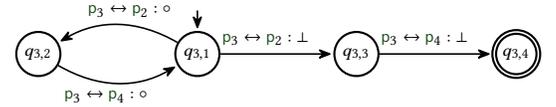

Observe that the finite automata are immediately acknowledging. 
Hence, once a message is sent, the sender will wait for the immediate acknowledgment, satisfying the first condition of \condX, so we focus on the actual messages being transmitted. 

In the synchronous semantics of the canonical SCSM, we could exhaustively check that \condX is true. 
We argue more abstractly though. 
No participant ever initiates an exchange with $\procA_3$, \ie it only receives immediate acknowledgments, so it trivially satisfies Condition~\condX. 
One can see that $\procA_1$ (resp.\ $\procA_5$) does only receive from $\procA_2$ (resp.\ $\procA_4$). 
Thus, they act ``in sync'': 
$\procA_1$ is in $q_{1,1}$ iff $\procA_2$ is in $q_{2,1}$ or $q_{2,2}$ 
and
$\procA_1$ is in $q_{1,2}$ iff $\procA_2$ is in $q_{2,3}$. 
It is easy to observe that $\procA_1$ (resp.\ $\procA$) can always receive immediately what $\procA_2$ (resp.\ $\procA_4$) sends, satisfying the first case of \condX. 
Except for immediate acknowledgments, only $\procA_3$ sends to $\procA_2$, \ie $\markNewRound$ and $\markEndLoop$. 
It can happen that $\procA_2$ is not able to receive this message immediately as it might be in 
$q_{2,2}$ or $q_{2,3}$ but it satisfies the second case of \condX: it only sends to another participant with immediate acknowledgments and can receive all possible values from $\procA_3$ next, \ie in $q_{2,1}$.  
It cannot receive any of those values in $q_{2,4}$ but it is easy to see that these values cannot be sent by $\procA_3$: 
$\markNewRound$ is never sent after $\markEndLoop$, which in turn is only sent once. 
The reasoning for $\procA_4$ is analogous. 
This shows that Condition \condX holds for the canonical SCSM and concludes the proof. 
\end{proof}

\syncRealizabilityUndecidable*
\begin{proof}
As explained before, we use the sink-final (mixed choice) protocol $\protocol_{\TM}$ that is asynchronously realizable if and only if there is no accepting computation for some $\TM$ \cite[Thm.\,9.1]{DBLP:phd/dnb/Stutz24}, which in turn is an undecidable problem. 
Hence, it suffices to show that $\protocol_{\TM}$ is synchronously realizable if and only if $\protocol_{\TM}$ is asynchronously realizable. 
The ``if'' direction follows directly from \cite[Thm.\,5.2]{NODBLPyet:10.1145/3756907.3756918}: while this work uses the term global types, they actually represent protocols as general as ours. 
For the ``only if'' direction, we make use of the fact that the canonical SCSM is indeed a synchronous realization because $\protocol_\TM$ is synchronously realizable. 
\cref{lm:canonical-realization-sat-Condition-X} shows that the canonical SCSM indeed satisfies Condition \condX.
This allows us to apply 
\cref{lm:sync-realization-and-Condition-X-yield-async-realization}, which gives that the canonical SCSM can be transformed into an asynchronous realization. 
\end{proof}

     \section{Additional Material for \cref{sec:type-system}}
\label{app:type-system}

\begin{lemma}[Admissibility of structural precongruence for typing]
\label{lm:str-congr-preserves-typability-procs}
Let $P_1$ and $P_2$ be processes.
If
    $
        \typingContextOne
        \typingContextCat
        \typingContextTwo
            \types P_1
    $
    and
    $
        P_1 \precongr P_2
    $,
then
    $
        \typingContextOne
        \typingContextCat
        \typingContextTwo
        \typingContextCat
        \typingContextThree
            \types P_2
    $.
\end{lemma}
\begin{proof}
We do a case analysis on $\congr$ and reason for both directions.
Subsequently, we consider the additional rule for $\precongr$.
\begin{itemize}
 \item $P_1 \parallel P_2
        \congr
        P_2 \parallel P_1$: \\
        By inversion, we know that \procTypingParallel is the first rule applied in the typing derivation.
        This rule is symmetric so the same typing derivation works.
 \item $(P_1 \parallel P_2) \parallel P_3
        \congr
        P_1 \parallel (P_2 \parallel P_3)$: \\
        By inversion, we know that \procTypingParallel is the first and second rule applied in the typing derivation.
        It is easy to see that the typing derivation can be rearranged to match the structure.
 \item $P \parallel \zero
        \congr
        P$: \\
        First, assume that there is a typing derivation

        \vspace{-2ex}
        { \small
        \begin{mathpar}
            \inferrule*[right=\procTypingParallel]{
                \typingContextOne
                    \typingContextCat
                    \typingContextTwo_1
                    \types
                    P \\
\typingContextOne
                    \typingContextCat
                    \typingContextTwo_2
                    \types
                    \zero \\
            }{
                \typingContextOne
                    \typingContextCat
                    \typingContextTwo_1,
                    \typingContextTwo_2
                    \types
                    P \parallel \zero
            }
        \end{mathpar}
        }

        We show there is a typing derivation
        $
            \typingContextOne
                \typingContextCat
                \typingContextTwo_1,
                \typingContextTwo_2
                \types
                P
        $.
        Inversion yields that two rules can be applied for the given typing derivation
        $
            \typingContextOne
                \typingContextCat
                \typingContextTwo_2
                \types
                \zero
        $:
        \procTypingEnd
        and
        \procTypingZero.
        Thus, it follows that
        $
            \typingContextTwo_2 =
\set{c \hasType \vec{q}_\procA}_{c \in S}
        $
for a set of session endpoints~$S$
        and
        $\EndState(\vec{q}_\procA)$ holds
for every
        $c \in S$.
        By inversion, there is a typing derivation for
        $
                \typingContextOne
                    \typingContextCat
                    \typingContextTwo_1
                    \types
                    P
        $.
The only difference to our goal is the typing context
        $\typingContextTwo_2$.
        This can be taken care of using \mbox{\procTypingEnd} as in the other typing derivation, concluding this~case.
        \\
        Second, assume there is a typing derivation for 
        $
            \typingContextOne
                \typingContextCat
                \typingContextTwo
                \types
            P
        $. 
        We show there is a typing derivation for
        $
            \typingContextOne
                \typingContextCat
                \typingContextTwo
                \types
                P \parallel \zero
        $.
        We first apply \procTypingParallel to obtain

        \vspace{-2ex}
        { \small
        \begin{mathpar}
            \inferrule*[right=\procTypingParallel]{
                \typingContextOne
                    \typingContextCat
                    \typingContextTwo
                    \types
                    P \\
\inferrule*[right=\procTypingZero]{
}{
                \typingContextOne
                    \typingContextCat
                    \emptyset
                    \types
                    \zero
                }
            }{
                \typingContextOne
                    \typingContextCat
                    \typingContextTwo
                    \types
                    P \parallel \zero
            }
        \end{mathpar}
        }

        for which the right premise is met with \procTypingZero and the left premise is given by~assumption.

 \item $(\restr s \hasType \CSMabb{A}) \, (\restr s' \hasType \CSMabb{B}) \, P
        \congr
        (\restr s' \hasType \CSMabb{B}) \, (\restr s \hasType \CSMabb{A}) \, P$: \\
        By inversion, both typing derivations need to apply
        \procTypingRestr
        twice in the beginning.
        It is straightforward that both rule applications do not interfere with each other, yielding the same premise to prove:
        \[
            \typingContextOne
                \typingContextCat
                \typingContextTwo,
                \typingContextTwo_s,
                \typingContextTwo_{s'}
            \types
            P
        \]
        Thus, this is given by assumption.

 \item $(\restr s \hasType \CSMabb{A}) \, (P_1 \parallel P_2)
        \congr
        P_1 \parallel (\restr s \hasType \CSMabb{A}) \, P_2$
        and $s$ is not free in $P_1$: \\
        First, we assume there is a typing derivation for
        $(\restr s \hasType \CSMabb{A}) \, (P_1 \parallel P_2)$
        and show there is a typing derivation for
        $P_1 \parallel (\restr s \hasType \CSMabb{A}) \, P_2$.
        Applying inversion twice yields

        \vspace{-2ex}
        { \footnotesize
        \begin{mathpar}
            \inferrule*[right=\procTypingRestr]{
                \inferrule*[left=\procTypingParallel]{
                    \typingContextOne
                        \typingContextCat
                        \typingContextTwo_1
                    \types
                    P_1
                    \\
                    \typingContextOne
                        \typingContextCat
                        \typingContextTwo_2,
                        \typingContextTwo_s
                    \types
                    P_2
                }{
                    \typingContextOne
                        \typingContextCat
                        \typingContextTwo_1,
                        \typingContextTwo_2,
                        \typingContextTwo_s
                    \types
                    P_1 \parallel P_2
                }
(\vec{q}, \xi) \in \reach(\CSMabb{A})
                \\
                \typingContextTwo_s =
                    \set{s[\procA] \hasType \vec{q}_\procA}_{\procA \in \ProcsOf{\CSMabb{A}}}
                \\
}{
                \typingContextOne
                    \typingContextCat
                    \typingContextTwo_1,
                    \typingContextTwo_2
                    \types
                    (\restr s \hasType \CSMabb{A}) \, (P_1 \parallel P_2)
            }
        \end{mathpar}
        }

        In the above deriviation, we used $\typingContextTwo_s$ in the typing derivation for
        $
            \typingContextOne
                \typingContextCat
                \typingContextTwo_2,
                \typingContextTwo_s
            \types
            P_2
        $.
        Let us explain why we can assume this. 
        By definition, these only contain type bindings related to~$s$, which does not occur in $P_1$ by assumption.
        There might exist a typing derivation where parts of $\typingContextTwo_s$ appear in the typing derivation for $P_1$ but these can only removed with the rules
        \procTypingEnd. 
        Hence, such derivations can be mimicked in the typing derivation for $P_2$, justifying our treatment of $\typingContextTwo_s$ and~$\typingContextThree_s$.
        We construct a typing derivation:

        \vspace{-2ex}
        { \small
        \begin{mathpar}
            \inferrule*[right=\procTypingParallel]{
                \inferrule*[right=\procTypingRestr]{
(\vec{q}, \xi) \in \reach(\CSMabb{A})
                    \\
                    \typingContextTwo_s =
                        \set{s[\procA] \hasType \vec{q}_\procA}_{\procA \in \ProcsOf{\CSMabb{A}}}
                    \\
                    \typingContextOne
                        \typingContextCat
                        \typingContextTwo_2,
                        \typingContextTwo_s
                    \types
                    P_2
                }{
                    \typingContextOne
                        \typingContextCat
                        \typingContextTwo_2
                    \types
                    (\restr s \hasType \CSMabb{A}) \, P_2
                }
                \\
\typingContextOne
                    \typingContextCat
                    \typingContextTwo_1
                \types
                P_1
            }{
                \typingContextOne
                    \typingContextCat
                    \typingContextTwo_1,
                    \typingContextTwo_2
                    \types
                    P_1 \parallel (\restr s \hasType \CSMabb{A}) \, P_2
            }
        \end{mathpar}
        }

        All premises coincide with the ones of the original typing derivation, concluding this case.
        \\
        Second, we assume there is a typing derivation for
        $P_1 \parallel (\restr s \hasType \CSMabb{A}) \, P_2$
        and show there is a typing derivation for
        $(\restr s \hasType \CSMabb{A}) \, (P_1 \parallel P_2)$.
        The proof is analogous to the previous case but we do not need to reason about the treatment of $\typingContextTwo_s$ and $\typingContextThree_s$ but it suffices to show there is one typing derivation and we can choose the respective treatment.

 \item $(\restr s \hasType \CSMabb{A}) \, \zero
    \precongr
    \zero$: \\
We assume there is a typing derivation for
    $
        \typingContextOne
            \typingContextCat
            \typingContextTwo
        \types
        (\restr s \hasType \CSMabb{A}) \, \zero
    $.
    We show there is a typing derivation for
    $
        \typingContextOne
            \typingContextCat
            \typingContextTwo
        \types
        \zero
    $.
    By inversion, we know that \procTypingRestr is the last rule to be applied and we get one of the premises:
    $
        \typingContextOne
            \typingContextCat
            \typingContextTwo,
            \typingContextTwo_s
        \types
        \zero
    $
    with
    $
        \typingContextTwo_s =
            \set{s[\procA] \hasType \reach(\CSMabb{A}_{\procA})}_{\procA \in \ProcsOf{\CSMabb{A}}}
    $.
    By inversion,
    \procTypingZero
    and
    \procTypingEnd
    are the only rules that can be applied in the typing derivation for
    $
        \typingContextOne
            \typingContextCat
            \typingContextTwo,
            \typingContextTwo_s
        \types
        \zero
    $.
    Since
    \procTypingZero
    needs the second typing context to be empty, it is applied last and all other derivations are applications of
    \procTypingEnd.
    Therefore,
    $
        \typingContextTwo =
            \Union_{s' \in \SessionName}
            \set{s'[\procA] \hasType \vec{q}_\procA}_{\procA \in \Procs_{s'}}
    $
    for some set of sessions $\SessionName$ that does not contain $s$
    and
    $\EndState(\vec{q}_\procA)$
for every $\procA \in \Procs_{s'}$ and $s' \in \SessionName$.
    Therefore, we can also first apply \procTypingEnd $\card{\typingContextTwo}$ times and
    last \procTypingZero to obtain a typing derivation for
    $
        \typingContextOne
            \typingContextCat
            \typingContextTwo
        \types
        \zero
    $.

\end{itemize}
This concludes the proof of admissibility of $\precongr$.
\end{proof}

\begin{lemma}
\label{lm:every-variable-in-process-needs-typing}
Let $P$ be a process.
If
$\typingContextOne
\typingContextCat
\typingContextTwo
\types P$
and $x$ is not in $\typingContextTwo$,
then $x$ cannot occur in~$P$.
\end{lemma}
\begin{proof}
Towards a contradiction, assume that $x$ occurs in $P$.
Then, at some point in the typing derivation
\[
    \typingContextOne
    \typingContextCat
    \typingContextTwo
    \types P
\]
one of these two rules applies to handle $x$:
\procTypingProcName or
\procTypingMixCh. 
Each requires all variables to occur in their respective typing contexts, yielding a contradiction.
\end{proof}

\begin{lemma}[Substitution Lemma]
\label{lm:substitution-lemma}
Let $P$ be a process.
For all $L$,
if it holds that
      $\typingContextOne
           \typingContextCat
           \typingContextTwo,
           x \hasType L
       \types
       P$,
then
       $
        \typingContextOne
           \typingContextCat
           \typingContextTwo,
           v \hasType L
       \types
       P[v / x]
       \enspace .
       $
\end{lemma}
\begin{proof}
We do an induction on the depth of the typing derivation and do a case analysis on the last applied rule of the derivation.

For the induction base, we consider the (only) rule with depth~$0$: \procTypingZero.
In this case, the second typing context for \procTypingZero is empty, which contradicts our assumption that $x \hasType L$ or $v \hasType L$ appear.

For the induction step, the induction hypothesis yields that the claim holds for typing derivations of smaller depth.

\begin{itemize}
 \item \procTypingProcName: 
    We have that

    \vspace{-2ex}
    { \small
    \begin{mathpar}
        \inferrule*[right=\procTypingProcName]{
            \typingContextOne(\pn{Q}) = L_1,\dots,L_n \\
        }{
            \typingContextOne
                \typingContextCat
                c_1 \hasType L_1,
                \ldots,
                c_{i-1} \hasType L_{i-1},
                x \hasType L_{i},
                c_{i+1} \hasType L_{i+1},
                \ldots
                c_n \hasType L_n
                \types
                \pn{Q}[\vec{c}]
        }
    \end{mathpar}
    }

    It is straightforward that we need to show precisely the same premise for the desired typing derivation:
    \begin{align*}
        \typingContextOne
            \typingContextCat
            & \; 
            c_1 \hasType L_1,
            \ldots,
            c_{i-1} \hasType L_{i-1},
            \\ & \;
            v \hasType L_{i},
            c_{i+1} \hasType L_{i+1},
            \ldots,
            c_n \hasType L_n
            \types
            (\pn{Q}[\vec{c}])[v / x]
        \enspace .
    \end{align*}

 \item \procTypingEnd: We have two cases. \\
    First, we have

    \vspace{-2ex}
    { \small
    \begin{mathpar}
        \inferrule*[right=\procTypingEnd]{
            \typingContextOne
            \typingContextCat
            \typingContextTwo
            \types P \\
\EndState(q)
}{
            \typingContextOne
            \typingContextCat
            x \hasType q,
            \typingContextTwo
                \types
            P
        }
    \end{mathpar}
    }

    We show that

    \vspace{-2ex}
    { \small
    \begin{mathpar}
        \inferrule*[right=\procTypingEnd]{
            \typingContextOne
            \typingContextCat
            \typingContextTwo
            \types P[v/x] \\
\EndState(q)
}{
            \typingContextOne
            \typingContextCat
            v \hasType q,
            \typingContextTwo
                \types
            P[v/x]
        }
    \end{mathpar}
    }

    For the first typing derivation, we have $x \hasType q, \typingContextTwo$ as typing context.
    By the fact that $\typingContextTwo$ is a typing context (and no syntactic typing context),
    $x$ does not occur in $\typingContextTwo$.
    By inversion, we have
    $
            \typingContextOne
            \typingContextCat
            \typingContextTwo
            \types P
    $.
    Thus, $x$ cannot occur in $P$.
    If it did, $P$ could only be typed with a typing context with $x$, which does not occur in~$\typingContextTwo$, given by contraposition of
    \cref{lm:every-variable-in-process-needs-typing}.
    Hence $P = P[v/x]$ and, thus, both premises coincide.

    Second, we have

    \vspace{-2ex}
    { \small
    \begin{mathpar}
        \inferrule*[right=\procTypingEnd]{
            \typingContextOne
            \typingContextCat
            \typingContextTwo,
            x \hasType L
            \types P \\
\EndState(q)
}{
            \typingContextOne
            \typingContextCat
            c \hasType q,
            \typingContextTwo,
            x \hasType L
                \types
            P
        }
    \end{mathpar}
    }

    We show that

    \vspace{-2ex}
    { \small
    \begin{mathpar}
        \inferrule*[right=\procTypingEnd]{
            \typingContextOne
            \typingContextCat
            \typingContextTwo,
            v \hasType L
            \types P[v/x] \\
\EndState(q)
}{
            \typingContextOne
            \typingContextCat
            c \hasType q,
            \typingContextTwo,
            v \hasType L
                \types
            P[v/x]
        }
    \end{mathpar}
    }

    By inversion of the first typing derivation, we know that both premises hold.
    The second premise is the same for both derivations.
    For the first premise, the induction hypothesis applies.

 \item \procTypingMixCh:  
    We do a case analysis if $x = c$, $x = c_k$ for some $k \in I$ or neither of both.

    \begin{itemize}
    \item $x \neq c$ and $x \neq c_k$ for any $k \in I$: \\
        In this case, we can apply inversion and the induction hypothesis applies to all cases of the second and third premise. 
    \item $x = c$: \\
We have 

        { \small
        \begin{mathpar}
        \inferrule*[right=\procTypingMixCh]{
\delta(q) =
            \set{(\msgFromTo{\procA}{\procB_i}{\labelAndType{l_i}{L_i}}, q_i) \mid i \in I}
            \dunion
            \set{(\msgFromTo{\procB_j}{\procA}{\labelAndType{l_j}{L_j}}, q_j) \mid j \in J}
            \\
\meta{\forall i \in I \st}
            \typingContextOne \typingContextCat
                \typingContextTwo , x \hasType q_i,
                \set{c_j \hasType L_j}_{j \in I\setminus \set{i}}
                \types P_i 
            \\
            \meta{\forall j \in J \st}
            \typingContextOne \typingContextCat
                \typingContextTwo,
                x \hasType q_j,
                \set{c_i \hasType L_i}_{i \in I}, 
                y_j \hasType L_j
                \types P_j \\
        }{
            \typingContextOne
            \typingContextCat
            \typingContextTwo,
            x \hasType q,
            \set{c_i \hasType L_i}_{i \in I}
            \\
                \types
            \MixCh_{i \in I} x[\procB_i] \sendOp \labelAndMsg{l_i}{c_i} \seq P_i
            + 
            \MixCh_{j \in J} x[\procB_j] \recOp \labelAndVar{l_j}{y_j} \seq P_j
        }
        \end{mathpar}
        }
        We need to show that 
        { \small
        \begin{mathpar}
        \inferrule*[right=\procTypingMixCh]{
\delta(q) =
            \set{(\msgFromTo{\procA}{\procB_i}{\labelAndType{l_i}{L_i}}, q_i) \mid i \in I}
            \dunion
            \set{(\msgFromTo{\procB_j}{\procA}{\labelAndType{l_j}{L_j}}, q_j) \mid j \in J}
            \\
\meta{\forall i \in I \st}
            \typingContextOne \typingContextCat
                \typingContextTwo , v \hasType q_i,
                \set{c_j \hasType L_j}_{j \in I\setminus \set{i}}
                \types P_i [v / x]
            \\
            \meta{\forall j \in J \st}
            \typingContextOne \typingContextCat
                \typingContextTwo,
                v \hasType q_j,
                \set{c_i \hasType L_i}_{i \in I}, 
                y_j \hasType L_j
                \types P_j [v / x] \\
        }{
            \typingContextOne
            \typingContextCat
            \typingContextTwo,
            v \hasType q,
            \set{c_i \hasType L_i}_{i \in I}
            \\
                \types
            (
            \MixCh_{i \in I} x[\procB_i] \sendOp \labelAndMsg{l_i}{c_i} \seq P_i
            + 
            \MixCh_{j \in J} x[\procB_j] \recOp \labelAndVar{l_j}{y_j} \seq P_j
            ) [v / x]
        }
        \end{mathpar}
        }

        The first permise is the same while second and third premise follow by the induction hypothesis.
    \item $x = c_k$ for $k \in I$: \\
        We have 
        { \small
        \begin{mathpar}
        \inferrule*[right=\procTypingMixCh]{
\delta(q) =
            \set{(\msgFromTo{\procA}{\procB_i}{\labelAndType{l_i}{L_i}}, q_i) \mid i \in I}
            \dunion
            \set{(\msgFromTo{\procB_j}{\procA}{\labelAndType{l_j}{L_j}}, q_j) \mid j \in J}
            \\
\meta{\forall i \in I \st}
            \typingContextOne \typingContextCat
                \typingContextTwo , c \hasType q_i,
                x \hasType L,
                \set{c_{i'} \hasType L_{i'}}_{{i'} \in I\setminus \set{i,k}}
                \types P_i 
            \\
            \\
            \typingContextOne \typingContextCat
                \typingContextTwo , c \hasType q_k,
\set{c_{i} \hasType L_{i}}_{i \in I\setminus \set{k}}
                \types P_k
            \\
            \\
            \meta{\forall j \in J \st}
            \typingContextOne \typingContextCat
                \typingContextTwo,
                c \hasType q_j,
                \set{c_i \hasType L_i}_{i \in I}, 
                y_j \hasType L_j
                \types P_j \\
        }{
            \typingContextOne
            \typingContextCat
            \typingContextTwo,
            c \hasType q,
            x \hasType L,
            \set{c_i \hasType L_i}_{i \in I \setminus \set{k}}
            \\
                \types
            \MixCh_{i \in I} c[\procB_i] \sendOp \labelAndMsg{l_i}{c_i} \seq P_i
            + 
            \MixCh_{j \in J} c[\procB_j] ? \labelAndVar{l_j}{y_j} \seq P_j
        }
        \end{mathpar}
        }
        We show that
        { \small
        \begin{mathpar}
        \inferrule*[right=\procTypingMixCh]{
\delta(q) =
            \set{(\msgFromTo{\procA}{\procB_i}{\labelAndType{l_i}{L_i}}, q_i) \mid i \in I}
            \dunion
            \set{(\msgFromTo{\procB_j}{\procA}{\labelAndType{l_j}{L_j}}, q_j) \mid j \in J}
            \\
\meta{\forall i \in I \setminus \set{k} \st}
            \typingContextOne \typingContextCat
                \typingContextTwo , c \hasType q_i,
                v \hasType L,
                \set{c_{i'} \hasType L_{i'}}_{{i'} \in I\setminus \set{i,k}}
                \types P_i [v / k]
            \\
            \typingContextOne \typingContextCat
                \typingContextTwo , c \hasType q_k,
\set{c_{i} \hasType L_{i}}_{i \in I\setminus \set{k}}
                \types P_k [v / x]
            \\
            \meta{\forall j \in J \st}
            \typingContextOne \typingContextCat
                \typingContextTwo,
                c \hasType q_j,
                \set{c_i \hasType L_i}_{i \in I}, 
                y_j \hasType L_j
                \types P_j [v / x] \\
        }{
            \typingContextOne
            \typingContextCat
            \typingContextTwo,
            c \hasType q,
            x \hasType L,
            \set{c_i \hasType L_i}_{i \in I \setminus \set{k}}
            \\
                \types
            (
            \MixCh_{i \in I} c[\procB_i] \sendOp \labelAndMsg{l_i}{c_i} \seq P_i
            + 
            \MixCh_{j \in J} c[\procB_j] ? \labelAndVar{l_j}{y_j} \seq P_j
            ) [v / x]
        }
        \end{mathpar}
        }
    \end{itemize}

    The first premise is the same.
    The second and fourth premise all follow by inversion and applying the induction hypotheses.
    For the third premise, we claim that $x$ cannot occur in~$P_k$.
    In the conclusion of the first derivation, we have the following typing context: 
    $
            \typingContextTwo,
            c \hasType q,
            x \hasType L,
            \set{c_i \hasType L_i}_{i \in I \setminus \set{k}}
    $.
    Thus, by assumption that each element has at most one type in a typing context, we know that $x$ cannot occur in
    $
            \typingContextTwo,
            c \hasType q,
            \set{c_i \hasType L_i}_{i \in I \setminus \set{k}}
    $.
    With \cref{lm:every-variable-in-process-needs-typing}, $x$ hence cannot occur in~$P_k$.
    Thus, $P_k[v/x] = P_k$ and the third premise in both typing derivations coincide, concluding this~case.

 \item \procTypingParallel: \\
    There are two symmetric cases.
    We only consider one of both.
    For this, we~have

    \vspace{-2ex}
    { \small
    \begin{mathpar}
        \inferrule*[right=\procTypingParallel]{
            \typingContextOne
                \typingContextCat
                \typingContextTwo_1,
                x \hasType L
                \types
                P_1 \\
\typingContextOne
                \typingContextCat
                \typingContextTwo_2
                \types
                P_2 \\
        }{
            \typingContextOne
                \typingContextCat
                \typingContextTwo_1,
                x \hasType L,
                \typingContextTwo_2
                \types
                P_1 \parallel P_2
        }
    \end{mathpar}
    }

    We show that there is a typing derivation for
    \[
            \typingContextOne
                \typingContextCat
                \typingContextTwo_1,
                v \hasType L,
                \typingContextTwo_2
                \types
                (P_1 \parallel P_2)[v / x]
    \]
    By our assumption that typing contexts have at most one type per element, $\typingContextTwo_1$ and $\typingContextTwo_2$ cannot share any names.
    By inversion,
    we have
           $ \typingContextOne
                \typingContextCat
                \typingContextTwo_2
                \types
                P_2$.
    Thus, by \cref{lm:every-variable-in-process-needs-typing}, $x$~cannot occur in $P_2$.
    Hence, $(P_1 \parallel P_2)[v / x]
                =
                P_1[v/x] \parallel P_2$.
    It suffices to show that the following typing derivation exists:

    \vspace{-2ex}
    { \small
    \begin{mathpar}
        \inferrule*[right=\procTypingParallel]{
            \typingContextOne
                \typingContextCat
                \typingContextTwo_1,
                v \hasType L,
                \types
                P_1 \\
\typingContextOne
                \typingContextCat
                \typingContextTwo_2
                \types
                P_2 \\
        }{
            \typingContextOne
                \typingContextCat
                \typingContextTwo_1,
                v \hasType L,
                \typingContextTwo_2
                \types
                P_1[v/x] \parallel P_2
        }
    \end{mathpar}
    }

    The second premise is the same as in the original typing derivation.
    The first premise can be obtained by inversion on the original typing derivation and applying the induction hypothesis.

 \item \procTypingRestr: \\
    We have a typing derivation for
    \[
        \typingContextOne
            \typingContextCat
            \typingContextTwo,
            x \hasType L,
            \typingContextTwo_s
            \types (\restr s \hasType \CSMabb{A}) \, P
        \enspace .
    \]
    This case follows easily from inversion and applying the induction hypothesis to obtain a typing derivation for
    \[
        \typingContextOne
            \typingContextCat
            \typingContextTwo,
            v \hasType L,
            \typingContextTwo_s
            \types (\restr s \hasType \CSMabb{A}) \, P[v/x]
        \enspace .
    \]

\end{itemize}
This concludes the proof of the substitution lemma.
\end{proof}

\subjectReduction*
\begin{proof}
We do an induction on the depth of the derivation tree for
 $P \redto P'$ 
 \ref{lm:subject-reduction-assumption-4}. 

For $d = 1$, $\procReductionExc$ is the only rule that applies. 
By inversion on~\ref{lm:subject-reduction-assumption-4}, we get: 
{ \small 
\begin{mathpar}
\inferrule*[right=\procReductionExc]{
        I \inters J = \emptyset \\
        k \in I  \\
\prefixPi_k = s[\procA][\procC] \sendOp \labelAndMsg{l_k}{c_k} \\
        k' \in J  \\
\prefixPi_{k'} = s[\procC][\procA] \recOp \labelAndVar{l_{k'}}{y_{k'}} \\
        l_k = l_{k'} 
}{
        \MixCh_{i \in I} \prefixPi_i \seq P_i
        \parallel
        \MixCh_{j \in J} \prefixPi_j \seq P_j
            \redto
        P_k
        \parallel
        P_{k'}[c_k / y_{k'}]
}
\end{mathpar}
} 

Hence, $P = Q_1 \parallel Q_2$ with
$
Q_1 = \MixCh_{i \in I} 
\prefixPi_i \seq P_i
$
and 
$
Q_2 = \MixCh_{j \in J} 
\prefixPi_j \seq P_j 
$ 
and $P' = P_{k'}[c_k / y_{k'}]$.
We do inversion on 
\ref{lm:subject-reduction-assumption-2}: 
\begin{mathpar}
  \inferrule*[right=\procTypingParallel]{
      (i)\; \typingContextOne
          \typingContextCat
          \typingContextTwo_1
          \types
          Q_1 \\
(ii)\; \typingContextOne
          \typingContextCat
          \typingContextTwo_2
          \types
          Q_2 \\
  }{
      \typingContextOne
          \typingContextCat
          \typingContextTwo_1,
          \typingContextTwo_2
          \types
          Q_1 \parallel Q_2
  }
\end{mathpar}
We do inversion on $(i)$ and $(ii)$: 

{ \small 
\begin{mathpar}
  \inferrule*[right=\procTypingMixCh]{
\delta(\hat{q}_1) =
      \set{(\msgFromTo{\procA}{\procB_i}{\labelAndType{l_i}{L_i}}, q_i) \mid i \in I_1}
      \dunion
      \set{(\msgFromTo{\procB_j}{\procA}{\labelAndType{l_j}{L_j}}, q_j) \mid j \in J_1}
      \\
\meta{\forall i \in I_1 \st}
      \typingContextOne \typingContextCat
          \hat{\typingContextTwo}_1,
           s[\procA] \hasType q_i,
          \set{c_j \hasType L_j}_{j \in I_1\setminus \set{i}}
          \types P_i 
      \\
      \meta{\forall j \in J_1 \st}
      \typingContextOne \typingContextCat
          \hat{\typingContextTwo}_1,
          s[\procA] \hasType q_j,
          \set{c_i \hasType L_i}_{i \in I_1}, 
          y_j \hasType L_j
          \types P_j \\
  }{
      (i) \;
      \typingContextOne
      \typingContextCat
      \hat{\typingContextTwo}_1,
      s[\procA] \hasType \hat{q}_1,
      \set{c_i \hasType L_i}_{i \in I_1}
      \\
          \types
      \MixCh_{i \in I_1} s[\procA][\procB_i] \sendOp \labelAndMsg{l_i}{c_i} \seq P_i
      + 
      \MixCh_{j \in J_1} s[\procA][\procB_j] \recOp \labelAndVar{l_j}{y_j} \seq P_j
  }

  \inferrule*[right=\procTypingMixCh]{
\delta(\hat{q}_2) =
      \set{(\msgFromTo{\procC}{\procD_i}{\labelAndType{l_i}{L_i}}, q_i') \mid i \in I_2}
      \dunion
      \set{(\msgFromTo{\procD_j}{\procC}{\labelAndType{l_j}{L_j}}, q_j') \mid j \in J_2}
      \\
\meta{\forall i \in I_2 \st}
      \typingContextOne \typingContextCat
          \hat{\typingContextTwo}_2,
           s[\procC] \hasType q_i',
          \set{c_j \hasType L_j}_{j \in I_2\setminus \set{i}}
          \types P'_i 
      \\
      \meta{\forall j \in J_2 \st}
      \typingContextOne \typingContextCat
          \hat{\typingContextTwo}_2,
          s[\procC] \hasType q_j',
          \set{c_i \hasType L_i}_{i \in I_2}, 
          y_j \hasType L_j
          \types P'_j \\
  }{
      (ii) \;
      \typingContextOne
      \typingContextCat
      \hat{\typingContextTwo}_2,
      s[\procC] \hasType \hat{q}_2,
      \set{c_i \hasType L_i}_{i \in I_2}
      \\
          \types
      \MixCh_{i \in I_2} s[\procC][\procD_i] \sendOp \labelAndMsg{l_i}{c_i} \seq P'_i
      + 
      \MixCh_{j \in J_2} s[\procC][\procD_j] \recOp \labelAndVar{l_j}{y_j} \seq P'_j
  }
\end{mathpar}
}

where $I = I_1 \dunion J_1$ and $J = I_2 \dunion J_2$ as well as 
\[
\typingContextTwo_1 = 
    \hat{\typingContextTwo}_1, 
    s[\procA] \hasType \hat{q}_1,
    \set{c_i \hasType L_i}_{i \in I_1}, \text{  and  }
 \typingContextTwo_2 = 
    \hat{\typingContextTwo}_2, 
    s[\procC] \hasType \hat{q}_2,
    \set{c_i \hasType L_i}_{i \in I_2} .
\]
Hence, we have 
\[ 
\typingContextTwo = 
    \hat{\typingContextTwo}_1, 
    s[\procA] \hasType \hat{q}_1,
    \set{c_i \hasType L_i}_{i \in I_1},
    \hat{\typingContextTwo}_2, 
    s[\procC] \hasType \hat{q}_2,
    \set{c_i \hasType L_i}_{i \in I_2}
\] 
and define 
\[
 \typingContextTwo' = 
    \hat{\typingContextTwo}_1, 
    s[\procA] \hasType q_k,
    \set{c_i \hasType L_i}_{i \in I_1},
    \hat{\typingContextTwo}_2, 
    s[\procC] \hasType q'_{k'},
    \set{c_i \hasType L_i}_{i \in I_2} \enspace .
\]
We need to show the following two claims. 
By assumption, we know that the label $l_k = l_{k'}$ determines the type: here, we use that $L_k = L_{k'}$ in both cases. 
\begin{itemize}
    \item $\typingContextOne, \typingContextTwo' \types P_k \parallel P'_{k'}[c_k / y_{k'}]$: 
        \begin{mathpar}[\phantom{soi}]
            \inferrule*[right=\procTypingParallel]{
                (a) \;
                \typingContextOne 
                \typingContextCat
                \hat{\typingContextTwo}_1, 
                s[\procA] \hasType q_k,
                \set{c_i \hasType L_i}_{i \in I_1 \setminus \set{k}},
                    \types P_k 
                \\
                (b) \;
                \typingContextOne 
                \typingContextCat
                \hat{\typingContextTwo}_2, 
                s[\procC] \hasType q_{k'},
                \set{c_i \hasType L_i}_{i \in I_2},
                c_k \hasType L_k 
                    \types P'_{k'}[c_k / y_{k'}]
            }{
                \typingContextOne
                \typingContextCat
                \typingContextTwo' 
                \types P_k \parallel P'_{k'}[c_k / y_{k'}]
            }
        \end{mathpar}
        Derivation $(a)$ follows directly from the second premise of the inversion on $(i)$. 
        Derivation~$(b)$ follows from the third premise of the inversion of $(ii)$ and application of the Substitution Lemma (\cref{lm:substitution-lemma}) and the fact that $L_k = L_{k'}$. 
        
    \item $\typingContextTwo \redto \typingContextTwo'$ 
          (or $\typingContextTwo = \typingContextTwo'$): 
        \begin{mathpar}
            \inferrule*[right=\typingReductionMixCh]{
                q'_1
                    \xrightarrow{\msgFromToNS{\procA}{\procB_i}{\labelAndType{l_k}{L_k}}}
                q_k
                \\
                q'_2
                    \xrightarrow{\msgFromToNS{\procD_j}{\procC}{\labelAndType{l_{k'}}{L_{k'}}}}
                q_{k'}
            }{
                \typingContextTwo
                \redto
                \typingContextTwo'
            }
        \end{mathpar}
        From the inversion on \ref{lm:subject-reduction-assumption-4}, we get 
        $\procB_i = \procC$, $\procA = \procD_j$ and $k = k'$, which suffices with $L_k = L_{k'}$. 
\end{itemize}

For the induction step, we assume that the claim holds for all $P$ and $P'$ with a derivation tree of depth at most $d$ for $P \redto P'$. 
We do a case analysis on the last applied rule of a derivation tree of depth $d+1$. 
\begin{itemize}
\item $\procReductionProcName$: \\
    We have 
    \begin{mathpar}
    \inferrule*[right=\procReductionProcName]{
            \Defs(\pn{Q}, \vec{c})
            \parallel
            P_2
            \redto
            P'
    }{
            \pn{Q}[\vec{c}]
            \parallel
            P_2
            \redto
            P'_2
    }
        \end{mathpar}
    In addition, we know that 
    \ref{lm:subject-reduction-assumption-1} to 
    \ref{lm:subject-reduction-assumption-4} hold for 
    $P = \pn{Q}[\vec{c}] \parallel P_2$ and $P' = P'_2$. 
    We need to show that there is $\typingContextTwo'$ with
    $\typingContextTwo \redto \typingContextTwo'$ 
    or
    $\typingContextTwo = \typingContextTwo'$ 
    such that
    $\typingContextOne \typingContextCat \typingContextTwo' \types P'_2$. 
    We want to apply the induction hypothesis for 
    $P = \Defs(\pn{Q}, \vec{c}) \parallel P_2$ and $P' = P'_2$. 
    We need to show that all conditions 
    \ref{lm:subject-reduction-assumption-1} to 
    \ref{lm:subject-reduction-assumption-4} hold. 
    It is obvious that 
    \ref{lm:subject-reduction-assumption-1} and 
    \ref{lm:subject-reduction-assumption-3} do not change so they hold. 
    By assumption, 
    $\Defs(\pn{Q}, \vec{c}) \parallel P_2 \redto P'_2$ has a derivation tree of depth $d$ and hence also satisfies 
    \ref{lm:subject-reduction-assumption-4}. 
    For \ref{lm:subject-reduction-assumption-2}, 
    we need to show that 
    \begin{mathpar}
        \inferrule*[right=\procTypingParallel]{
            (a) \;
            \typingContextOne
                \typingContextCat
                \vec{c} \hasType \vec{L}
                \types
                \Defs(\pn{Q}, \vec{c}) \\
(b) \;
            \typingContextOne
                \typingContextCat
                \typingContextTwo_2
                \types
                P_2 \\
        }{
            \typingContextOne
                \typingContextCat
                \vec{c} \hasType \vec{L},
                \typingContextTwo_2
                \types
                \Defs(\pn{Q}, \vec{c}) 
                \parallel 
                P_2
        }
    \end{mathpar}
    We do an inversion on the original 
    \ref{lm:subject-reduction-assumption-2}:
    \begin{mathpar}
        \inferrule*[right=\procTypingParallel]{
            \inferrule*[left=\procTypingProcName]{
                \typingContextOne(\pn{Q}) = \vec{L}
            }{
                \typingContextOne
                    \typingContextCat
                    \vec{c} \hasType \vec{L}
                    \types
                    \pn{Q}[\vec{c}] \\
            }
\typingContextOne
                \typingContextCat
                \typingContextTwo_2
                \types
                P_2 \\
        }{
            \typingContextOne
                \typingContextCat
                \vec{c} \hasType \vec{L},
                \typingContextTwo_2
                \types
                \pn{Q}[\vec{c}]
                \parallel
                P_2
        }
    \end{mathpar}
    From this, $(b)$ immediately follows. 
    For $(a)$, we rely on the well-typedness of process definitions 
    \ref{lm:subject-reduction-assumption-1}. 
    By assumption, $\Defs$ has the following shape: 
    $ \Defs = \Defs_1; (\pn{Q}[\vec{x} = P]); \Defs_2 $ 
    such that $\Defs(\pn{Q}, \vec{c}) = P [\vec{c} / \vec{x}]$.
    Hence our goal $(a)$ is to provide a typing derivation for 
    $\typingContextOne \typingContextCat \vec{c} \hasType \vec{L} \types P[\vec{c} / \vec{x}]$. 
    By applying inversion to \ref{lm:subject-reduction-assumption-1} $\card{\Defs_1}$ times, we get 
    $\typingContextOne \typingContextCat \vec{x} \hasType \vec{L} \types P$. 
    By applying the Substitution Lemma (\cref{lm:substitution-lemma}) $\card{\vec{x}}$ times, we obtain $(a)$. 
    With that, we have shown that the induction hypothesis applies, yielding its conclusion:
    there is $\typingContextTwo'$ 
    with
    $\typingContextTwo \redto \typingContextTwo'$ or
    $\typingContextTwo = \typingContextTwo'$ such that
    $\typingContextOne \typingContextCat \typingContextTwo' \types P'_2$. 
    This is precisely our goal, which concludes this~case. 
\item $\procReductionContext$: \\
    We do a case analysis on the different kind of contexts. 
    \begin{itemize}
    \item $\redContext \parallel Q$: \\
        We have 
        \begin{mathpar}
        \inferrule*[right=\procReductionContext]{
                P_1 \redto P'_1
        }{
                P_1 \parallel Q \redto P'_1 \parallel Q
        }
        \end{mathpar}
        In addition, we know that 
        \ref{lm:subject-reduction-assumption-1} to 
        \ref{lm:subject-reduction-assumption-4} hold for 
        $P = P_1 \parallel Q$ and $P' = P'_1 \parallel Q$. 
        We need to show that there is $\typingContextTwo'$ 
        $\typingContextTwo \redto \typingContextTwo'$ 
        or
        $\typingContextTwo = \typingContextTwo'$ 
        such that 
        $\typingContextOne \typingContextCat \typingContextTwo' \types P'_1 \parallel Q$. 
        We want to apply the induction hypothesis for 
        $P = P_1$ and $P' = P'_1$. 
        We need to show that all conditions 
        \ref{lm:subject-reduction-assumption-1} to 
        \ref{lm:subject-reduction-assumption-4} hold. 
        It is obvious that 
        \ref{lm:subject-reduction-assumption-1} and 
        \ref{lm:subject-reduction-assumption-3} do not change so they hold. 
        By assumption, 
        $P_1 \parallel Q \redto P'_1 \parallel Q$ has a derivation tree of depth $d$ and hence also satisfies 
        \ref{lm:subject-reduction-assumption-4}. 
        For~\ref{lm:subject-reduction-assumption-2}, we need to show that 
        $\typingContextOne_1 \typingContextCat \typingContextTwo \types P_1$, which follows directly from inversion on the original~\ref{lm:subject-reduction-assumption-2}: 
        \begin{mathpar}
        \inferrule*[right=\procTypingParallel]{
            \typingContextOne
                \typingContextCat
                \typingContextTwo_1
                \types
                P_1 \\
\typingContextOne
                \typingContextCat
                \typingContextTwo_2
                \types
                Q \\
        }{
            \typingContextOne
                \typingContextCat
                \typingContextTwo_1,
                \typingContextTwo_2
                \types
                P_1
                \parallel
                Q
        }
        \end{mathpar}
        From the conclusion of the induction hypothesis, 
        there exists 
        $\typingContextTwo'_1$ with
        $\typingContextTwo_1 \redto \typingContextTwo'_1$ 
        or
        $\typingContextTwo_1 = \typingContextTwo'_1$ such that
        $\typingContextOne \typingContextCat \typingContextTwo'_1 \types P'_1$. 
        Let us assume that $\typingContextTwo_1 \neq \typingContextTwo'_1$. 
        With \cref{lm:typing-reduction-cong}, it follows that 
        $\typingContextTwo_1, \typingContextTwo_2
        \redto 
        \typingContextTwo'_1, \typingContextTwo_2$.
        Again for both cases, \ie
        $\typingContextTwo_1 = \typingContextTwo'_1$ or not, 
        this leaves us to prove that 
        $\typingContextOne \typingContextCat \typingContextTwo'_1, \typingContextTwo_2 \types P'_1 \parallel Q$. 
        This is easy to see: 
        \begin{mathpar}
        \inferrule*[right=\procTypingParallel]{
            \typingContextOne
                \typingContextCat
                \typingContextTwo'_1
                \types
                P'_1 \\
\typingContextOne
                \typingContextCat
                \typingContextTwo_2
                \types
                Q \\
        }{
            \typingContextOne
                \typingContextCat
                \typingContextTwo'_1,
                \typingContextTwo_2
                \types
                P'_1
                \parallel
                Q
        }
        \end{mathpar}
        where the first premise is given by the conclusion of the induction hypothesis and the second premise is the same as the inversion above. 
    \item $Q \parallel \redContext$: \\
    The proof is analogous to the previous case. 
    \item $(\restr s \hasType \CSMabb{A}) \, \redContext$: \\
        We have 
        \begin{mathpar}
        \inferrule*[right=\procReductionContext]{
                Q \redto Q'
        }{
                (\restr s \hasType \CSMabb{A}) \, Q 
                \redto 
                (\restr s \hasType \CSMabb{A}) \, Q'
        }
        \end{mathpar}
        In addition, we know that 
        \ref{lm:subject-reduction-assumption-1} to 
        \ref{lm:subject-reduction-assumption-4} hold for 
        $P = (\restr s \hasType \CSMabb{A}) \, Q$ and $P' = (\restr s \hasType \CSMabb{A}) \, Q'$. 
        We need to show that there is $\typingContextTwo'$ with
        $\typingContextTwo \redto \typingContextTwo'$ 
        or
        $\typingContextTwo = \typingContextTwo'$ 
        such that  
        $\typingContextOne \typingContextCat \typingContextTwo' \types (\restr s \hasType \CSMabb{A}) \, Q'$. 
        We want to apply the induction hypothesis for 
        $P = Q$ and $P' = Q'$. 
        We need to show that all conditions 
        \ref{lm:subject-reduction-assumption-1} to 
        \ref{lm:subject-reduction-assumption-4} hold. 
        It is obvious that \ref{lm:subject-reduction-assumption-1} is the same as before and that 
        \ref{lm:subject-reduction-assumption-4} holds by assumption. 
        Let us consider \ref{lm:subject-reduction-assumption-2}. 
        We do inversion on the original 
        \ref{lm:subject-reduction-assumption-2} and obtain: 
        \begin{mathpar}
        \inferrule*[right=\procTypingRestr]{
\vec{q} \in \reach(\CSMabb{A})
            \\
            \typingContextTwo_s =
                \set{s[\procA] \hasType \vec{q}_\procA}_{\procA \in \ProcsOf{\CSMabb{A}}}
            \\
\typingContextOne
                \typingContextCat
                \typingContextTwo,
                \typingContextTwo_s
            \types
            Q
}{
            \typingContextOne
                \typingContextCat
                \typingContextTwo
                \types
                (\restr s \hasType \CSMabb{A})\, Q
        }
        \end{mathpar}
        Technically, we would need to generalize the typing context $\typingContextTwo$ in our claim first to do the following. 
        This is straightforward and only applies in this single case and hence omitted in the general case. 
        We choose $\typingContextTwo, \typingContextTwo_s$ as typing context and immediately obtain \ref{lm:subject-reduction-assumption-2}, with the side condition met by $\typingContextTwo$ and $\typingContextTwo_s$ individually. 
        For \ref{lm:subject-reduction-assumption-3}, 
        we observe that the condition is met for all sessions $s' \in \mathcal{S} \setminus \set{s}$ by inversion on the original \ref{lm:subject-reduction-assumption-2} while it is met for $s$ by the above inversion. 
        The conclusion of the induction hypothesis yields that there is 
        $\typingContextTwo'$ and $\typingContextTwo'_s$ 
        such that 
        $\typingContextTwo, \typingContextTwo_s 
        \redto 
         \typingContextTwo', \typingContextTwo'_s$
        or
        $\typingContextTwo, \typingContextTwo_s 
         = 
         \typingContextTwo', \typingContextTwo'_s$
        such that  
        $\typingContextOne \typingContextCat \typingContextTwo', \typingContextTwo'_s \types Q'$. 
        We have to show two claims. 

        First, we show that 
        \begin{mathpar}
        \inferrule*[right=\procTypingRestr]{
\vec{q} \in \reach(\CSMabb{A})
            \\
            \typingContextTwo_s =
                \set{s[\procA] \hasType \vec{q}_\procA}_{\procA \in \ProcsOf{\CSMabb{A}}}
            \\
\typingContextOne
                \typingContextCat
                \typingContextTwo',
                \typingContextTwo'_s
            \types
            Q'
}{
            \typingContextOne
                \typingContextCat
                \typingContextTwo'
                \types
                (\restr s \hasType \CSMabb{A})\, Q'
        }
        \end{mathpar}
        While the first two premises are the same as above, the third premise follows from the conclusion of the induction hypothesis. 

        Second, we need to show that 
        $\typingContextTwo \redto \typingContextTwo'$ 
        or
        $\typingContextTwo = \typingContextTwo'$.
        We do a case analysis if 
        $\typingContextTwo, \typingContextTwo_s 
        \redto 
         \typingContextTwo', \typingContextTwo'_s$
        or
        $\typingContextTwo, \typingContextTwo_s 
         = 
         \typingContextTwo', \typingContextTwo'_s$.
        If the latter is the case, then 
        $\typingContextTwo = \typingContextTwo'$ (because no type binding for session $s$ can appear in $\typingContextTwo$). 
        If the former is the case, we do another case analysis 
        if $\typingContextTwo_s' = \typingContextTwo_s$ or not.
        This basically means if session $s$ is not concerned with this step or it is: the step cannot concern both because a step always happens within a single session. 
        If $\typingContextTwo'_s = \typingContextTwo_s$, then $\typingContextTwo \redto \typingContextTwo'$. 
        If $\typingContextTwo'_s \neq \typingContextTwo_s$, then $\typingContextTwo = \typingContextTwo'$. 
        This concludes this case. 
    \item $[\,]$: The proof is trivial because the induction hypothesis applies immediately. 
    \end{itemize}
\item $\procReductionCongr$: \\
    We have 
    \begin{mathpar}
    \inferrule*[right=\procReductionCongr]{
            P_1 \precongr P'_1
            \\
            P'_1 \redto P'_2
            \\
            P'_2 \precongr P_2
    }{
            P_1 \redto P_2
    }
    \end{mathpar}
    In addition, we know that 
    \ref{lm:subject-reduction-assumption-1} to 
    \ref{lm:subject-reduction-assumption-4} hold for 
    $P = P_1$ and $P' = P_2$. 
    We need to show that there is 
    $\typingContextTwo'$ with
    $\typingContextTwo \redto \typingContextTwo'$ or
    $\typingContextTwo = \typingContextTwo'$ such that
    $\typingContextOne \typingContextCat \typingContextTwo \types P_2$. 
    We want to apply the induction hypothesis for 
    $P = P'_1$ and $P' = P'_2$. 
    We need to show that all conditions 
    \ref{lm:subject-reduction-assumption-1} to 
    \ref{lm:subject-reduction-assumption-4} hold. 
    It is obvious that 
    \ref{lm:subject-reduction-assumption-1} and 
    \ref{lm:subject-reduction-assumption-3} do not change so they hold. 
    \ref{lm:subject-reduction-assumption-4} holds by assumption. 
    For \ref{lm:subject-reduction-assumption-2}, we observe that 
    $\typingContextOne \typingContextCat \typingContextTwo \types P_1$ holds by assumption.
    With the admissibility lemma
    (\cref{lm:str-congr-preserves-typability-procs}),
    we get 
    $\typingContextOne \typingContextCat \typingContextTwo \types P'_1$. 
    From conclusion of the induction hypothesis, we get 
    $\typingContextTwo'$ with
    $\typingContextTwo \redto \typingContextTwo'$ or
    $\typingContextTwo = \typingContextTwo'$ such that 
    $\typingContextOne \typingContextCat \typingContextTwo' \types P'_2$. 
    Applying the admissibility lemma 
    (\cref{lm:str-congr-preserves-typability-procs}), we obtain 
    $\typingContextOne \typingContextCat \typingContextTwo \types P_2$, which concludes this case.
\end{itemize}
This concludes the proof of subject reduction.
\end{proof}

\paragraph{On Type Safety.}
Usually, one proves type safety as a consequence of subject reduction. 
That is, no typable process can reduce to an error. 
However, we have not defined any error in our reduction rules, like many synchronous MST frameworks 
\cite{DBLP:journals/entcs/BejleriY09,DBLP:conf/icdcit/YoshidaG20,DBLP:conf/itp/EkiciKY25,NODBLPyet:journals/corr/abs-2511-22692}. 
\citet{DBLP:conf/lics/PetersY24}, closest to our work by supporting mixed choice, define two kind of errors: 
label mismatch and value error. 
The latter checks that if-then-else conditions are actually booleans. 
Our processes do not support if-then-else conditions. 
Accommodating this feature would not be difficult as it is orthogonal to message-passing concurrency. 
Label mismatch has previously been considered for processes with directed choice, \ie a process at any given branching can only send to one dedicated receiver or receive from one dedicated sender. 
In this setting, a label mismatch occurs when $\procA$ wants to send to $\procB$ but there is no common label. 
With mixed choice, such a situation is not necessarily problematic: a receiver might be waiting for a sender who might send to someone else in the meantime. 
\citet{DBLP:conf/lics/PetersY24} generalized this condition for mixed choice (for a single session which we call $s$) as follows: 
$P$~has a label error if 
\[
P \congr 
    \MixCh_{i \in I} 
        s[\procA][\procB_i] \, 
        \hat{\prefixPi}_i \seq P_i \parallel
    \MixCh_{j \in J} 
        s[\procC][\procD_j] \, 
        \hat{\prefixPi}_j \seq P_j \parallel P'
\]
such that there is $i \in I$ with 
$\procB_i = \procC$ and 
$\hat{\prefixPi}_i = \sendOp \labelAndMsg{l_i}{\val_i}$,
and for all $j \in J$ with 
$\hat{\prefixPi}_j = \recOp \labelAndVar{l_j}{x_j}$, 
it holds that $l_j \neq l_i$. 

Intuitively, a label error arises when $\procA$ wants to send a message to $\procB$, but $\procB$ is unable to receive it. 
We do not view all such scenarios as problematic. 
Consider the following simple protocol, in which $\procB$ first interacts with $\procC$ before interacting with $\procA$. 
Consider this simple protocol: 
$\msgFromTo{\procB}{\procC}{\labelAndMsg{l_1}{\val_1}} \cat \msgFromTo{\procA}{\procB}{\labelAndMsg{l_2}{\val_2}}$. 
The following serves as a deadlock-free implementation of the protocol, despite exhibiting a label error according to \cite{DBLP:conf/lics/PetersY24}:
\begin{align*}
P_\procA & = 
    s[\procA][\procB] \sendOp \labelAndMsg{l_2}{\val_2} \seq \zero
\\
P_\procB & = 
    s[\procB][\procC] \sendOp \labelAndMsg{l_1}{\val_1} \seq 
    s[\procB][\procA] \recOp \labelAndVar{l_2}{x_2} \seq \zero
\\
P_\procC & = 
    s[\procC][\procB] \recOp \labelAndVar{l_1}{x_1} \seq \zero
\end{align*} 
At the same time, the following reduction sequence is possible, according to their and our definition of reductions: 
\begin{align*}
    P_\procA 
        \parallel 
    P_\procB 
        \parallel 
    P_\procC 
\redto
    P_\procA 
        \parallel 
    s[\procB][\procA] \recOp \labelAndVar{l_2}{x_2} \seq \zero
        \parallel
    \zero
\redto
    \zero
        \parallel 
    \zero
        \parallel
    \zero
\congr \zero
\end{align*}
Hence, it is impossible to only consider two participants to anticipate errors. 
What would be problematic is if there is a circle of dependent processes who cannot move on. 
If this concerns all participants, we detect this as a deadlock but if some participants can still act, we do not detect this. 
In other words, progress ensures that some participant can take a step at any time but some participants might not be able to take a step anymore at some point. 
For any SCSM, one can check if the latter is the case, and a typed process will match the SCSM behavior.

\progressThm*
\begin{proof}
Without loss of generality, we assume the transition label for \ref{sf-cond-3} is 
 $\msgFromTo{\procE}{\procC}{\labelAndType{l}{L}}$. 
 We do inversion on \ref{sf-cond-2} and get
\begin{mathpar}
\inferrule*[right=\procTypingRestr ']{
\meta{\forall \procA \in \ProcsOf{\CSMabb{A}} \st}
    \forall c \hasType q' \in \typingContextTwo_\procA \st
    \EndState(q')
    \\
    \vec{q} \in \reach(\CSMabb{A})
    \\
\meta{\forall \procA \in \ProcsOf{\CSMabb{A}} \st}
        \typingContextOne
        \typingContextCat
        \typingContextTwo_\procA,
        s[\procA] \hasType \vec{q}_\procA
        \typesSFd
        Q_\procA
}{
    \typingContextOne
        \typingContextCat
        \set{\typingContextTwo_\procA}_{\procA \in \ProcsOf{\CSMabb{A}}}
\typesSFs
        (\restr s \hasType \CSMabb{A})\,
        (\Parallel_{\procA \in \ProcsOf{\CSMabb{A}}} Q_\procA)
}
\end{mathpar}
We instantiate the premise
\[
\meta{\forall \procA \in \ProcsOf{\CSMabb{A}} \st}
    \typingContextOne
    \typingContextCat
    \typingContextTwo_\procA,
    s[\procA] \hasType \vec{q}_\procA
    \typesSFd
    Q_\procA
\]
with $\procE$ and $\procC$ and observe how their typing derivation can be obtained. 
We know that $\vec{q}_\procE$ and $\vec{q}_\procC$ both have outgoing transitions so 
they are not final and hence neither $\EndState(\vec{q}_\procE)$ nor $\EndState(\vec{q}_\procC)$ holds. 
Hence, only \procTypingMixCh and \procTypingProcName can apply. 
We do a case analysis and show that applying \procTypingProcName will eventually lead to applying \procTypingMixCh as well.
First, we do an inversion:
\begin{mathpar}
\inferrule*[right=\procTypingProcName]{
    \typingContextOne(\pn{Q}) = \vec{L} }{
    \typingContextOne
        \typingContextCat
        \vec{c}  \hasType \vec{L}
\typesSFd
        \pn{Q}[\vec{c}]
}
\end{mathpar}
Assume that we have $\pn{Q}[\vec{x}] = P'$ in the process definitions $\typingContextOne$.
From~\ref{sf-cond-1},
it follows that
$
    \typingContextOne
    \typingContextCat
    \vec{x} \hasType \vec{L}
    \typesSFd
    P'
$.
By assumption, process definitions are guarded.
Thus, $P'$ ought to be typed with \procTypingMixCh itself.  
From here on, the reasoning of both cases are very similar again. 
In fact, there are two differences:
first, we would carry
$\vec{x} \hasType \vec{L}$ around, and
second, we would apply
\procReductionProcName
to prove that $P \redto P'$. 
For conciseness, we avoid the indirection through a process definition. 

Therefore, we consider \procTypingMixCh as typing rule for both 
\begin{align*}
Q_\procE & = 
    \MixCh_{i \in I_1} s[\procE][\procB_i] \sendOp \labelAndMsg{l_i}{c_i} \seq P_i
    + 
    \MixCh_{j \in J_1} s[\procE][\procB_j] \recOp \labelAndVar{l_j}{y_j} \seq P_j
\text{  and  }
\\
Q_\procC & = 
    \MixCh_{i \in I_2} s[\procC][\procD_i] \sendOp \labelAndMsg{l_i}{c_i} \seq P'_i
    + 
    \MixCh_{j \in J_2} s[\procC][\procD_j] \recOp \labelAndVar{l_j}{y_j} \seq P'_j
\enspace .
\end{align*}

We do inversion on their typing derivations and obtain

{ \small 
\begin{mathpar}
  \inferrule*[right=\procTypingMixCh]{
\delta(\vec{q}_\procE) =
      \set{(\msgFromTo{\procE}{\procB_i}{\labelAndType{l_i}{L_i}}, q_i) \mid i \in I_1}
      \dunion
      \set{(\msgFromTo{\procB_j}{\procE}{\labelAndType{l_j}{L_j}}, q_j) \mid j \in J_1}
      \\
\meta{\forall i \in I_1 \st}
      \typingContextOne \typingContextCat
          \hat{\typingContextTwo}_\procE,
           s[\procE] \hasType q_i,
          \set{c_j \hasType L_j}_{j \in I_1\setminus \set{i}}
          \types P_i 
      \\
      \meta{\forall j \in J_1 \st}
      \typingContextOne \typingContextCat
          \hat{\typingContextTwo}_\procE,
          s[\procE] \hasType q_j,
          \set{c_i \hasType L_i}_{i \in I_1}, 
          y_j \hasType L_j
          \types P_j \\
  }{
      (i) \;
      \typingContextOne
      \typingContextCat
      \hat{\typingContextTwo}_\procE,
      \set{c_i \hasType L_i}_{i \in I_1}, 
      s[\procE] \hasType \vec{q}_\procE
      \\
          \types
      \MixCh_{i \in I_1} s[\procE][\procB_i] \sendOp \labelAndMsg{l_i}{c_i} \seq P_i
      + 
      \MixCh_{j \in J_1} s[\procE][\procB_j] \recOp \labelAndVar{l_j}{y_j} \seq P_j
  }

  \inferrule*[right=\procTypingMixCh]{
\delta(\vec{q}_\procC) =
      \set{(\msgFromTo{\procC}{\procD_i}{\labelAndType{l_i}{L_i}}, q'_i) \mid i \in I_2}
      \dunion
      \set{(\msgFromTo{\procD_j}{\procC}{\labelAndType{l_j}{L_j}}, q'_j) \mid j \in J_2}
      \\
\meta{\forall i \in I_2 \st}
      \typingContextOne \typingContextCat
          \hat{\typingContextTwo}_\procC,
           s[\procC] \hasType q_i',
          \set{c_j \hasType L_j}_{j \in I_2\setminus \set{i}}
          \types P'_i 
      \\
      \meta{\forall j \in J_2 \st}
      \typingContextOne \typingContextCat
          \hat{\typingContextTwo}_\procC,
          s[\procC] \hasType q_j',
          \set{c_i \hasType L_i}_{i \in I_2}, 
          y_j \hasType L_j
          \types P'_j \\
  }{
      (ii) \;
      \typingContextOne
      \typingContextCat
      \hat{\typingContextTwo}_\procC,
      \set{c_i \hasType L_i}_{i \in I_2}, 
      s[\procC] \hasType \vec{q}_\procC
      \\
          \types
      \MixCh_{i \in I_2} s[\procC][\procD_i] \sendOp \labelAndMsg{l_i}{c_i} \seq P'_i
      + 
      \MixCh_{j \in J_2} s[\procC][\procD_j] \recOp \labelAndVar{l_j}{y_j} \seq P'_j
  }
\end{mathpar}
} 

where 
$
    \typingContextTwo_\procE = 
    \hat{\typingContextTwo}_\procE,
    \set{c_i \hasType L_i}_{i \in I_1} 
$
and 
$
    \typingContextTwo_\procC = 
    \hat{\typingContextTwo}_\procC,
    \set{c_i \hasType L_i}_{i \in I_2} 
$. 
By \ref{sf-cond-3} and our assumption that the transition label is 
$\msgFromTo{\procE}{\procC}{\labelAndType{l}{L}}$, 
there is $k \in (I_1 \inters J_2) \union (I_2 \inters J_1)$ with $l_k = l$.
Without loss of generality, assume that $k \in I_1 \inters J_2$, \ie $\procE$ sends and $\procC$ receives. 
Then, 
\[ 
P' = \Parallel_{\procE \in \ProcsOf{\CSMabb{A}} \setminus \set{\procE, \procC}} Q_\procA 
        \parallel P_k \parallel P'_k[c_k / y_k]
\enspace .
\]
We show that there is a typing derivation
\begin{mathpar}
\inferrule*[right=\procTypingRestr ']{
(a) \;
    \meta{\forall \procA \in \ProcsOf{\CSMabb{A}} \st}
    \forall c \hasType q' \in \typingContextTwo_\procA \st
    \EndState(q')
    \\
    (b) \;
    \pvec{q}' \in \reach(\CSMabb{A})
    \\
(c) \;
    \meta{\forall \procA \in \ProcsOf{\CSMabb{A} \setminus \set{\procE, \procC}} \st}
        \typingContextOne
        \typingContextCat
        \typingContextTwo_\procA,
        s[\procA] \hasType \pvec{q}'_\procA
        \typesSFd
        Q_\procA
    \\
    (d) \;
    \typingContextOne
    \typingContextCat
    \hat{\typingContextTwo}_\procE,
    s[\procE] \hasType \pvec{q}'_\procE,
    \set{c_j \hasType L_j}_{j \in I_1\setminus \set{k}} 
    \typesSFd
    P_k
    \\
    (e) \;
    \typingContextOne
    \typingContextCat
    \typingContextTwo_\procC,
    s[\procC] \hasType \pvec{q}'_\procC, 
    c_k \hasType L_k
    \typesSFd
    P'_k[c_k / y_k]
}{
    \typingContextOne
        \typingContextCat
        \set{\typingContextTwo_\procA}_{\procA \in \ProcsOf{\CSMabb{A}}}
\typesSFs
        (\restr s \hasType \CSMabb{A})\,
        \left(
            \Parallel_{\procA \in \ProcsOf{\CSMabb{A}} \setminus \set{\procE, \procC}} Q_\procA 
            \parallel P_k \parallel P'_k[c_k / y_k]
        \right)
}
\end{mathpar}
Premises $(a)$ and  $(b)$ are the same as the premises for the first inversion on \ref{sf-cond-2} while $(c)$ follows directly from the third premise of the same inversion. 
Premise $(d)$ can be obtained from the second premise of the inversion on $(i)$ when instantiated to $k$ because $\pvec{q}'_\procE = q_k$. 
Premise $(e)$ can be obtained similarly: 
from the third premise of the inversion on $(ii)$ when instantiated to $k$ because $\pvec{q}'_\procC = q_k'$, we obtain 
$
    \typingContextOne
    \typingContextCat
    \typingContextTwo_\procC,
    s[\procC] \hasType \pvec{q}'_\procC, 
    y_k \hasType L_k
    \typesSFd
    P'_k
$, 
to which one applies the substitution lemma (\cref{lm:substitution-lemma}) to obtain the goal. 
\end{proof}

\begin{restatable}[Deadlock freedom with sink-final FSMs]{lemma}{deadlockFreedom}
Let $\CSMabb{A}$ be a sink-final deadlock-free SCSM and let $P$ be a process. 
We assume that 
\begin{itemize}
    \item $\typesSFs \Defs \hasType \typingContextOne$,
    \item $\typingContextOne
            \typingContextCat
            \typingContextTwo
            \typesSFs
            (\restr s \hasType \CSMabb{A}) \, P$, and
    \item $(\restr s \hasType \CSMabb{A}) \, P \redto^* P_1$.
\end{itemize}
Then, it holds that $P_1 \precongr \zero$ or there is $P_2$ such that $P_1 \redto P_2$.
\end{restatable}
\textit{Proof Sketch.}
First, we claim that for all $k$ with $k \geq 0$ such that
$(\restr s \hasType \CSMabb{A}) \, P \redto^k P'$,
it holds that
$
\typingContextOne
    \typingContextCat
    \typingContextTwo
    \overset{\vec{q}}{\typesSFs}
    P'
$
for some
$\vec{q}$. 
This can be shown similarly to subject reduction (\cref{thm:subject-reduction}).

We do a case analysis if there is
$
    \pvec{q}'
$
with
$
    \vec{q}
        \redto
    \pvec{q}'
$.

If so, we know from
\cref{lm:session-fidelity}
that there is $P''$ with $P' \redto P''$, which concludes this case.

Suppose that there is no
$
    \pvec{q}'
$
with
$
    \vec{q}
        \redto
    \pvec{q}'
$.
By assumption $\CSMabb{A}$ is deadlock-free and, thus, all states in $\vec{q}$ are final states. 
Hence, $\EndState(\vec{q}_\procA)$ holds for every $\procA \in \Procs$. 
At the same time, by sink finality, one cannot take a transition from any of these states.
We do inversion on the typing derivation
$
\typingContextOne
    \typingContextCat
    \typingContextTwo
    \overset{\vec{q}}{\typesSFs}
    P'
$ and obtain the following when rewriting $\typingContextTwo$ and $P'$: 

\begin{mathpar}
\inferrule*[right=\procTypingRestr ']{
\meta{\forall \procA \in \ProcsOf{\CSMabb{A}} \st}
    \forall c \hasType q' \in \typingContextTwo_\procA \st
    \EndState(q')
    \\
    \vec{q} \in \reach(\CSMabb{A})
    \\
\meta{\forall \procA \in \ProcsOf{\CSMabb{A}} \st}
        \typingContextOne
        \typingContextCat
        \typingContextTwo_\procA,
        s[\procA] \hasType \vec{q}_\procA
        \typesSFd
        Q_\procA
}{
    \typingContextOne
        \typingContextCat
        \set{\typingContextTwo_\procA}_{\procA \in \ProcsOf{\CSMabb{A}}}
\typesSFs
        (\restr s \hasType \CSMabb{A})\,
        (\Parallel_{\procA \in \ProcsOf{\CSMabb{A}}} Q_\procA)
}
\end{mathpar}
We claim that $Q_\procA = \zero$ for every $\procA \in \Procs$.
This suffices to prove our goal because it entails that $P' \precongr \zero$. 
Let $\procA \in \Procs$. 
We do (multiple) inversions on 
$
        \typingContextOne
        \typingContextCat
        \typingContextTwo_\procA,
        s[\procA] \hasType \vec{q}_\procA
        \typesSFd
        Q_\procA
$ to understand which rules can appear. 
The process $Q_\procA$ is restriction-free so $\procTypingRestr '$ cannot appear. 
We know that for each state $q$ in the typing context 
$
        \typingContextTwo_\procA,
        s[\procA] \hasType \vec{q}_\procA
$, 
we have that $\EndState(q)$ holds. 
This is why $\procTypingMixCh$ cannot appear. 
The rule $\procTypingProcName$ can also not appear because process definitions are guarded and such a guard could not be typed with this typing context whose states all satisfy $\EndState(\hole)$. 
Hence, only $\procTypingParallel$, $\procTypingEnd$ and $\procTypingZero$ can appear, which ultimately gives a parallel composition of $\zero$, possibly nested. 
It follows that $P' \precongr \zero$.


\clearpage


\begin{thebibliography}{70}



\ifx \showCODEN    \undefined \def \showCODEN     #1{\unskip}     \fi
\ifx \showDOI      \undefined \def \showDOI       #1{#1}\fi
\ifx \showISBNx    \undefined \def \showISBNx     #1{\unskip}     \fi
\ifx \showISBNxiii \undefined \def \showISBNxiii  #1{\unskip}     \fi
\ifx \showISSN     \undefined \def \showISSN      #1{\unskip}     \fi
\ifx \showLCCN     \undefined \def \showLCCN      #1{\unskip}     \fi
\ifx \shownote     \undefined \def \shownote      #1{#1}          \fi
\ifx \showarticletitle \undefined \def \showarticletitle #1{#1}   \fi
\ifx \showURL      \undefined \def \showURL       {\relax}        \fi
\providecommand\bibfield[2]{#2}
\providecommand\bibinfo[2]{#2}
\providecommand\natexlab[1]{#1}
\providecommand\showeprint[2][]{arXiv:#2}

\bibitem[Aalbersberg and Welzl(1986)]{DBLP:journals/ita/AalbersbergW86}
\bibfield{author}{\bibinfo{person}{IJsbrand~Jan Aalbersberg} {and}
  \bibinfo{person}{Emo Welzl}.} \bibinfo{year}{1986}\natexlab{}.
\newblock \showarticletitle{Trace Languages Defined by Regular String
  Languages}.
\newblock \bibinfo{journal}{\emph{{RAIRO} Theor. Informatics Appl.}}
  \bibinfo{volume}{20}, \bibinfo{number}{2} (\bibinfo{year}{1986}),
  \bibinfo{pages}{103--119}.
\newblock
\urldef\tempurl \url{https://doi.org/10.1051/ITA/1986200201031}
\showDOI{\tempurl}


\bibitem[Alur et~al\mbox{.}(2003)]{DBLP:journals/tse/AlurEY03}
\bibfield{author}{\bibinfo{person}{Rajeev Alur}, \bibinfo{person}{Kousha
  Etessami}, {and} \bibinfo{person}{Mihalis Yannakakis}.}
  \bibinfo{year}{2003}\natexlab{}.
\newblock \showarticletitle{Inference of Message Sequence Charts}.
\newblock \bibinfo{journal}{\emph{{IEEE} Trans. Software Eng.}}
  \bibinfo{volume}{29}, \bibinfo{number}{7} (\bibinfo{year}{2003}),
  \bibinfo{pages}{623--633}.
\newblock
\urldef\tempurl \url{https://doi.org/10.1109/TSE.2003.1214326}
\showDOI{\tempurl}


\bibitem[Alur et~al\mbox{.}(2005)]{DBLP:journals/tcs/AlurEY05}
\bibfield{author}{\bibinfo{person}{Rajeev Alur}, \bibinfo{person}{Kousha
  Etessami}, {and} \bibinfo{person}{Mihalis Yannakakis}.}
  \bibinfo{year}{2005}\natexlab{}.
\newblock \showarticletitle{Realizability and verification of {MSC} graphs}.
\newblock \bibinfo{journal}{\emph{Theor. Comput. Sci.}} \bibinfo{volume}{331},
  \bibinfo{number}{1} (\bibinfo{year}{2005}), \bibinfo{pages}{97--114}.
\newblock
\urldef\tempurl \url{https://doi.org/10.1016/J.TCS.2004.09.034}
\showDOI{\tempurl}


\bibitem[Alur and Yannakakis(1999)]{DBLP:conf/concur/AlurY99}
\bibfield{author}{\bibinfo{person}{Rajeev Alur} {and} \bibinfo{person}{Mihalis
  Yannakakis}.} \bibinfo{year}{1999}\natexlab{}.
\newblock \showarticletitle{Model Checking of Message Sequence Charts}. In
  \bibinfo{booktitle}{\emph{{CONCUR} '99: Concurrency Theory, 10th
  International Conference, Eindhoven, The Netherlands, August 24-27, 1999,
  Proceedings}} \emph{(\bibinfo{series}{Lecture Notes in Computer Science},
  Vol.~\bibinfo{volume}{1664})}, \bibfield{editor}{\bibinfo{person}{Jos C.~M.
  Baeten} {and} \bibinfo{person}{Sjouke Mauw}} (Eds.).
  \bibinfo{publisher}{Springer}, \bibinfo{pages}{114--129}.
\newblock
\urldef\tempurl \url{https://doi.org/10.1007/3-540-48320-9\_10}
\showDOI{\tempurl}


\bibitem[Barbanera et~al\mbox{.}(2020)]{DBLP:conf/coordination/BarbaneraLT20}
\bibfield{author}{\bibinfo{person}{Franco Barbanera}, \bibinfo{person}{Ivan
  Lanese}, {and} \bibinfo{person}{Emilio Tuosto}.}
  \bibinfo{year}{2020}\natexlab{}.
\newblock \showarticletitle{Choreography Automata}. In
  \bibinfo{booktitle}{\emph{Coordination Models and Languages - 22nd {IFIP}
  {WG} 6.1 International Conference, {COORDINATION} 2020, Held as Part of the
  15th International Federated Conference on Distributed Computing Techniques,
  DisCoTec 2020, Valletta, Malta, June 15-19, 2020, Proceedings}}
  \emph{(\bibinfo{series}{Lecture Notes in Computer Science},
  Vol.~\bibinfo{volume}{12134})}, \bibfield{editor}{\bibinfo{person}{Simon
  Bliudze} {and} \bibinfo{person}{Laura Bocchi}} (Eds.).
  \bibinfo{publisher}{Springer}, \bibinfo{pages}{86--106}.
\newblock
\urldef\tempurl \url{https://doi.org/10.1007/978-3-030-50029-0\_6}
\showDOI{\tempurl}


\bibitem[Barbanera et~al\mbox{.}(2023)]{DBLP:journals/lmcs/BarbaneraLT23}
\bibfield{author}{\bibinfo{person}{Franco Barbanera}, \bibinfo{person}{Ivan
  Lanese}, {and} \bibinfo{person}{Emilio Tuosto}.}
  \bibinfo{year}{2023}\natexlab{}.
\newblock \showarticletitle{A Theory of Formal Choreographic Languages}.
\newblock \bibinfo{journal}{\emph{Log. Methods Comput. Sci.}}
  \bibinfo{volume}{19}, \bibinfo{number}{3} (\bibinfo{year}{2023}).
\newblock
\urldef\tempurl \url{https://doi.org/10.46298/LMCS-19(3:9)2023}
\showDOI{\tempurl}


\bibitem[Beier et~al\mbox{.}(2017)]{DBLP:journals/corr/abs-1708-06460}
\bibfield{author}{\bibinfo{person}{Simon Beier}, \bibinfo{person}{Markus
  Holzer}, {and} \bibinfo{person}{Martin Kutrib}.}
  \bibinfo{year}{2017}\natexlab{}.
\newblock \showarticletitle{On the Descriptional Complexity of Operations on
  Semilinear Sets}. In \bibinfo{booktitle}{\emph{Proceedings 15th International
  Conference on Automata and Formal Languages, {AFL} 2017, Debrecen, Hungary,
  September 4-6, 2017}} \emph{(\bibinfo{series}{{EPTCS}},
  Vol.~\bibinfo{volume}{252})},
  \bibfield{editor}{\bibinfo{person}{Erzs{\'{e}}bet Csuhaj{-}Varj{\'{u}}},
  \bibinfo{person}{P{\'{a}}l D{\"{o}}m{\"{o}}si}, {and}
  \bibinfo{person}{Gy{\"{o}}rgy Vaszil}} (Eds.). \bibinfo{pages}{41--55}.
\newblock
\urldef\tempurl \url{https://doi.org/10.4204/EPTCS.252.8}
\showDOI{\tempurl}


\bibitem[Bejleri and Yoshida(2008)]{DBLP:journals/entcs/BejleriY09}
\bibfield{author}{\bibinfo{person}{Andi Bejleri} {and} \bibinfo{person}{Nobuko
  Yoshida}.} \bibinfo{year}{2008}\natexlab{}.
\newblock \showarticletitle{Synchronous Multiparty Session Types}. In
  \bibinfo{booktitle}{\emph{Proceedings of the First Workshop on Programming
  Language Approaches to Concurrency and Communication-cEntric Software,
  PLACES@DisCoTec 2008, Oslo, Norway, June 7, 2008}}
  \emph{(\bibinfo{series}{Electronic Notes in Theoretical Computer Science},
  Vol.~\bibinfo{volume}{241})}, \bibfield{editor}{\bibinfo{person}{Vasco~T.
  Vasconcelos} {and} \bibinfo{person}{Nobuko Yoshida}} (Eds.).
  \bibinfo{publisher}{Elsevier}, \bibinfo{pages}{3--33}.
\newblock
\urldef\tempurl \url{https://doi.org/10.1016/J.ENTCS.2009.06.002}
\showDOI{\tempurl}


\bibitem[Bertoni et~al\mbox{.}(1982)]{DBLP:conf/icalp/BertoniMS82}
\bibfield{author}{\bibinfo{person}{Alberto Bertoni}, \bibinfo{person}{Giancarlo
  Mauri}, {and} \bibinfo{person}{Nicoletta Sabadini}.}
  \bibinfo{year}{1982}\natexlab{}.
\newblock \showarticletitle{Equivalence and Membership Problems for Regular
  Trace Languages}. In \bibinfo{booktitle}{\emph{Automata, Languages and
  Programming, 9th Colloquium, Aarhus, Denmark, July 12-16, 1982, Proceedings}}
  \emph{(\bibinfo{series}{Lecture Notes in Computer Science},
  Vol.~\bibinfo{volume}{140})}, \bibfield{editor}{\bibinfo{person}{Mogens
  Nielsen} {and} \bibinfo{person}{Erik~Meineche Schmidt}} (Eds.).
  \bibinfo{publisher}{Springer}, \bibinfo{pages}{61--71}.
\newblock
\urldef\tempurl \url{https://doi.org/10.1007/BFB0012757}
\showDOI{\tempurl}


\bibitem[Bocchi et~al\mbox{.}(2012)]{DBLP:conf/tgc/BocchiDY12}
\bibfield{author}{\bibinfo{person}{Laura Bocchi}, \bibinfo{person}{Romain
  Demangeon}, {and} \bibinfo{person}{Nobuko Yoshida}.}
  \bibinfo{year}{2012}\natexlab{}.
\newblock \showarticletitle{A Multiparty Multi-session Logic}. In
  \bibinfo{booktitle}{\emph{Trustworthy Global Computing - 7th International
  Symposium, {TGC} 2012, Newcastle upon Tyne, UK, September 7-8, 2012, Revised
  Selected Papers}} \emph{(\bibinfo{series}{Lecture Notes in Computer Science},
  Vol.~\bibinfo{volume}{8191})}, \bibfield{editor}{\bibinfo{person}{Catuscia
  Palamidessi} {and} \bibinfo{person}{Mark~Dermot Ryan}} (Eds.).
  \bibinfo{publisher}{Springer}, \bibinfo{pages}{97--111}.
\newblock
\urldef\tempurl \url{https://doi.org/10.1007/978-3-642-41157-1\_7}
\showDOI{\tempurl}


\bibitem[Bocchi et~al\mbox{.}(2010)]{DBLP:conf/concur/BocchiHTY10}
\bibfield{author}{\bibinfo{person}{Laura Bocchi}, \bibinfo{person}{Kohei
  Honda}, \bibinfo{person}{Emilio Tuosto}, {and} \bibinfo{person}{Nobuko
  Yoshida}.} \bibinfo{year}{2010}\natexlab{}.
\newblock \showarticletitle{A Theory of Design-by-Contract for Distributed
  Multiparty Interactions}. In \bibinfo{booktitle}{\emph{{CONCUR} 2010 -
  Concurrency Theory, 21th International Conference, {CONCUR} 2010, Paris,
  France, August 31-September 3, 2010. Proceedings}}
  \emph{(\bibinfo{series}{Lecture Notes in Computer Science},
  Vol.~\bibinfo{volume}{6269})}, \bibfield{editor}{\bibinfo{person}{Paul
  Gastin} {and} \bibinfo{person}{Fran{\c{c}}ois Laroussinie}} (Eds.).
  \bibinfo{publisher}{Springer}, \bibinfo{pages}{162--176}.
\newblock
\urldef\tempurl \url{https://doi.org/10.1007/978-3-642-15375-4\_12}
\showDOI{\tempurl}


\bibitem[Brand and Zafiropulo(1983)]{DBLP:journals/jacm/BrandZ83}
\bibfield{author}{\bibinfo{person}{Daniel Brand} {and} \bibinfo{person}{Pitro
  Zafiropulo}.} \bibinfo{year}{1983}\natexlab{}.
\newblock \showarticletitle{On Communicating Finite-State Machines}.
\newblock \bibinfo{journal}{\emph{J. {ACM}}} \bibinfo{volume}{30},
  \bibinfo{number}{2} (\bibinfo{year}{1983}), \bibinfo{pages}{323--342}.
\newblock
\urldef\tempurl \url{https://doi.org/10.1145/322374.322380}
\showDOI{\tempurl}


\bibitem[Bruschi et~al\mbox{.}(1994)]{DBLP:journals/iandc/BruschiPS94}
\bibfield{author}{\bibinfo{person}{Danilo Bruschi}, \bibinfo{person}{Giovanni
  Pighizzini}, {and} \bibinfo{person}{Nicoletta Sabadini}.}
  \bibinfo{year}{1994}\natexlab{}.
\newblock \showarticletitle{On the Existence of Minimum Asynchronous Automata
  and on the Equivalence Problem for Unambiguous Regular Trace Languages}.
\newblock \bibinfo{journal}{\emph{Inf. Comput.}} \bibinfo{volume}{108},
  \bibinfo{number}{2} (\bibinfo{year}{1994}), \bibinfo{pages}{262--285}.
\newblock
\urldef\tempurl \url{https://doi.org/10.1006/INCO.1994.1010}
\showDOI{\tempurl}


\bibitem[Castellani et~al\mbox{.}(1999)]{DBLP:conf/fsttcs/CastellaniMT99}
\bibfield{author}{\bibinfo{person}{Ilaria Castellani},
  \bibinfo{person}{Madhavan Mukund}, {and} \bibinfo{person}{P.~S.
  Thiagarajan}.} \bibinfo{year}{1999}\natexlab{}.
\newblock \showarticletitle{Synthesizing Distributed Transition Systems from
  Global Specification}. In \bibinfo{booktitle}{\emph{Foundations of Software
  Technology and Theoretical Computer Science, 19th Conference, Chennai, India,
  December 13-15, 1999, Proceedings}} \emph{(\bibinfo{series}{Lecture Notes in
  Computer Science}, Vol.~\bibinfo{volume}{1738})},
  \bibfield{editor}{\bibinfo{person}{C.~Pandu Rangan},
  \bibinfo{person}{Venkatesh Raman}, {and} \bibinfo{person}{Ramaswamy
  Ramanujam}} (Eds.). \bibinfo{publisher}{Springer}, \bibinfo{pages}{219--231}.
\newblock
\urldef\tempurl \url{https://doi.org/10.1007/3-540-46691-6\_17}
\showDOI{\tempurl}


\bibitem[Castro{-}Perez et~al\mbox{.}(2026)]{NODBLPyet:journals/corr/abs-2511-22692}
\bibfield{author}{\bibinfo{person}{David Castro{-}Perez},
  \bibinfo{person}{Francisco Ferreira}, {and} \bibinfo{person}{Sung{-}Shik
  Jongmans}.} \bibinfo{year}{2026}\natexlab{}.
\newblock \showarticletitle{A Synthetic Reconstruction of Multiparty Session
  Types}.
\newblock \bibinfo{journal}{\emph{Proc. {ACM} Program. Lang.}}
  \bibinfo{volume}{10}, \bibinfo{number}{{POPL}} (\bibinfo{year}{2026}),
  \bibinfo{pages}{50:1--50:29}.
\newblock
\urldef\tempurl \url{https://doi.org/10.1145/3776692}
\showDOI{\tempurl}


\bibitem[Coppo et~al\mbox{.}(2016)]{DBLP:journals/mscs/CoppoDYP16}
\bibfield{author}{\bibinfo{person}{Mario Coppo}, \bibinfo{person}{Mariangiola
  Dezani{-}Ciancaglini}, \bibinfo{person}{Nobuko Yoshida}, {and}
  \bibinfo{person}{Luca Padovani}.} \bibinfo{year}{2016}\natexlab{}.
\newblock \showarticletitle{Global progress for dynamically interleaved
  multiparty sessions}.
\newblock \bibinfo{journal}{\emph{Math. Struct. Comput. Sci.}}
  \bibinfo{volume}{26}, \bibinfo{number}{2} (\bibinfo{year}{2016}),
  \bibinfo{pages}{238--302}.
\newblock
\urldef\tempurl \url{https://doi.org/10.1017/S0960129514000188}
\showDOI{\tempurl}


\bibitem[Cruz{-}Filipe and Montesi(2020)]{DBLP:journals/tcs/Cruz-FilipeM20}
\bibfield{author}{\bibinfo{person}{Lu{\'{\i}}s Cruz{-}Filipe} {and}
  \bibinfo{person}{Fabrizio Montesi}.} \bibinfo{year}{2020}\natexlab{}.
\newblock \showarticletitle{A core model for choreographic programming}.
\newblock \bibinfo{journal}{\emph{Theor. Comput. Sci.}}  \bibinfo{volume}{802}
  (\bibinfo{year}{2020}), \bibinfo{pages}{38--66}.
\newblock
\urldef\tempurl \url{https://doi.org/10.1016/j.tcs.2019.07.005}
\showDOI{\tempurl}


\bibitem[Di~Giusto et~al\mbox{.}(2025)]{NODBLPyet:10.1145/3756907.3756918}
\bibfield{author}{\bibinfo{person}{Cinzia Di~Giusto}, \bibinfo{person}{Etienne
  Lozes}, {and} \bibinfo{person}{Pascal Urso}.}
  \bibinfo{year}{2025}\natexlab{}.
\newblock \showarticletitle{Realisability and Complementability of Multiparty
  Session Types}. In \bibinfo{booktitle}{\emph{Proceedings of the 27th
  International Symposium on Principles and Practice of Declarative
  Programming}} \emph{(\bibinfo{series}{PPDP '25})}.
  \bibinfo{publisher}{Association for Computing Machinery},
  \bibinfo{address}{New York, NY, USA}, Article \bibinfo{articleno}{11},
  \bibinfo{numpages}{12}~pages.
\newblock
\showISBNx{9798400720857}
\urldef\tempurl \url{https://doi.org/10.1145/3756907.3756918}
\showDOI{\tempurl}


\bibitem[Diekert and Rozenberg(1995)]{DBLP:books/ws/95/DR1995}
\bibfield{editor}{\bibinfo{person}{Volker Diekert} {and}
  \bibinfo{person}{Grzegorz Rozenberg}} (Eds.).
  \bibinfo{year}{1995}\natexlab{}.
\newblock \bibinfo{booktitle}{\emph{The Book of Traces}}.
\newblock \bibinfo{publisher}{World Scientific}.
\newblock
\showISBNx{978-981-02-2058-7}
\urldef\tempurl \url{https://doi.org/10.1142/2563}
\showDOI{\tempurl}


\bibitem[Ekici et~al\mbox{.}(2025)]{DBLP:conf/itp/EkiciKY25}
\bibfield{author}{\bibinfo{person}{Burak Ekici}, \bibinfo{person}{Tadayoshi
  Kamegai}, {and} \bibinfo{person}{Nobuko Yoshida}.}
  \bibinfo{year}{2025}\natexlab{}.
\newblock \showarticletitle{Formalising Subject Reduction and Progress for
  Multiparty Session Processes}. In \bibinfo{booktitle}{\emph{16th
  International Conference on Interactive Theorem Proving, {ITP} 2025,
  September 28 to October 1, 2025, Reykjavik, Iceland}}
  \emph{(\bibinfo{series}{LIPIcs}, Vol.~\bibinfo{volume}{352})},
  \bibfield{editor}{\bibinfo{person}{Yannick Forster} {and}
  \bibinfo{person}{Chantal Keller}} (Eds.). \bibinfo{publisher}{Schloss
  Dagstuhl - Leibniz-Zentrum f{\"{u}}r Informatik},
  \bibinfo{pages}{19:1--19:23}.
\newblock
\urldef\tempurl \url{https://doi.org/10.4230/LIPICS.ITP.2025.19}
\showDOI{\tempurl}


\bibitem[Fischer and Rosenberg(1968)]{DBLP:journals/jcss/FischerR68}
\bibfield{author}{\bibinfo{person}{Patrick~C. Fischer} {and}
  \bibinfo{person}{Arnold~L. Rosenberg}.} \bibinfo{year}{1968}\natexlab{}.
\newblock \showarticletitle{Multitape One-Way Nonwriting Automata}.
\newblock \bibinfo{journal}{\emph{J. Comput. Syst. Sci.}} \bibinfo{volume}{2},
  \bibinfo{number}{1} (\bibinfo{year}{1968}), \bibinfo{pages}{88--101}.
\newblock
\urldef\tempurl \url{https://doi.org/10.1016/S0022-0000(68)80006-6}
\showDOI{\tempurl}


\bibitem[Gazagnaire et~al\mbox{.}(2007)]{DBLP:conf/concur/GazagnaireGHTY07}
\bibfield{author}{\bibinfo{person}{Thomas Gazagnaire}, \bibinfo{person}{Blaise
  Genest}, \bibinfo{person}{Lo{\"{\i}}c H{\'{e}}lou{\"{e}}t},
  \bibinfo{person}{P.~S. Thiagarajan}, {and} \bibinfo{person}{Shaofa Yang}.}
  \bibinfo{year}{2007}\natexlab{}.
\newblock \showarticletitle{Causal Message Sequence Charts}. In
  \bibinfo{booktitle}{\emph{{CONCUR} 2007 - Concurrency Theory, 18th
  International Conference, {CONCUR} 2007, Lisbon, Portugal, September 3-8,
  2007, Proceedings}} \emph{(\bibinfo{series}{Lecture Notes in Computer
  Science}, Vol.~\bibinfo{volume}{4703})},
  \bibfield{editor}{\bibinfo{person}{Lu{\'{\i}}s Caires} {and}
  \bibinfo{person}{Vasco~Thudichum Vasconcelos}} (Eds.).
  \bibinfo{publisher}{Springer}, \bibinfo{pages}{166--180}.
\newblock
\urldef\tempurl \url{https://doi.org/10.1007/978-3-540-74407-8\_12}
\showDOI{\tempurl}


\bibitem[Genest and Muscholl(2005)]{DBLP:conf/acsd/GenestM05}
\bibfield{author}{\bibinfo{person}{Blaise Genest} {and} \bibinfo{person}{Anca
  Muscholl}.} \bibinfo{year}{2005}\natexlab{}.
\newblock \showarticletitle{Message Sequence Charts: {A} Survey}. In
  \bibinfo{booktitle}{\emph{Fifth International Conference on Application of
  Concurrency to System Design {(ACSD} 2005), 6-9 June 2005, St. Malo,
  France}}. \bibinfo{publisher}{{IEEE} Computer Society},
  \bibinfo{pages}{2--4}.
\newblock
\urldef\tempurl \url{https://doi.org/10.1109/ACSD.2005.25}
\showDOI{\tempurl}


\bibitem[Genest et~al\mbox{.}(2004)]{DBLP:conf/dlt/GenestMK04}
\bibfield{author}{\bibinfo{person}{Blaise Genest}, \bibinfo{person}{Anca
  Muscholl}, {and} \bibinfo{person}{Dietrich Kuske}.}
  \bibinfo{year}{2004}\natexlab{}.
\newblock \showarticletitle{A Kleene Theorem for a Class of Communicating
  Automata with Effective Algorithms}. In
  \bibinfo{booktitle}{\emph{Developments in Language Theory, 8th International
  Conference, {DLT} 2004, Auckland, New Zealand, December 13-17, 2004,
  Proceedings}} \emph{(\bibinfo{series}{Lecture Notes in Computer Science},
  Vol.~\bibinfo{volume}{3340})}, \bibfield{editor}{\bibinfo{person}{Cristian
  Calude}, \bibinfo{person}{Elena Calude}, {and} \bibinfo{person}{Michael~J.
  Dinneen}} (Eds.). \bibinfo{publisher}{Springer}, \bibinfo{pages}{30--48}.
\newblock
\urldef\tempurl \url{https://doi.org/10.1007/978-3-540-30550-7\_4}
\showDOI{\tempurl}


\bibitem[Genest et~al\mbox{.}(2003)]{DBLP:conf/ac/GenestMP03}
\bibfield{author}{\bibinfo{person}{Blaise Genest}, \bibinfo{person}{Anca
  Muscholl}, {and} \bibinfo{person}{Doron~A. Peled}.}
  \bibinfo{year}{2003}\natexlab{}.
\newblock \showarticletitle{Message Sequence Charts}. In
  \bibinfo{booktitle}{\emph{Lectures on Concurrency and Petri Nets, Advances in
  Petri Nets [This tutorial volume originates from the 4th Advanced Course on
  Petri Nets, {ACPN} 2003, held in Eichst{\"{a}}tt, Germany in September 2003.
  In addition to lectures given at {ACPN} 2003, additional chapters have been
  commissioned]}} \emph{(\bibinfo{series}{Lecture Notes in Computer Science},
  Vol.~\bibinfo{volume}{3098})}, \bibfield{editor}{\bibinfo{person}{J{\"{o}}rg
  Desel}, \bibinfo{person}{Wolfgang Reisig}, {and} \bibinfo{person}{Grzegorz
  Rozenberg}} (Eds.). \bibinfo{publisher}{Springer}, \bibinfo{pages}{537--558}.
\newblock
\urldef\tempurl \url{https://doi.org/10.1007/978-3-540-27755-2\_15}
\showDOI{\tempurl}


\bibitem[Genest et~al\mbox{.}(2006)]{DBLP:journals/jcss/GenestMSZ06}
\bibfield{author}{\bibinfo{person}{Blaise Genest}, \bibinfo{person}{Anca
  Muscholl}, \bibinfo{person}{Helmut Seidl}, {and} \bibinfo{person}{Marc
  Zeitoun}.} \bibinfo{year}{2006}\natexlab{}.
\newblock \showarticletitle{Infinite-state high-level MSCs: Model-checking and
  realizability}.
\newblock \bibinfo{journal}{\emph{J. Comput. Syst. Sci.}} \bibinfo{volume}{72},
  \bibinfo{number}{4} (\bibinfo{year}{2006}), \bibinfo{pages}{617--647}.
\newblock
\urldef\tempurl \url{https://doi.org/10.1016/J.JCSS.2005.09.007}
\showDOI{\tempurl}


\bibitem[Giallorenzo et~al\mbox{.}(2021)]{DBLP:conf/ecoop/GiallorenzoMPRS21}
\bibfield{author}{\bibinfo{person}{Saverio Giallorenzo},
  \bibinfo{person}{Fabrizio Montesi}, \bibinfo{person}{Marco Peressotti},
  \bibinfo{person}{David Richter}, \bibinfo{person}{Guido Salvaneschi}, {and}
  \bibinfo{person}{Pascal Weisenburger}.} \bibinfo{year}{2021}\natexlab{}.
\newblock \showarticletitle{Multiparty Languages: The Choreographic and
  Multitier Cases (Pearl)}. In \bibinfo{booktitle}{\emph{35th European
  Conference on Object-Oriented Programming, {ECOOP} 2021, July 11-17, 2021,
  Aarhus, Denmark (Virtual Conference)}} \emph{(\bibinfo{series}{LIPIcs},
  Vol.~\bibinfo{volume}{194})}, \bibfield{editor}{\bibinfo{person}{Anders
  M{\o}ller} {and} \bibinfo{person}{Manu Sridharan}} (Eds.).
  \bibinfo{publisher}{Schloss Dagstuhl - Leibniz-Zentrum f{\"{u}}r Informatik},
  \bibinfo{pages}{22:1--22:27}.
\newblock
\urldef\tempurl \url{https://doi.org/10.4230/LIPIcs.ECOOP.2021.22}
\showDOI{\tempurl}


\bibitem[Giusto et~al\mbox{.}(2025)]{DBLP:conf/ppdp/GiustoLU25}
\bibfield{author}{\bibinfo{person}{Cinzia~Di Giusto},
  \bibinfo{person}{{\'{E}}tienne Lozes}, {and} \bibinfo{person}{Pascal Urso}.}
  \bibinfo{year}{2025}\natexlab{}.
\newblock \showarticletitle{Realisability and Complementability of Multiparty
  Session Types}. In \bibinfo{booktitle}{\emph{Proceedings of the 27th
  International Symposium on Principles and Practice of Declarative
  Programming, {PPDP} 2025, Rende, Italy, September 10-11, 2025}},
  \bibfield{editor}{\bibinfo{person}{Malgorzata Biernacka},
  \bibinfo{person}{Carlos Olarte}, \bibinfo{person}{Francesco Ricca}, {and}
  \bibinfo{person}{James Cheney}} (Eds.). \bibinfo{publisher}{{ACM}},
  \bibinfo{pages}{11:1--11:12}.
\newblock
\urldef\tempurl \url{https://doi.org/10.1145/3756907.3756918}
\showDOI{\tempurl}


\bibitem[Hinrichsen et~al\mbox{.}(2020)]{DBLP:journals/pacmpl/HinrichsenBK20}
\bibfield{author}{\bibinfo{person}{Jonas~Kastberg Hinrichsen},
  \bibinfo{person}{Jesper Bengtson}, {and} \bibinfo{person}{Robbert Krebbers}.}
  \bibinfo{year}{2020}\natexlab{}.
\newblock \showarticletitle{Actris: session-type based reasoning in separation
  logic}.
\newblock \bibinfo{journal}{\emph{Proc. {ACM} Program. Lang.}}
  \bibinfo{volume}{4}, \bibinfo{number}{{POPL}} (\bibinfo{year}{2020}),
  \bibinfo{pages}{6:1--6:30}.
\newblock
\urldef\tempurl \url{https://doi.org/10.1145/3371074}
\showDOI{\tempurl}


\bibitem[Hinrichsen et~al\mbox{.}(2022)]{DBLP:journals/lmcs/HinrichsenBK22}
\bibfield{author}{\bibinfo{person}{Jonas~Kastberg Hinrichsen},
  \bibinfo{person}{Jesper Bengtson}, {and} \bibinfo{person}{Robbert Krebbers}.}
  \bibinfo{year}{2022}\natexlab{}.
\newblock \showarticletitle{Actris 2.0: Asynchronous Session-Type Based
  Reasoning in Separation Logic}.
\newblock \bibinfo{journal}{\emph{Log. Methods Comput. Sci.}}
  \bibinfo{volume}{18}, \bibinfo{number}{2} (\bibinfo{year}{2022}).
\newblock
\urldef\tempurl \url{https://doi.org/10.46298/LMCS-18(2:16)2022}
\showDOI{\tempurl}


\bibitem[Hirsch and Garg(2021)]{DBLP:journals/corr/abs-2111-03484}
\bibfield{author}{\bibinfo{person}{Andrew~K. Hirsch} {and}
  \bibinfo{person}{Deepak Garg}.} \bibinfo{year}{2021}\natexlab{}.
\newblock \showarticletitle{Pirouette: Higher-Order Typed Functional
  Choreographies}.
\newblock \bibinfo{journal}{\emph{CoRR}}  \bibinfo{volume}{abs/2111.03484}
  (\bibinfo{year}{2021}).
\newblock
\showeprint[arXiv]{2111.03484}
\urldef\tempurl \url{https://arxiv.org/abs/2111.03484}
\showURL{\tempurl}


\bibitem[Hirsch and Garg(2022)]{DBLP:journals/pacmpl/HirschG22}
\bibfield{author}{\bibinfo{person}{Andrew~K. Hirsch} {and}
  \bibinfo{person}{Deepak Garg}.} \bibinfo{year}{2022}\natexlab{}.
\newblock \showarticletitle{Pirouette: higher-order typed functional
  choreographies}.
\newblock \bibinfo{journal}{\emph{Proc. {ACM} Program. Lang.}}
  \bibinfo{volume}{6}, \bibinfo{number}{{POPL}} (\bibinfo{year}{2022}),
  \bibinfo{pages}{1--27}.
\newblock
\urldef\tempurl \url{https://doi.org/10.1145/3498684}
\showDOI{\tempurl}


\bibitem[Honda et~al\mbox{.}(2008)]{DBLP:conf/popl/HondaYC08}
\bibfield{author}{\bibinfo{person}{Kohei Honda}, \bibinfo{person}{Nobuko
  Yoshida}, {and} \bibinfo{person}{Marco Carbone}.}
  \bibinfo{year}{2008}\natexlab{}.
\newblock \showarticletitle{Multiparty asynchronous session types}. In
  \bibinfo{booktitle}{\emph{Proceedings of the 35th {ACM} {SIGPLAN-SIGACT}
  Symposium on Principles of Programming Languages, {POPL} 2008, San Francisco,
  California, USA, January 7-12, 2008}},
  \bibfield{editor}{\bibinfo{person}{George~C. Necula} {and}
  \bibinfo{person}{Philip Wadler}} (Eds.). \bibinfo{publisher}{{ACM}},
  \bibinfo{pages}{273--284}.
\newblock
\urldef\tempurl \url{https://doi.org/10.1145/1328438.1328472}
\showDOI{\tempurl}


\bibitem[H{\"{u}}ttel et~al\mbox{.}(2016)]{DBLP:journals/csur/HuttelLVCCDMPRT16}
\bibfield{author}{\bibinfo{person}{Hans H{\"{u}}ttel}, \bibinfo{person}{Ivan
  Lanese}, \bibinfo{person}{Vasco~T. Vasconcelos}, \bibinfo{person}{Lu{\'{\i}}s
  Caires}, \bibinfo{person}{Marco Carbone}, \bibinfo{person}{Pierre{-}Malo
  Deni{\'{e}}lou}, \bibinfo{person}{Dimitris Mostrous}, \bibinfo{person}{Luca
  Padovani}, \bibinfo{person}{Ant{\'{o}}nio Ravara}, \bibinfo{person}{Emilio
  Tuosto}, \bibinfo{person}{Hugo~Torres Vieira}, {and}
  \bibinfo{person}{Gianluigi Zavattaro}.} \bibinfo{year}{2016}\natexlab{}.
\newblock \showarticletitle{Foundations of Session Types and Behavioural
  Contracts}.
\newblock \bibinfo{journal}{\emph{{ACM} Comput. Surv.}} \bibinfo{volume}{49},
  \bibinfo{number}{1} (\bibinfo{year}{2016}), \bibinfo{pages}{3:1--3:36}.
\newblock
\urldef\tempurl \url{https://doi.org/10.1145/2873052}
\showDOI{\tempurl}


\bibitem[Ibarra(1978)]{DBLP:journals/jacm/Ibarra78}
\bibfield{author}{\bibinfo{person}{Oscar~H. Ibarra}.}
  \bibinfo{year}{1978}\natexlab{}.
\newblock \showarticletitle{Reversal-Bounded Multicounter Machines and Their
  Decision Problems}.
\newblock \bibinfo{journal}{\emph{J. {ACM}}} \bibinfo{volume}{25},
  \bibinfo{number}{1} (\bibinfo{year}{1978}), \bibinfo{pages}{116--133}.
\newblock
\urldef\tempurl \url{https://doi.org/10.1145/322047.322058}
\showDOI{\tempurl}


\bibitem[Jacobs et~al\mbox{.}(2023)]{DBLP:journals/pacmpl/JacobsHK23}
\bibfield{author}{\bibinfo{person}{Jules Jacobs},
  \bibinfo{person}{Jonas~Kastberg Hinrichsen}, {and} \bibinfo{person}{Robbert
  Krebbers}.} \bibinfo{year}{2023}\natexlab{}.
\newblock \showarticletitle{Dependent Session Protocols in Separation Logic
  from First Principles (Functional Pearl)}.
\newblock \bibinfo{journal}{\emph{Proc. {ACM} Program. Lang.}}
  \bibinfo{volume}{7}, \bibinfo{number}{{ICFP}} (\bibinfo{year}{2023}),
  \bibinfo{pages}{768--795}.
\newblock
\urldef\tempurl \url{https://doi.org/10.1145/3607856}
\showDOI{\tempurl}


\bibitem[Kobayashi(2006)]{DBLP:conf/concur/Kobayashi06}
\bibfield{author}{\bibinfo{person}{Naoki Kobayashi}.}
  \bibinfo{year}{2006}\natexlab{}.
\newblock \showarticletitle{A New Type System for Deadlock-Free Processes}. In
  \bibinfo{booktitle}{\emph{{CONCUR}}} \emph{(\bibinfo{series}{Lecture Notes in
  Computer Science}, Vol.~\bibinfo{volume}{4137})}.
  \bibinfo{publisher}{Springer}, \bibinfo{pages}{233--247}.
\newblock
\urldef\tempurl \url{https://doi.org/10.1007/11817949_16}
\showDOI{\tempurl}


\bibitem[Kopczynski and To(2010)]{DBLP:conf/lics/KopczynskiT10}
\bibfield{author}{\bibinfo{person}{Eryk Kopczynski} {and}
  \bibinfo{person}{Anthony~Widjaja To}.} \bibinfo{year}{2010}\natexlab{}.
\newblock \showarticletitle{Parikh Images of Grammars: Complexity and
  Applications}. In \bibinfo{booktitle}{\emph{Proceedings of the 25th Annual
  {IEEE} Symposium on Logic in Computer Science, {LICS} 2010, 11-14 July 2010,
  Edinburgh, United Kingdom}}. \bibinfo{publisher}{{IEEE} Computer Society},
  \bibinfo{pages}{80--89}.
\newblock
\urldef\tempurl \url{https://doi.org/10.1109/LICS.2010.21}
\showDOI{\tempurl}


\bibitem[Li et~al\mbox{.}(2024)]{DBLP:conf/esop/LiSW24}
\bibfield{author}{\bibinfo{person}{Elaine Li}, \bibinfo{person}{Felix Stutz},
  {and} \bibinfo{person}{Thomas Wies}.} \bibinfo{year}{2024}\natexlab{}.
\newblock \showarticletitle{Deciding Subtyping for Asynchronous Multiparty
  Sessions}. In \bibinfo{booktitle}{\emph{Programming Languages and Systems -
  33rd European Symposium on Programming, {ESOP} 2024, Held as Part of the
  European Joint Conferences on Theory and Practice of Software, {ETAPS} 2024,
  Luxembourg City, Luxembourg, April 6-11, 2024, Proceedings, Part {I}}}
  \emph{(\bibinfo{series}{Lecture Notes in Computer Science},
  Vol.~\bibinfo{volume}{14576})}, \bibfield{editor}{\bibinfo{person}{Stephanie
  Weirich}} (Ed.). \bibinfo{publisher}{Springer}, \bibinfo{pages}{176--205}.
\newblock
\urldef\tempurl \url{https://doi.org/10.1007/978-3-031-57262-3\_8}
\showDOI{\tempurl}


\bibitem[Li et~al\mbox{.}(2023)]{DBLP:conf/cav/LiSWZ23}
\bibfield{author}{\bibinfo{person}{Elaine Li}, \bibinfo{person}{Felix Stutz},
  \bibinfo{person}{Thomas Wies}, {and} \bibinfo{person}{Damien Zufferey}.}
  \bibinfo{year}{2023}\natexlab{}.
\newblock \showarticletitle{Complete Multiparty Session Type Projection with
  Automata}. In \bibinfo{booktitle}{\emph{Computer Aided Verification - 35th
  International Conference, {CAV} 2023, Paris, France, July 17-22, 2023,
  Proceedings, Part {III}}} \emph{(\bibinfo{series}{Lecture Notes in Computer
  Science}, Vol.~\bibinfo{volume}{13966})},
  \bibfield{editor}{\bibinfo{person}{Constantin Enea} {and}
  \bibinfo{person}{Akash Lal}} (Eds.). \bibinfo{publisher}{Springer},
  \bibinfo{pages}{350--373}.
\newblock
\urldef\tempurl \url{https://doi.org/10.1007/978-3-031-37709-9\_17}
\showDOI{\tempurl}


\bibitem[Li et~al\mbox{.}(2025a)]{DBLP:journals/pacmpl/LiSWZ25}
\bibfield{author}{\bibinfo{person}{Elaine Li}, \bibinfo{person}{Felix Stutz},
  \bibinfo{person}{Thomas Wies}, {and} \bibinfo{person}{Damien Zufferey}.}
  \bibinfo{year}{2025}\natexlab{a}.
\newblock \showarticletitle{Characterizing Implementability of Global Protocols
  with Infinite States and Data}.
\newblock \bibinfo{journal}{\emph{Proc. {ACM} Program. Lang.}}
  \bibinfo{volume}{9}, \bibinfo{number}{{OOPSLA1}} (\bibinfo{year}{2025}),
  \bibinfo{pages}{1434--1463}.
\newblock
\urldef\tempurl \url{https://doi.org/10.1145/3720493}
\showDOI{\tempurl}


\bibitem[Li et~al\mbox{.}(2025b)]{extendedversion}
\bibfield{author}{\bibinfo{person}{Elaine Li}, \bibinfo{person}{Felix Stutz},
  \bibinfo{person}{Thomas Wies}, {and} \bibinfo{person}{Damien Zufferey}.}
  \bibinfo{year}{2025}\natexlab{b}.
\newblock \bibinfo{title}{Characterizing Implementability of Global Protocols
  with Infinite States and Data}.
\newblock
\newblock
\showeprint[arxiv]{2411.05722}~[cs.PL]
\urldef\tempurl \url{https://arxiv.org/abs/2411.05722}
\showURL{\tempurl}


\bibitem[Li et~al\mbox{.}(2025c)]{DBLP:conf/cav/LiSWZ25}
\bibfield{author}{\bibinfo{person}{Elaine Li}, \bibinfo{person}{Felix Stutz},
  \bibinfo{person}{Thomas Wies}, {and} \bibinfo{person}{Damien Zufferey}.}
  \bibinfo{year}{2025}\natexlab{c}.
\newblock \showarticletitle{Sprout: {A} Verifier for Symbolic Multiparty
  Protocols}. In \bibinfo{booktitle}{\emph{Computer Aided Verification - 37th
  International Conference, {CAV} 2025, Zagreb, Croatia, July 23-25, 2025,
  Proceedings, Part {III}}} \emph{(\bibinfo{series}{Lecture Notes in Computer
  Science}, Vol.~\bibinfo{volume}{15933})},
  \bibfield{editor}{\bibinfo{person}{Ruzica Piskac} {and}
  \bibinfo{person}{Zvonimir Rakamaric}} (Eds.). \bibinfo{publisher}{Springer},
  \bibinfo{pages}{304--317}.
\newblock
\urldef\tempurl \url{https://doi.org/10.1007/978-3-031-98682-6\_16}
\showDOI{\tempurl}


\bibitem[Li and Wies(2025)]{DBLP:conf/itp/LiW25}
\bibfield{author}{\bibinfo{person}{Elaine Li} {and} \bibinfo{person}{Thomas
  Wies}.} \bibinfo{year}{2025}\natexlab{}.
\newblock \showarticletitle{Certified Implementability of Global Multiparty
  Protocols}. In \bibinfo{booktitle}{\emph{16th International Conference on
  Interactive Theorem Proving, {ITP} 2025, Reykjavik, Iceland, September 28 -
  October 1, 2025}} \emph{(\bibinfo{series}{LIPIcs},
  Vol.~\bibinfo{volume}{352})}, \bibfield{editor}{\bibinfo{person}{Yannick
  Forster} {and} \bibinfo{person}{Chantal Keller}} (Eds.).
  \bibinfo{publisher}{Schloss Dagstuhl - Leibniz-Zentrum f{\"{u}}r Informatik},
  \bibinfo{pages}{15:1--15:20}.
\newblock
\urldef\tempurl \url{https://doi.org/10.4230/LIPICS.ITP.2025.15}
\showDOI{\tempurl}


\bibitem[Lohrey(2003)]{DBLP:journals/tcs/Lohrey03}
\bibfield{author}{\bibinfo{person}{Markus Lohrey}.}
  \bibinfo{year}{2003}\natexlab{}.
\newblock \showarticletitle{Realizability of high-level message sequence
  charts: closing the gaps}.
\newblock \bibinfo{journal}{\emph{Theor. Comput. Sci.}} \bibinfo{volume}{309},
  \bibinfo{number}{1-3} (\bibinfo{year}{2003}), \bibinfo{pages}{529--554}.
\newblock
\urldef\tempurl \url{https://doi.org/10.1016/J.TCS.2003.08.002}
\showDOI{\tempurl}


\bibitem[Majumdar et~al\mbox{.}(2021a)]{DBLP:conf/concur/MajumdarMSZ21}
\bibfield{author}{\bibinfo{person}{Rupak Majumdar}, \bibinfo{person}{Madhavan
  Mukund}, \bibinfo{person}{Felix Stutz}, {and} \bibinfo{person}{Damien
  Zufferey}.} \bibinfo{year}{2021}\natexlab{a}.
\newblock \showarticletitle{Generalising Projection in Asynchronous Multiparty
  Session Types}. In \bibinfo{booktitle}{\emph{32nd International Conference on
  Concurrency Theory, {CONCUR} 2021, August 24-27, 2021, Virtual Conference}}
  \emph{(\bibinfo{series}{LIPIcs}, Vol.~\bibinfo{volume}{203})},
  \bibfield{editor}{\bibinfo{person}{Serge Haddad} {and}
  \bibinfo{person}{Daniele Varacca}} (Eds.). \bibinfo{publisher}{Schloss
  Dagstuhl - Leibniz-Zentrum f{\"{u}}r Informatik},
  \bibinfo{pages}{35:1--35:24}.
\newblock
\urldef\tempurl \url{https://doi.org/10.4230/LIPICS.CONCUR.2021.35}
\showDOI{\tempurl}


\bibitem[Majumdar et~al\mbox{.}(2021b)]{DBLP:journals/corr/abs-2107-03984}
\bibfield{author}{\bibinfo{person}{Rupak Majumdar}, \bibinfo{person}{Madhavan
  Mukund}, \bibinfo{person}{Felix Stutz}, {and} \bibinfo{person}{Damien
  Zufferey}.} \bibinfo{year}{2021}\natexlab{b}.
\newblock \showarticletitle{Generalising Projection in Asynchronous Multiparty
  Session Types}.
\newblock \bibinfo{journal}{\emph{CoRR}}  \bibinfo{volume}{abs/2107.03984}
  (\bibinfo{year}{2021}).
\newblock
\showeprint[arXiv]{2107.03984}
\urldef\tempurl \url{https://arxiv.org/abs/2107.03984}
\showURL{\tempurl}


\bibitem[Mauw and Reniers(1997)]{DBLP:conf/sdl/MauwR97}
\bibfield{author}{\bibinfo{person}{Sjouke Mauw} {and}
  \bibinfo{person}{Michel~A. Reniers}.} \bibinfo{year}{1997}\natexlab{}.
\newblock \showarticletitle{High-level message sequence charts}. In
  \bibinfo{booktitle}{\emph{{SDL} '97 Time for Testing, SDL, {MSC} and Trends -
  8th International {SDL} Forum, Evry, France, 23-29 September 1997,
  Proceedings}}, \bibfield{editor}{\bibinfo{person}{Ana~R. Cavalli} {and}
  \bibinfo{person}{Amardeo Sarma}} (Eds.). \bibinfo{publisher}{Elsevier},
  \bibinfo{pages}{291--306}.
\newblock


\bibitem[Morin(2002)]{DBLP:conf/stacs/Morin02}
\bibfield{author}{\bibinfo{person}{R{\'{e}}mi Morin}.}
  \bibinfo{year}{2002}\natexlab{}.
\newblock \showarticletitle{Recognizable Sets of Message Sequence Charts}. In
  \bibinfo{booktitle}{\emph{{STACS} 2002, 19th Annual Symposium on Theoretical
  Aspects of Computer Science, Antibes - Juan les Pins, France, March 14-16,
  2002, Proceedings}} \emph{(\bibinfo{series}{Lecture Notes in Computer
  Science}, Vol.~\bibinfo{volume}{2285})},
  \bibfield{editor}{\bibinfo{person}{Helmut Alt} {and} \bibinfo{person}{Afonso
  Ferreira}} (Eds.). \bibinfo{publisher}{Springer}, \bibinfo{pages}{523--534}.
\newblock
\urldef\tempurl \url{https://doi.org/10.1007/3-540-45841-7\_43}
\showDOI{\tempurl}


\bibitem[Muscholl and Peled(1999)]{DBLP:conf/mfcs/MuschollP99}
\bibfield{author}{\bibinfo{person}{Anca Muscholl} {and}
  \bibinfo{person}{Doron~A. Peled}.} \bibinfo{year}{1999}\natexlab{}.
\newblock \showarticletitle{Message Sequence Graphs and Decision Problems on
  Mazurkiewicz Traces}. In \bibinfo{booktitle}{\emph{Mathematical Foundations
  of Computer Science 1999, 24th International Symposium, MFCS'99, Szklarska
  Poreba, Poland, September 6-10, 1999, Proceedings}}
  \emph{(\bibinfo{series}{Lecture Notes in Computer Science},
  Vol.~\bibinfo{volume}{1672})}, \bibfield{editor}{\bibinfo{person}{Miroslaw
  Kutylowski}, \bibinfo{person}{Leszek Pacholski}, {and}
  \bibinfo{person}{Tomasz Wierzbicki}} (Eds.). \bibinfo{publisher}{Springer},
  \bibinfo{pages}{81--91}.
\newblock
\urldef\tempurl \url{https://doi.org/10.1007/3-540-48340-3\_8}
\showDOI{\tempurl}


\bibitem[Peled et~al\mbox{.}(1998)]{DBLP:journals/tcs/PeledWW98}
\bibfield{author}{\bibinfo{person}{Doron~A. Peled}, \bibinfo{person}{Thomas
  Wilke}, {and} \bibinfo{person}{Pierre Wolper}.}
  \bibinfo{year}{1998}\natexlab{}.
\newblock \showarticletitle{An Algorithmic Approach for Checking Closure
  Properties of Temporal Logic Specifications and Omega-Regular Languages}.
\newblock \bibinfo{journal}{\emph{Theor. Comput. Sci.}} \bibinfo{volume}{195},
  \bibinfo{number}{2} (\bibinfo{year}{1998}), \bibinfo{pages}{183--203}.
\newblock
\urldef\tempurl \url{https://doi.org/10.1016/S0304-3975(97)00219-3}
\showDOI{\tempurl}


\bibitem[Peters and Yoshida(2024)]{DBLP:conf/lics/PetersY24}
\bibfield{author}{\bibinfo{person}{Kirstin Peters} {and}
  \bibinfo{person}{Nobuko Yoshida}.} \bibinfo{year}{2024}\natexlab{}.
\newblock \showarticletitle{Separation and Encodability in Mixed Choice
  Multiparty Sessions}. In \bibinfo{booktitle}{\emph{Proceedings of the 39th
  Annual {ACM/IEEE} Symposium on Logic in Computer Science, {LICS} 2024,
  Tallinn, Estonia, July 8-11, 2024}}, \bibfield{editor}{\bibinfo{person}{Pawel
  Sobocinski}, \bibinfo{person}{Ugo~Dal Lago}, {and} \bibinfo{person}{Javier
  Esparza}} (Eds.). \bibinfo{publisher}{{ACM}}, \bibinfo{pages}{62:1--62:15}.
\newblock
\urldef\tempurl \url{https://doi.org/10.1145/3661814.3662085}
\showDOI{\tempurl}


\bibitem[Roychoudhury et~al\mbox{.}(2012)]{DBLP:journals/tosem/RoychoudhuryGS12}
\bibfield{author}{\bibinfo{person}{Abhik Roychoudhury}, \bibinfo{person}{Ankit
  Goel}, {and} \bibinfo{person}{Bikram Sengupta}.}
  \bibinfo{year}{2012}\natexlab{}.
\newblock \showarticletitle{Symbolic Message Sequence Charts}.
\newblock \bibinfo{journal}{\emph{{ACM} Trans. Softw. Eng. Methodol.}}
  \bibinfo{volume}{21}, \bibinfo{number}{2} (\bibinfo{year}{2012}),
  \bibinfo{pages}{12:1--12:44}.
\newblock
\urldef\tempurl \url{https://doi.org/10.1145/2089116.2089122}
\showDOI{\tempurl}


\bibitem[Sakarovitch(1987)]{DBLP:journals/tcs/Sakarovitch87}
\bibfield{author}{\bibinfo{person}{Jacques Sakarovitch}.}
  \bibinfo{year}{1987}\natexlab{}.
\newblock \showarticletitle{On Regular Trace Languages}.
\newblock \bibinfo{journal}{\emph{Theor. Comput. Sci.}}  \bibinfo{volume}{52}
  (\bibinfo{year}{1987}), \bibinfo{pages}{59--75}.
\newblock
\urldef\tempurl \url{https://doi.org/10.1016/0304-3975(87)90080-6}
\showDOI{\tempurl}


\bibitem[Sakarovitch(1992)]{DBLP:conf/latin/Sakarovitch92}
\bibfield{author}{\bibinfo{person}{Jacques Sakarovitch}.}
  \bibinfo{year}{1992}\natexlab{}.
\newblock \showarticletitle{The "Last" Decision Problem for Rational Trace
  Languages}. In \bibinfo{booktitle}{\emph{{LATIN} '92, 1st Latin American
  Symposium on Theoretical Informatics, S{\~{a}}o Paulo, Brazil, April 6-10,
  1992, Proceedings}} \emph{(\bibinfo{series}{Lecture Notes in Computer
  Science}, Vol.~\bibinfo{volume}{583})},
  \bibfield{editor}{\bibinfo{person}{Imre Simon}} (Ed.).
  \bibinfo{publisher}{Springer}, \bibinfo{pages}{460--473}.
\newblock
\urldef\tempurl \url{https://doi.org/10.1007/BFB0023848}
\showDOI{\tempurl}


\bibitem[Scalas and Yoshida(2019)]{DBLP:journals/pacmpl/ScalasY19}
\bibfield{author}{\bibinfo{person}{Alceste Scalas} {and}
  \bibinfo{person}{Nobuko Yoshida}.} \bibinfo{year}{2019}\natexlab{}.
\newblock \showarticletitle{Less is more: multiparty session types revisited}.
\newblock \bibinfo{journal}{\emph{Proc. {ACM} Program. Lang.}}
  \bibinfo{volume}{3}, \bibinfo{number}{{POPL}} (\bibinfo{year}{2019}),
  \bibinfo{pages}{30:1--30:29}.
\newblock
\urldef\tempurl \url{https://doi.org/10.1145/3290343}
\showDOI{\tempurl}


\bibitem[Shen et~al\mbox{.}(2023)]{DBLP:journals/corr/abs-2303-00924}
\bibfield{author}{\bibinfo{person}{Gan Shen}, \bibinfo{person}{Shun Kashiwa},
  {and} \bibinfo{person}{Lindsey Kuper}.} \bibinfo{year}{2023}\natexlab{}.
\newblock \showarticletitle{HasChor: Functional Choreographic Programming for
  All (Functional Pearl)}.
\newblock \bibinfo{journal}{\emph{CoRR}}  \bibinfo{volume}{abs/2303.00924}
  (\bibinfo{year}{2023}).
\newblock
\urldef\tempurl \url{https://doi.org/10.48550/ARXIV.2303.00924}
\showDOI{\tempurl}
\showeprint[arXiv]{2303.00924}


\bibitem[Stutz(2023)]{DBLP:conf/ecoop/Stutz23}
\bibfield{author}{\bibinfo{person}{Felix Stutz}.}
  \bibinfo{year}{2023}\natexlab{}.
\newblock \showarticletitle{Asynchronous Multiparty Session Type
  Implementability is Decidable - Lessons Learned from Message Sequence
  Charts}. In \bibinfo{booktitle}{\emph{37th European Conference on
  Object-Oriented Programming, {ECOOP} 2023, July 17-21, 2023, Seattle,
  Washington, United States}} \emph{(\bibinfo{series}{LIPIcs},
  Vol.~\bibinfo{volume}{263})}, \bibfield{editor}{\bibinfo{person}{Karim Ali}
  {and} \bibinfo{person}{Guido Salvaneschi}} (Eds.).
  \bibinfo{publisher}{Schloss Dagstuhl - Leibniz-Zentrum f{\"{u}}r Informatik},
  \bibinfo{pages}{32:1--32:31}.
\newblock
\urldef\tempurl \url{https://doi.org/10.4230/LIPICS.ECOOP.2023.32}
\showDOI{\tempurl}


\bibitem[Stutz(2024)]{DBLP:phd/dnb/Stutz24}
\bibfield{author}{\bibinfo{person}{Felix Stutz}.}
  \bibinfo{year}{2024}\natexlab{}.
\newblock \emph{\bibinfo{title}{Implementability of Asynchronous Communication
  Protocols - The Power of Choice}}.
\newblock \bibinfo{thesistype}{Ph.\,D. Dissertation}.
  \bibinfo{school}{Kaiserslautern University of Technology, Germany}.
\newblock
\urldef\tempurl \url{https://kluedo.ub.rptu.de/frontdoor/index/index/docId/8077}
\showURL{\tempurl}


\bibitem[Stutz and D'Osualdo(2025)]{DBLP:conf/esop/StutzD25}
\bibfield{author}{\bibinfo{person}{Felix Stutz} {and} \bibinfo{person}{Emanuele
  D'Osualdo}.} \bibinfo{year}{2025}\natexlab{}.
\newblock \showarticletitle{An Automata-theoretic Basis for Specification and
  Type Checking of Multiparty Protocols}. In
  \bibinfo{booktitle}{\emph{Programming Languages and Systems - 34th European
  Symposium on Programming, {ESOP} 2025, Held as Part of the International
  Joint Conferences on Theory and Practice of Software, {ETAPS} 2025, Hamilton,
  ON, Canada, May 3-8, 2025, Proceedings, Part {II}}}
  \emph{(\bibinfo{series}{Lecture Notes in Computer Science},
  Vol.~\bibinfo{volume}{15695})}, \bibfield{editor}{\bibinfo{person}{Viktor
  Vafeiadis}} (Ed.). \bibinfo{publisher}{Springer}, \bibinfo{pages}{314--346}.
\newblock
\urldef\tempurl \url{https://doi.org/10.1007/978-3-031-91121-7\_13}
\showDOI{\tempurl}


\bibitem[Tirore et~al\mbox{.}(2023)]{DBLP:conf/itp/TiroreBC23}
\bibfield{author}{\bibinfo{person}{Dawit~Legesse Tirore},
  \bibinfo{person}{Jesper Bengtson}, {and} \bibinfo{person}{Marco Carbone}.}
  \bibinfo{year}{2023}\natexlab{}.
\newblock \showarticletitle{A Sound and Complete Projection for Global Types}.
  In \bibinfo{booktitle}{\emph{14th International Conference on Interactive
  Theorem Proving, {ITP} 2023, July 31 to August 4, 2023, Bia{\l}ystok,
  Poland}} \emph{(\bibinfo{series}{LIPIcs}, Vol.~\bibinfo{volume}{268})},
  \bibfield{editor}{\bibinfo{person}{Adam Naumowicz} {and}
  \bibinfo{person}{Ren{\'{e}} Thiemann}} (Eds.). \bibinfo{publisher}{Schloss
  Dagstuhl - Leibniz-Zentrum f{\"{u}}r Informatik},
  \bibinfo{pages}{28:1--28:19}.
\newblock
\urldef\tempurl \url{https://doi.org/10.4230/LIPICS.ITP.2023.28}
\showDOI{\tempurl}


\bibitem[Toninho and Yoshida(2017)]{DBLP:journals/jlp/ToninhoY17}
\bibfield{author}{\bibinfo{person}{Bernardo Toninho} {and}
  \bibinfo{person}{Nobuko Yoshida}.} \bibinfo{year}{2017}\natexlab{}.
\newblock \showarticletitle{Certifying data in multiparty session types}.
\newblock \bibinfo{journal}{\emph{J. Log. Algebraic Methods Program.}}
  \bibinfo{volume}{90} (\bibinfo{year}{2017}), \bibinfo{pages}{61--83}.
\newblock
\urldef\tempurl \url{https://doi.org/10.1016/J.JLAMP.2016.11.005}
\showDOI{\tempurl}


\bibitem[Udomsrirungruang and Yoshida(2025)]{DBLP:journals/pacmpl/UdomsrirungruangY25}
\bibfield{author}{\bibinfo{person}{Thien Udomsrirungruang} {and}
  \bibinfo{person}{Nobuko Yoshida}.} \bibinfo{year}{2025}\natexlab{}.
\newblock \showarticletitle{Top-Down or Bottom-Up? Complexity Analyses of
  Synchronous Multiparty Session Types}.
\newblock \bibinfo{journal}{\emph{Proc. {ACM} Program. Lang.}}
  \bibinfo{volume}{9}, \bibinfo{number}{{POPL}} (\bibinfo{year}{2025}),
  \bibinfo{pages}{1040--1071}.
\newblock
\urldef\tempurl \url{https://doi.org/10.1145/3704872}
\showDOI{\tempurl}


\bibitem[Union(1996)]{z120-standard}
\bibfield{author}{\bibinfo{person}{International~Telecommunication Union}.}
  \bibinfo{year}{1996}\natexlab{}.
\newblock \bibinfo{booktitle}{\emph{Z.120: Message Sequence Chart}}.
\newblock \bibinfo{type}{{T}echnical {R}eport}.
  \bibinfo{institution}{International Telecommunication Union}.
\newblock
\urldef\tempurl \url{https://www.itu.int/rec/T-REC-Z.120}
\showURL{\tempurl}


\bibitem[Vassor and Yoshida(2024)]{DBLP:conf/ecoop/VassorY24}
\bibfield{author}{\bibinfo{person}{Martin Vassor} {and} \bibinfo{person}{Nobuko
  Yoshida}.} \bibinfo{year}{2024}\natexlab{}.
\newblock \showarticletitle{Refinements for Multiparty Message-Passing
  Protocols: Specification-Agnostic Theory and Implementation}. In
  \bibinfo{booktitle}{\emph{38th European Conference on Object-Oriented
  Programming, {ECOOP} 2024, September 16-20, 2024, Vienna, Austria}}
  \emph{(\bibinfo{series}{LIPIcs}, Vol.~\bibinfo{volume}{313})},
  \bibfield{editor}{\bibinfo{person}{Jonathan Aldrich} {and}
  \bibinfo{person}{Guido Salvaneschi}} (Eds.). \bibinfo{publisher}{Schloss
  Dagstuhl - Leibniz-Zentrum f{\"{u}}r Informatik},
  \bibinfo{pages}{41:1--41:29}.
\newblock
\urldef\tempurl \url{https://doi.org/10.4230/LIPICS.ECOOP.2024.41}
\showDOI{\tempurl}


\bibitem[Yoshida(2024)]{DBLP:books/sp/24/Yoshida24}
\bibfield{author}{\bibinfo{person}{Nobuko Yoshida}.}
  \bibinfo{year}{2024}\natexlab{}.
\newblock \showarticletitle{Programming Language Implementations with
  Multiparty Session Types}.
\newblock In \bibinfo{booktitle}{\emph{Active Object Languages: Current
  Research Trends}}, \bibfield{editor}{\bibinfo{person}{Frank~S. de~Boer},
  \bibinfo{person}{Ferruccio Damiani}, \bibinfo{person}{Reiner H{\"{a}}hnle},
  \bibinfo{person}{Einar~Broch Johnsen}, {and} \bibinfo{person}{Eduard
  Kamburjan}} (Eds.). \bibinfo{series}{Lecture Notes in Computer Science},
  Vol.~\bibinfo{volume}{14360}. \bibinfo{publisher}{Springer},
  \bibinfo{pages}{147--165}.
\newblock
\urldef\tempurl \url{https://doi.org/10.1007/978-3-031-51060-1\_6}
\showDOI{\tempurl}


\bibitem[Yoshida and Gheri(2020)]{DBLP:conf/icdcit/YoshidaG20}
\bibfield{author}{\bibinfo{person}{Nobuko Yoshida} {and}
  \bibinfo{person}{Lorenzo Gheri}.} \bibinfo{year}{2020}\natexlab{}.
\newblock \showarticletitle{A Very Gentle Introduction to Multiparty Session
  Types}. In \bibinfo{booktitle}{\emph{Distributed Computing and Internet
  Technology - 16th International Conference, {ICDCIT} 2020, Bhubaneswar,
  India, January 9-12, 2020, Proceedings}} \emph{(\bibinfo{series}{Lecture
  Notes in Computer Science}, Vol.~\bibinfo{volume}{11969})},
  \bibfield{editor}{\bibinfo{person}{Dang~Van Hung} {and}
  \bibinfo{person}{Meenakshi D'Souza}} (Eds.). \bibinfo{publisher}{Springer},
  \bibinfo{pages}{73--93}.
\newblock
\urldef\tempurl \url{https://doi.org/10.1007/978-3-030-36987-3\_5}
\showDOI{\tempurl}


\bibitem[Yoshida and Hou(2024)]{DBLP:series/lncs/YoshidaH24}
\bibfield{author}{\bibinfo{person}{Nobuko Yoshida} {and} \bibinfo{person}{Ping
  Hou}.} \bibinfo{year}{2024}\natexlab{}.
\newblock \showarticletitle{Less is More Revisited: Association with Global
  Multiparty Session Types}.
\newblock In \bibinfo{booktitle}{\emph{The Practice of Formal Methods: Essays
  in Honour of Cliff Jones, Part {II}}}, \bibfield{editor}{\bibinfo{person}{Ana
  Cavalcanti} {and} \bibinfo{person}{James Baxter}} (Eds.).
  \bibinfo{series}{Lecture Notes in Computer Science},
  Vol.~\bibinfo{volume}{14781}. \bibinfo{publisher}{Springer},
  \bibinfo{pages}{268--291}.
\newblock
\urldef\tempurl \url{https://doi.org/10.1007/978-3-031-66673-5\_14}
\showDOI{\tempurl}


\bibitem[Zhou et~al\mbox{.}(2020)]{DBLP:journals/pacmpl/00020HNY20}
\bibfield{author}{\bibinfo{person}{Fangyi Zhou}, \bibinfo{person}{Francisco
  Ferreira}, \bibinfo{person}{Raymond Hu}, \bibinfo{person}{Rumyana Neykova},
  {and} \bibinfo{person}{Nobuko Yoshida}.} \bibinfo{year}{2020}\natexlab{}.
\newblock \showarticletitle{Statically verified refinements for multiparty
  protocols}.
\newblock \bibinfo{journal}{\emph{Proc. {ACM} Program. Lang.}}
  \bibinfo{volume}{4}, \bibinfo{number}{{OOPSLA}} (\bibinfo{year}{2020}),
  \bibinfo{pages}{148:1--148:30}.
\newblock
\urldef\tempurl \url{https://doi.org/10.1145/3428216}
\showDOI{\tempurl}


\bibitem[Zielonka(1987)]{DBLP:journals/ita/Zielonka87}
\bibfield{author}{\bibinfo{person}{Wieslaw Zielonka}.}
  \bibinfo{year}{1987}\natexlab{}.
\newblock \showarticletitle{Notes on Finite Asynchronous Automata}.
\newblock \bibinfo{journal}{\emph{{RAIRO} Theor. Informatics Appl.}}
  \bibinfo{volume}{21}, \bibinfo{number}{2} (\bibinfo{year}{1987}),
  \bibinfo{pages}{99--135}.
\newblock
\urldef\tempurl \url{https://doi.org/10.1051/ITA/1987210200991}
\showDOI{\tempurl}


\end{thebibliography}
\end{document}